%% file: main.tex
\newcommand\ThesisTitle{%
	Discrete geometric analysis of message passing algorithm on graphs}
\newcommand\ThesisAuthor{%
	Yusuke\ Watanabe }
\newcommand\ThesisProgramme{%
	DOCTOR\ OF\ PHILOSOPHY}
\newcommand\ThesisYear{March, 2010}
\documentclass[11pt,a4paper,twoside,final]{report} 
\usepackage{suthesis-2e}

\usepackage{graphicx}
\usepackage{xspace} %
\usepackage{amsmath}
\usepackage{amsfonts}
\usepackage{amssymb}
\usepackage{amsbsy} %
\usepackage{amsthm}
\usepackage{eufrak}
\usepackage{booktabs} %

\newtheorem{thm}{Theorem}[chapter]   %
\newtheorem{prop}{Proposition}[chapter]
\newtheorem{lem}{Lemma}[chapter]
\newtheorem{cor}{Corollary}[chapter]
\theoremstyle{definition} \newtheorem{defn}{Definition}
\theoremstyle{definition} \newtheorem{asm}{Assumption}
\theoremstyle{plain} \newtheorem*{claim}{Claim}
\theoremstyle{plain} \newtheorem*{remark}{Remark}
\theoremstyle{definition} \newtheorem{example}{Example}

\DeclareMathOperator*{\argmax}{argmax}
\DeclareMathOperator*{\argmin}{argmin}
\DeclareMathOperator*{\sgn}{sgn}

\def\bs#1{\boldsymbol#1}
\def\iunit{i}                          %

\newcommand{\sfrac}[2]{\frac{{\scriptstyle #1}}{{\scriptstyle #2}}}   %
\newcommand{\mat}[2]{{\rm M}(#1,#2)}                                  %
\newcommand{\pd}[2]{\frac{\partial #1}{\partial #2}}                  %
\newcommand{\pds}[3]{\frac{\partial^2 #1}{\partial #2 \partial #3}}   %
\newcommand{\pdseta}[3]{\frac{\partial^2 #1}{\partial \eta_{#2} \partial \eta_{#3}}} %
\newcommand{\norm}[1]{\| #1 \|}
\newcommand{\spec}[1]{{\rm Spec}(#1)}                 %
\newcommand{\specr}[1]{\rho(#1)}                      %
\newcommand{\interior}[1]{{\rm int}#1}
\newcommand{\E}[2]{{\rm E}_{#1}[#2]}                  %
\newcommand{\cov}[3]{{\rm Cov}_{#1}[#2,#3]}
\newcommand{\corr}[3]{{\rm Cor}_{#1}[#2,#3]}           %
\newcommand{\var}[2]{{\rm Var}_{#1}[#2]}

\newcommand{\inp}[2]{\langle#1,#2\rangle}             %

\newcommand{\core}{{\rm core}}
\newcommand{\cg}[1]{\hat{#1}}                         %
\newcommand{\diag}{{\rm diag}}

\newcommand{\ug}[1]{{\rm G}_{#1}}                     %

\newcommand{\e}{{\rm e}}                     %
\newcommand{\di}{{\rm d}}                    %
\newcommand{\tr}{{\rm tr}}                   %
 
\newcommand{\beliefs}{ \{b_{\alpha}(x_{\alpha}), b_i(x_i) \}_{\alpha \in F, i \in V}  }  %
\newcommand{\beliefsw}{ \{b_{\alpha}(x_{\alpha}), b_i(x_i) \}  }  %
\newcommand{\thetasw}{ \{\fa{\theta}, \theta_i \}  }                                 %

\newcommand{\fgdefn}{H=(V \cup F,\vec{E})}   %
\newcommand{\match}{{\bf D}}  %

\newcommand{\etea}{e' \rightharpoonup e}

\newcommand{\ete}[2]{#1 \rightharpoonup #2}
\newcommand{\vfe}{\mathfrak{X}(\vec{E})}     %
\newcommand{\vfv}{\mathfrak{X}(V)}           %
\newcommand{\sfe}{\mathfrak{F}(\vec{E})}     %
\newcommand{\sfv}{\mathfrak{F}(V)}           %
\newcommand{\matmu}{\mathcal{M}(\bs{u})}     %
\newcommand{\matmc}{\mathcal{M}(\bs{c})}     %

\newcommand{\ed}[2]{#1 \rightarrow #2}       %
\newcommand{\edai}{\alpha \rightarrow i}
\newcommand{\edaj}{\alpha \rightarrow j}
\newcommand{\edbi}{\beta \rightarrow i}
\newcommand{\edbj}{\beta \rightarrow j}
\newcommand{\edgi}{\gamma \rightarrow i}

\newcommand{\edbione}{\beta \rightarrow i_{1}}
\newcommand{\edbik}{\beta \rightarrow i_{k}}
\newcommand{\edij}{i \rightarrow j}
\newcommand{\edji}{j \rightarrow i}
\newcommand{\edik}{i \rightarrow k}
\newcommand{\edjk}{j \rightarrow k}

\newcommand{\edkl}{k \rightarrow l}

\newcommand{\fai}{\{i_{1},\ldots,i_{d_{\alpha}}\}} %
\newcommand{\prodv}{\prod_{i \in V}}
\newcommand{\prodf}{\prod_{\alpha \in F}}

\newcommand{\Hesse}{\nabla^{2}}    %

\newcommand{\bsb}{\bs{b}}
\newcommand{\bsu}{\bs{u}}

\newcommand{\bsw}{\bs{w}}
\newcommand{\bsx}{\bs{x}}
\newcommand{\bsh}{\bs{h}}
\newcommand{\bsphi}{\bs{\phi}}
\newcommand{\bstheta}{\bs{\theta}}  
\newcommand{\bseta}{\bs{\eta}}
\newcommand{\bsmu}{\bs{\mu}}
\newcommand{\bsbeta}{\bs{\beta}}
\newcommand{\bsgamma}{\bs{\gamma}}

\newcommand\Bfe{Bethe free energy\xspace}
\newcommand\LBP{Loopy Belief Propagation\xspace}
\newcommand\lbp{loopy belief propagation\xspace}
\newcommand\BP{Belief Propagation\xspace}
\newcommand\bp{belief propagation\xspace}
\newcommand\ifa{inference family\xspace}
\newcommand\ifas{inference families\xspace}
\newcommand\Ifa{Inference family\xspace}
\newcommand\epara{expectation parameter\xspace}
\newcommand\eparas{expectation parameters\xspace}
\newcommand\npara{natural parameter\xspace}
\newcommand\nparas{natural parameters\xspace}

\newcommand\ls{loop series\xspace}
\newcommand\Bzf{Bethe-zeta formula\xspace}
\newcommand\IB{Ihara-Bass\xspace}
\newcommand\ccm{correlation coefficient matrix\xspace}
\newcommand\ccms{correlation coefficient matrices\xspace}
\newcommand\dcr{deletion-contraction relation\xspace}
\newcommand\dcrs{deletion-contraction relations\xspace}

\newcommand{\fa}[1]{{{#1}_{\alpha}}}     %
      
\newcommand{\fx}[2]{{{#1}_{#2}}}      
\newcommand{\pa}[1]{ {{#1}_{{\scriptscriptstyle \langle}\! \! \;  \alpha \! \! \; {\scriptscriptstyle \rangle} } } }   %
\newcommand{\pb}[1]{ {{#1}_{{\scriptscriptstyle\langle} \! \! \; \beta \! \! \; {\scriptscriptstyle \rangle} } } } 
\newcommand{\px}[2]{{{#1}_{ {\scriptscriptstyle\langle}\! \! \;  #2 \! \! \; {\scriptscriptstyle \rangle} }}}
\newcommand{\pxw}[2]{{{#1}_{ {\scriptscriptstyle\langle}  #2 {\scriptscriptstyle \rangle} }}}
\newcommand{\va}[2]{{#1}_{\alpha : #2}}    
\newcommand{\vb}[2]{{#1}_{\beta : #2}} 
\newcommand{\vx}[3]{{#1}_{#2 : #3}} 

\makeatletter
\renewcommand\maketitle{
\begin{titlepage}
\begin{center}
\hrule height 4pt
\vspace{9mm}
{\bf {\huge \@title{} }} \\
\vspace{10mm}
\hrule height 1pt
\vspace{30mm}
{\bf {\Large \@author{}} }\\

\vspace{18mm}
{\Large \ThesisProgramme}\\
{ Department of Statistical Science \\
School of Multidisciplinary Science\\
The Graduate University for Advanced Studies}\\

\vspace{15mm}
\ThesisYear
\end{center}
\end{titlepage}
}
\makeatother

\title{\ThesisTitle}
\author{\ThesisAuthor}

\begin{document}
\maketitle

\prefacesection{Abstract} 
We often encounter probability distributions given as unnormalized products of non-negative functions.
The factorization structures of the probability distributions are represented by hypergraphs
called factor graphs. 
Such distributions appear in various fields, including
statistics, artificial intelligence, statistical physics, error correcting codes, etc.

Given such a distribution,
computations of marginal distributions and the normalization constant are often required.
However, they are computationally intractable because of their computational costs.
One, empirically successful and tractable, approximation method is the Loopy Belief Propagation (LBP) algorithm.

The focus of this thesis is an analysis of the LBP algorithm.
If the factor graph is a tree, i.e. having no cycle, the algorithm gives the exact quantities, not approximations.
If the factor graph has cycles, however, the LBP algorithm does not give exact results and possibly exhibits oscillatory and non convergent
behaviors. 
The thematic question of this thesis is 
``How do the behaviors of the LBP algorithm are affected by the discrete geometry of the factor graph?''
Here, the word ``discrete geometry'' means the geometry of the factor graph as a space.  

The primary contribution of this thesis is the discovery of a formula called the \Bzf,
which establishes the relation between the LBP, the \Bfe and the graph zeta function.
This formula provides new techniques for analysis of the LBP algorithm, connecting properties of the graph and
of the LBP and the \Bfe.
We demonstrate applications of the techniques to several problems including
(non) convexity of the \Bfe, the uniqueness and stability of the LBP fixed point. 

We also discuss the \ls initiated by Chertkov and Chernyak (2006).
The \ls is a subgraph expansion of the normalization constant, or partition function, and reflects the graph geometry.
We investigate theoretical natures of the series.  
Moreover, we show a partial connection between the \ls and the graph zeta function.

\prefacesection{Acknowledgments}
I would like to express my sincere thanks to my adviser, Prof. Kenji Fukumizu, for his generous time and commitment. 
He placed his trust in me and allowed me the freedom to find my own way. 
At the same time, he made time for discussion on my research every week and gave me excellent feedbacks.
I also thank to his helps to improve my writings.

It is a pleasure for me to thank all people in the Institute of Statistical Mathematics
who helped my research during my Ph.D. course.
Especially, I would like to thank Prof. Shiro Ikeda for his careful reading of the manuscript and helpful comments.

I would like to thank Prof. Michael Chertkov and Dr. Jason Johnson for their hospitality during my stay in Los Alamos National Laboratory
and for things they taught me about the belief propagation algorithm.

Last but not least,
I would like to thank my parents, Kenjiro and Kanako, for their unwavering support.

This research was financially supported by Grant-in-Aid for JSPS Fellows
20-993 and Grant-in-Aid for Scientific Research (C) 19500249.

\chapter*{Notational Conventions}
\begin{table}[th]
\renewcommand{\arraystretch}{1.2}  %
\begin{tabular}{p{40mm}p{100mm}}
\toprule
{\bf General Notation} \\
\midrule
\mat{$r_1$}{$r_2$}               & set of $r_1 \times r_2$ matrices  \\
$\inp{x}{y}$                     & inner product of vectors $x$ and $y$: $\sum x_i y_i$  \\
$\norm{\cdot}$                   & norm                                  \\
$\spec{X}$                       & set of eigenvalues of matrix $X$  \\                          
$\specr{X}$                      & spectral radius of $X$                \\
$\diag (x)$                      & diagonal matrix with diagonal elements $x_i$  \\
$\Hesse$                         & Hessian operator \\
$\interior{\Theta}$              & interior of a set $\Theta$\\
$\sgn (x) $                      & sign of a real value $x$\\
$\E{p}{\phi}$                    & expectation of $\phi(x)$ under $p$ \\
$\cov{p}{\phi}{\phi'}$           & covariance of $\phi(x)$ and $\phi'(x)$ under $p(x)$ \\
$\corr{p}{\phi}{\phi'}$          & correlation of $\phi(x)$ and $\phi'(x)$ under $p(x)$ \\
$\var{p}{\phi}$                  & variance $\phi(x)$ and under $p(x)$\\
\bottomrule
\end{tabular}
\end{table}
\begin{table}[h]
\renewcommand{\arraystretch}{1.2}  %
\begin{tabular}{p{40mm}p{100mm}}
\toprule
{\bf Graphs} \\
\midrule
$G$                              & undirected graph $G=(V,E)$\\
$V$                              & vertex set \\
$E$                              & edge set                                  \\
$H$                              & hypergraph $H=(V,F)$ \\                          
$F$                              & hyperedge set                              \\
$N_i$                            & neighborhood of vertex $i$.                            \\
$N_{\alpha}$                     & neighborhood of hyperedge $\alpha$.                            \\
$d_i$                            & degree of vertex $i$.                            \\
$d_{\alpha}$                     & degree of hyperedge $\alpha$.                            \\
$\core (H)$                      & core of hypergraph $H$  \\
$e$                              & directed edge (undirected edge in chapter 7) \\
\bottomrule
\end{tabular}
\end{table}
\begin{table}[th]
\renewcommand{\arraystretch}{1.2}  %
\begin{tabular}{p{40mm}p{100mm}}
\toprule
{\bf Graphical models} \\
\midrule
$\Psi=\{\Psi_{\alpha}\}$         & set of compatibility functions, graphical model \\
$x_i$                            & random variable on $i \in V$\\
$\mathcal{X}_i$                  & value set of $x_i$\\
$\mathcal{X}$                    & $\prod_{i \in V} \mathcal{X}_i$ \\
$Z$                              & partition function, normalization constant \\
\bottomrule
\end{tabular}
\end{table}
\begin{table}[h]
\renewcommand{\arraystretch}{1.2}  %
\begin{tabular}{p{40mm}p{100mm}}
\toprule
{\bf Exponential families} \\
\midrule
$\mathcal{E}$                    & exponential family \\
$\phi(x)$                        & sufficient statistics \\
$\psi$                           & log partition function \\
$\varphi$                        & negative entropy \\
$\Lambda$                        & Legendre map        \\
$\Theta$                         & set of \nparas  \\
$Y$                              & set of \eparas \\
\bottomrule
\end{tabular}
\end{table}
\begin{table}[th]
\renewcommand{\arraystretch}{1.2}  %
\begin{tabular}{l@{\hspace{4mm}}l}
\toprule
{\bf Inference family, LBP} \\
\midrule
$\mathcal{I}=\{\mathcal{E}_{\alpha},\mathcal{E}_i \}_{\alpha \in F, i \in V}$                    
& \ifa \\
$\fa{\phi}(x_{\alpha})=(\pa{\phi}(x_{\alpha}),\phi_{i_1}(x_{i_1}),\ldots,\phi_{i_{d_{\alpha}}}(x_{i_{d_{\alpha}}}) )$      
&sufficient statistics of $\mathcal{E}_{\alpha}$ \\
$ \fa{\theta}=(\pa{\theta},\va{\theta}{i_1},\ldots,\va{\theta}{i_{d_{\alpha}}} )$
&\nparas of $\mathcal{E}_{\alpha}$ \\
$\fa{\eta}=(\pa{\eta},\va{\eta}{i_1},\ldots,\va{\eta}{i_{d_{\alpha}}})$
&\eparas of $\mathcal{E}_{\alpha}$ \\
$\Theta_{\alpha}$
&set of \nparas of $\mathcal{E}_{\alpha}$. $\fa{\theta} \in \Theta_{\alpha}$ \\
$Y_{\alpha}$
&set of \eparas of $\mathcal{E}_{\alpha}$. $\fa{\eta} \in Y_{\alpha}$ \\
$\mathcal{E}(\mathcal{I})$       & global exponential family \\
$F$                              & type 1 \Bfe function \\
$\mathcal{F}$                    & type 2 \Bfe function \\
$L(\mathcal{I})$                 & domain of type 1 BFE function \\
$A(\mathcal{I},\Psi)$            & domain of type 2 BFE function \\
$m_{\edai}$                      & message from $\alpha$ to $i$ \\
$\mu_{\edai}$                    & natural parameter of $m_{\edai}$ \\
\bottomrule
\end{tabular}
\end{table}
\begin{table}[h]
\renewcommand{\arraystretch}{1.2}  %
\begin{tabular}{p{40mm}p{108.2mm}}
\toprule
{\bf Graph zeta} \\
\midrule
$\vec{E}$                        & set of directed edges \\
$s(e)$                           & starting factor of directed edge $e$\\
$t(e)$                           & terminus vertex of directed edge $e$\\
$\ete{e}{e'}$                    & $t(e) \in s(e')$ and $t(e) \neq t(e')$\\ 
$\mathfrak{P}_H$                 & set of prime cycles of hypergraph $H$ \\
$\mathfrak{p}$                   & prime cycle\\
$r_e$                            & positive integer (dimension) associated with $e$ \\
$r_i$                            & positive integer (dimension) associated with $i \in V$ \\
$\vfe$                           & set of $\mathbb{C}^{r_e}$ valued functions on $\vec{E}$ \\
$\vfv$                           & set of $\mathbb{C}^{r_i}$ valued functions on ${V}$ \\
$\matmu$                         & linear operator on $\vfe$, defined in Eq.~(\ref{def:matmu}) \\   
$\bs{\iota}(\bs{u})$             & linear operator on $\vfe$, defined in Eq.~(\ref{def:iota}) \\ 
$U_{\alpha}$                     & diagonal block of $I+\bs{\iota}(\bs{u})$ \\   
$w^{\alpha}_{\edij}$             & $(j,i)$ element of $W_{\alpha}=U_{\alpha}^{-1}$ \\   
$\mathcal{D}, \mathcal{W}$       & linear operators on $\vfv$, defined in Eq.~(\ref{eq:defDW}) \\    
$\zeta_H$                        & zeta function of a hypergraph $H$\\
$Z_G$                            & zeta function of a graph $G$\\
\bottomrule
\end{tabular}
\end{table}
\newpage
\tableofcontents %

\chapter{Introduction}
%
\input{chapter1}

\chapter{Preliminaries}
\input{chapter2}

\part{Graph zeta in \Bfe and \lbp}
\chapter{Graph zeta function}\label{chap:zeta}
\input{chapter3}

\chapter{\Bzf}\label{chap:Bzf}
\input{chapter4}

\chapter{Uniqueness of LBP fixed point}\label{chap:unique}
\input{chapter5}

\part{Loop Series}
\chapter{Loop Series}
\input{chapter6}

\chapter{Graph polynomials from LS}
\input{chapter7a}

\input{chapter7b}

\chapter{Conclusion}
\input{chapter8}

\appendix
\chapter{Useful theorems}
\input{appendix1}
\bibliographystyle{plain}
\bibliography{BP,LS,staphys,zeta}

\end{document}
\endinput

%% file: chapter1.tex

This chapter provides a background of main topics of this thesis,
motivating our approach: discrete geometric analysis.
Formal definitions and problem settings are given in later chapters.
The first section gives a short introduction of graphical models, which is the main object discussed in this thesis,
explaining the important associated computational tasks, i.e., the computation of marginal distributions and the partition function.
The second section introduces an efficient and powerful approximation algorithm: Loopy Belief Propagation (LBP),
which is thoroughly analyzed in this thesis.
This section also explains the importance of considering the graph geometry for the analysis.
In this thesis, we refer to ``graph geometry'' as discrete geometry.
In the third section, we discuss the discrete geometry that should be considered in the context of LBP.
We first review that interplays between geometric spaces and objects on it, often have appeared in the history of mathematics.
We also review tools in graph theory devised for understanding graphs.
The final section is devoted to the description of the organization of this thesis as well as
a short summary of each chapter.

\section{Graphical models}

\subsection{Introduction of graphical models}
A {\it graphical model} consists of a set of random variables which has a dependence structure represented by a certain type of graph.
There are many classes of graphical models such as pairwise models and Bayesian networks.
Among them, factor graph models include wide classes of graphical models and express factorization structure that are needed for our purpose.

Here we give an example of a factor graph model.
Let $x=(x_1,x_2,x_3,x_4)$ be random variables
and assume that the probability density function of $x$ is given by a factorized form, e.g.,
\begin{equation}
p(x)=\frac{1}{Z}
\Psi_{123}(x_1,x_2,x_3) \Psi_{134}(x_1,x_3,x_4), \label{eq:factorgraphexample1}
\end{equation}
where $\Psi_{123}$ and $\Psi_{134}$ are non-negative functions called {\it compatibility functions} and $Z$ is the normalization constant.
This factorization structure is cleverly depicted by a factor graph;
the factor graph for this example is given by Figure \ref{factorgraphexample1}.
Each square represents a compatibility function and a circle represents a variable.
If a compatibility function has a variable as an argument, the corresponding square and the circle are joined by an edge.

Formally, in this thesis, a {\it graphical model} is referred to as a set of compatibility functions,
which defines the probability distribution by the product.
For general graphical models, the way of illustrating factor graphs is obvious, i.e.,
draw squares and circles corresponding to the compatibility functions and variables respectively,
and join them if a variable is an argument of a compatibility function.
Usually, the index sets of variables and compatibility functions are denoted by $V$ and $F$ respectively.

\begin{figure}
\begin{center}
\includegraphics[scale=0.25]{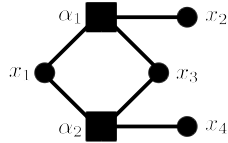} 
\caption{The factor graph associated with the factorization Eq.~(\ref{eq:factorgraphexample1})} \label{factorgraphexample1}
\end{center}
\end{figure}

Note that, if all the compatibility functions have two variables as the arguments,
the factor graph is more simply represented by an undirected graph $G=(V,E)$.
Indeed, we can replace each square and the two edges in the factor graph by a single edge
without loss of information.

\subsection{Important computational tasks}
When a graphical model on $x=(x_i)_{i \in V}$ is given by 
\begin{equation}
 p(x) \propto \prod_{\alpha \in F} \Psi_{\alpha}(x_{\alpha}), \qquad x_{\alpha}=(x_i)_{i \in \alpha}
\end{equation}
one sometimes needs to compute {\it marginal distributions} over relatively small subsets of the variables.
For example, the marginal distribution of $x_1$ is 
\begin{equation}
 p_1(x_1)= \sum_{(x_{i})_{i \in V \smallsetminus 1}} p(x).
\end{equation}
It is also important to compute the partition function, i.e. the normalization constant:
\begin{equation}
 Z= \sum_{x}  \prod_{\alpha \in F} \Psi_{\alpha}(x_{\alpha}).
\end{equation}
Several examples that need these computations are given in the next subsection.

Despite its importance, the computation of marginal distributions and the partition function are
often unfeasible tasks especially if the number of variables is large and the ranges of variables are discrete.
Assuming that the variables are binary,
one observes that the direct computation of each of these quantities requires $O(2^{|V|})$ sums. 
In fact, in the discrete variables cases, 
the problem of computing partition functions is NP-hard \cite{Bcomputational,Ccomputational}. 
Therefore, we need to develop approximation methods that give accurate results for graphical models appearing in real worlds.

It is noteworthy that the exact computation is sometimes feasible using devices for reducing computational cost.
A major approach is the {\it junction tree} algorithm \cite{CDLSprobabilistic}, 
which makes a tree from the associated graph.
This algorithm requires the computational cost exponential to the largest clique size of
the triangulated graph.
Rather than the junction tree algorithm, we analyze the LBP algorithm in this thesis.
One reason is that the largest size of cliques is often too big for running the junction tree algorithm in a practical time 
even if the LBP algorithm can be executed quickly. 
Another reason is more theoretical; we would like to capture graph geometry
as discrepancies between local computations and global computations.
Since the junction tree algorithm reduces to a tree, which has globally trivial geometry, 
the junction tree algorithm does not have such an aspect.

\subsection{Examples and applications of graphical models}
This subsection gives examples of graphical models and explains the importance
of the partition function and/or marginal distributions in each case.
Typically, graphical model is used to implement our knowledge of dependency among random variables.
One example comes from statistical physics.
Let us consider the following form of graphical model on $G=(V,E)$
called (disordered) {\it Ising model}, {\it Spin-glass model} or {\it binary pairwise model}:
\begin{equation}
  p(x)=\frac{1}{Z} \exp(\sum_{ij \in E} J_{ij} x_i x_j  + \sum_{i \in V} h_i x_i), \label{defIsing}
\end{equation}
where $x_i = \pm 1$.
This model abstracts behaviors of spins laid on vertices of the graph.
Each spin has two (up and down) states and only interacts with neighbors. 
Importance of the one-variable marginal distributions may be agreed because they describe probabilities of states of spins. 
However, importance of the partition function may be less obvious.
One reason comes from its logarithmic derivatives;
the expected values and correlations of the variables are given by the derivatives of the log partition function:
\begin{equation}
 \pd{\log Z}{J_{ij}}= \E{p}{x_i x_j}, \quad \pd{\log Z}{h_i}= \E{p}{x_i}. \label{derivativeZ}
\end{equation}
Important physical quantities such as energy, entropy etc are easily calculated by 
the log partition function, or equivalently the {\it Gibbs free energy} \cite{Nstatistical}.
Another example comes from error correcting codes.
From the original information, the sender generates a certain class of binary sequence called codeword
and transmit it thorough a noisy channel \cite{PWerror}.
If the number of errors is relatively small, the receiver can correct them using added redundancy.
The decoding process can be formulated as an inference problem of finding a plausible codeword.
In linear codes, a codeword is made to satisfy local constraints, i.e.,
the sum of certain subsets of bits is equal to zero in modulo two arithmetic.
Then the probabilities of codewords are given by a graphical model. 
The marginals can be used as an estimate of each bit.
LDPC codes and turbo codes are included in this type of algorithms  \cite{MCturbo,Mgood,GLDPC}.

We also find examples in the field of image processing.
In the super-resolution problem, one would like to infer a high resolution image
from multiple low resolution images \cite{BSsuper,BBsuper}.
The high resolution image can be interpreted as a graphical model imposing local constraints on pixels.
The marginal distributions of the model give the inferred image.
The compatibility functions are often learned from training images \cite{FPClowlevelvision,GRPHnonparametric}.
Another example is found in statistics and artificial intelligence.
Domain-specific knowledge of experts can be structured and quantified using graphical models.
One can perform inference from new information, processing the model.
Such a system is called {\it expert system} \cite{CDLSprobabilistic}.  
For example, medical knowledge and experiences can be interpreted as a graphical model.
If a person is a smoker, he or she is likely to have lung cancer or bronchitis compared to non smokers.
This empirical knowledge is represented by a compatibility function
having variables ``smoking,'' ``lung cancer'' and ``bronchitis.''
Moreover, medical experts have a lot of data on the relation between diseases, symptom and other information.
Utilizing such knowledge, statisticians can make a graphical model over the variables related to medical diagnosis, such as
``smoking,'' ``lung cancer,'' ``bronchitis,'' ``dyspnoea,'' ``cough'' etc.
If a new patient, who is coughing and non smoker, comes,
the probability of being bronchitic is the marginal probability of the graphical model with fixed observed variables.
This example is taken from a book by Cowell et al \cite{CDLSprobabilistic} and is called CH-ASIA.
Moreover, there are many general computational problems that are reduced to computations of the partition functions.
Indeed, the counting problems of perfect matchings, graph colorings and SAT are equivalent to 
evaluating the partition functions of certain class of graphical models. 
Computation of the permanent of a matrix is also translated into the partition function of a graphical model on a complete bipartite graph.
The partition function of the perfect matching problem will be discussed in Section \ref{sec:perfectmatching}.

\subsection{Approximation methods}
Because of the computational difficulty,
problem settings in the language of graphical models are useful only if such a formulation is combined with efficient algorithms.
In this subsection, we list approximation approaches except for the \lbp algorithm,
which is comprehensively discussed in the next section. 

\subsubsection{Mean field approximation}
One of the simplest approximation scheme is the (naive) {\it mean field approximation} \cite{Pstatistical}.
For simplicity, let us consider the binary pairwise model Eq.~(\ref{defIsing}) on a graph $G=(V,E)$.
Let $m_i=\E{}{x_i}$ be the mean.
We approximate the partition function by replacing the state of the nearest neighbor variables 
by its mean:
\begin{align*}
 Z 
&=\sum_{x} \exp(\sum_{ij \in E} J_{ij} (m_i+\delta x_i)(m_j+\delta x_j)  + \sum_{i \in V} h_i x_i) \\
&\approx \sum_{x} \exp(\sum_{ij \in E} J_{ij}m_i m_j+ \sum_{i \in V}\sum_{j \in N_i} J_{ij} m_i \delta x_j  + \sum_{i \in V} h_i x_i)  \label{MFapprox}\\
&=2^{|V|} \exp(-\sum_{ij \in E} J_{ij}m_i m_j) \prod_{i \in V}\cosh( \sum_{j \in N_i}J_{ij}m_j + h_i). 
\end{align*}
From $\E{}{x_i}=\pd{\log Z}{h_i}$, we obtain constraints called {\it self consistent equation}.
\begin{equation}
 m_i=    \tanh (  \sum_{j \in N_i}J_{ij}m_j + h_i). \label{eq:selfconsistent}
\end{equation}
The solution of this equation gives an approximation of the means.

This approximation is also formulated as a variational problem \cite{JGJSintroduction}. %
Let $p$ be the probability distribution in Eq.~(\ref{defIsing}) and let $q$ be a fully decoupled distribution
with means $m_i$:
\begin{equation}
 q(x)= \prod_{i \in V} \left( \frac{1+ m_i x_i}{2} \right). \label{eq:fullfactor}
\end{equation}
The variational problem is the minimization of the Kullback-Leibler divergence 
\begin{align*}
 D(q||p)
&= \sum_{x} q(x) \log \left( \frac{q(x)}{p(x)} \right) \\
&= \sum_{i \in V} \left[ \frac{1+m_i}{2} \log \left( \frac{1+m_i}{2} \right)
+  \frac{1-m_i}{2} \log \left( \frac{1-m_i}{2} \right) \right] \\
& \qquad - \sum_{ ij \in E} J_{ij} m_i m_j - \sum_{i \in V} h_i m_i
+ \log Z
\end{align*}
with respect to $q$.
One observes that the condition $\pd{}{m_i}D(q||p)=0$ is equivalent to Eq.~(\ref{eq:selfconsistent}).
Therefore, this variational problem is equivalent to the mean field approximation method.

Empirically, this approximation gives good results
especially for large and densely connected graphical models with relatively weak interactions  \cite{Pstatistical,Jtutorial}.
However, the full factorization assumption Eq.~(\ref{eq:fullfactor}) is often too strong to capture the structure of the true distribution,
yielding a poor approximation.
One approach for correction is the {\it structured mean field approximation} \cite{SJexploiting},
which extends the region of variation to a sub-tree structured distributions,
keeping computational tractability \cite{Wstochastic}.

\subsubsection{Randomized approximations}
Another lines of approximation is randomized (or Monte Carlo) methods.
For the computation of a marginal probability distribution, one can generate a stochastic process that
converges to the distribution \cite{GSsampling}.
The partition function can also computed by sampling.
For ferromagnetic (attractive) case ($J_{ij} \geq 0$), the partition function of the 
Ising model Eq.~(\ref{defIsing}) is accurately approximated in a polynomial time \cite{JSpolynomial}.
More precisely the algorithm is a {\it fully polynomial randomized approximation scheme} (FPRAS).
One major disadvantage of these methods is that these are often too slow for practical purposes \cite{JGJSintroduction}. %
In this thesis, we do not focus on such randomized approaches.

\section{Loopy belief propagation}
\subsection{Introduction to (loopy) belief propagation}
\label{sec:introLBPalgorithm}
Though the evaluation of marginal distributions and the partition function are 
intractable tasks in general, there is a tractable class of graph structure: tree.
A graph is called a tree if it is connected and does not contain cycles.
In 1982, Judea Pearl proposed an efficient algorithm for
calculation of marginal distributions on tree structured models, 
called {\it \BP} (BP) \cite{Preverend,Pearl}.
Roughly speaking, the \bp is a message passing algorithm, i.e.
a message vector is associated with each direction of an edge and
updated by local operations.
Since these local operations can be defined irrespective of the global graph structure,
BP algorithm is directly applicable to graphical models with cycles.
This method is called the {\it Loopy Belief Propagation} (LBP), showing empirically successful performance \cite{MWJempiricalstudy}.
Especially, the method is good for sparse graphs, which do not have short cycles. 
Here we simply explain operations of the (loopy) \bp algorithm.
First let us consider
a pairwise binary model Eq.~(\ref{defIsing}) on a tree in Figure \ref{fig:ExampleTree}. 
We write $\Psi_{ij}(x_i,x_j)= \exp(J_{ij} x_i x_j)$ and $\Psi_{i}(x_i)=\exp(h_i x_i)$.
Then a marginal distribution $p_2$ is given by
\begin{align}
 p_2(x_2)  \propto 
&\sum_{x_1,x_3,x_4,x_5} \Psi_{12} \Psi_{23} \Psi_{34} \Psi_{35} \Psi_1 \Psi_2 \Psi_3 \Psi_4 \Psi_5  \nonumber \\
=&\Psi_2
\left( \sum_{x_1} \Psi_{12} \Psi_1 \right)
\left(
\sum_{x_3} \Psi_{23} \Psi_{3}
\left( \sum_{x_4} \Psi_{34} \Psi_4 \right)
\left( \sum_{x_5} \Psi_{35} \Psi_5 \right) \label{eq:marginalp2}
\right).
 \end{align}
In the above equality, we used the commutativity and the distributive law of the sum and the product. 
If we define messages by
\begin{align*}
 m_{1 \rightarrow 2}(x_2):= \sum_{x_1} \Psi_{12}(x_1,x_2) \Psi_1(x_1) \\
 m_{4 \rightarrow 3}(x_3):= \sum_{x_4} \Psi_{34}(x_3,x_4) \Psi_4(x_4) \\
 m_{5 \rightarrow 3}(x_3):= \sum_{x_5} \Psi_{35}(x_3,x_5) \Psi_5(x_5) \\
\end{align*}
and 
\begin{equation*}
  m_{3 \rightarrow 2}(x_2):= \sum_{x_3} \Psi_{23}(x_2,x_3) \Psi_3(x_3)  m_{4 \rightarrow 3}(x_3) m_{5 \rightarrow 3}(x_3),
\end{equation*}
Eq.~(\ref{eq:marginalp2}) becomes
\begin{equation*}
 p_2(x_2) \propto \Psi_{2}(x_2)  m_{1 \rightarrow 2}(x_2)  m_{3 \rightarrow 2}(x_2).
\end{equation*}
\begin{figure}
\begin{center}
\includegraphics[scale=0.25]{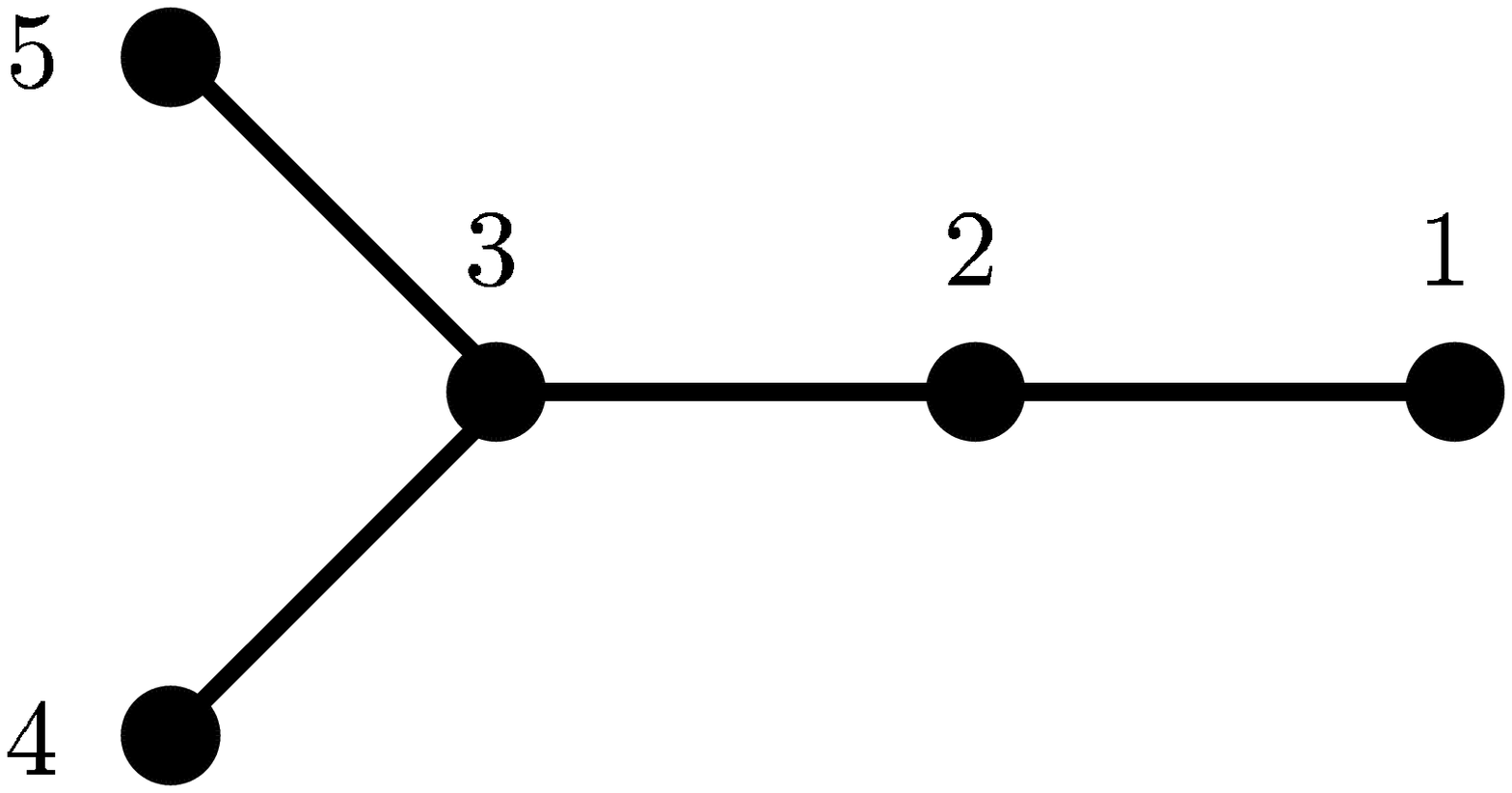} 
\caption{A tree.} \label{fig:ExampleTree}
\end{center}
\end{figure} 
The partition function is also computed using messages;
\begin{equation}
 Z= \sum_{x_2}  \Psi_{2}(x_2)  m_{1 \rightarrow 2}(x_2)  m_{3 \rightarrow 2}(x_2).
\end{equation}
Obviously, this method is applicable to arbitrary trees; it is called the belief propagation algorithm.
More formally, define messages for all directed edges $(j \rightarrow i)$ inductively by
\begin{equation}
 m_{j \rightarrow i}(x_i) = \sum_{x_j} \Psi_i(x_i) \Psi_{ij}(x_i,x_j) \prod_{k \in N_j \smallsetminus i} m_{k \rightarrow j}(x_j), \label{eq:bplbpupdate}\\
\end{equation}
where $N_j$ is the neighboring vertices of $j$.
Since we are considering a tree, this equation uniquely determines the messages. 
The marginal distribution of $x_i$ and the partition function are given by
\begin{align}
p_i(x_i) \propto \Psi_{i}(x_i) \prod_{j \in N_i} m_{j \rightarrow i}(x_i), \label{eq:bplbpmarginal} \\
Z = \sum_{x_i} \Psi_{i}(x_i) \prod_{j \in N_i} m_{j \rightarrow i}(x_i).
\end{align}
These computations requires only $O(|E|)$ steps. 
Therefore, the marginals and the partition function of a graphical model associated with a tree can be computed in practical time.

Secondly, let us consider the case that the underlying graph is not a tree.
In this case, Eq.~(\ref{eq:bplbpupdate}) does not determine the messages explicitly.
However, we can solve Eq.~(\ref{eq:bplbpupdate}) and obtain a set of messages as a solution.
Though this equation has possibly many solutions,
we take one solution that is obtained by iterative applications of Eq.~(\ref{eq:bplbpupdate}).
In other word, we use the equation as an update rule of the messages and find a fixed point.
Then, at a fixed point, the approximation of a marginal distribution is given by Eq.~(\ref{eq:bplbpmarginal}).
This method is called the loopy belief propagation.
The approximation for the partition function is slightly involved; we will explain it in the next subsection.

\subsection{Variational formulation of LBP}
At first sight, the \lbp looks groundless because it is just a diversion of the belief propagation, which is guaranteed to work only on trees.
However, Yedidia et al \cite{YFWGBP} have shown the equivalence to the Bethe approximation, making the algorithm on a concrete theoretical ground.
More precisely, the LBP algorithm is formulated as a variational problem of the Bethe free energy function.

Again, we explain the variational formulation in the case of the model Eq.~(\ref{defIsing}).
Let $b=\{b_{ij},b_i\}_{ij \in E, i \in V}$ be a set of {\it pseudomarginals},
i.e., functions satisfying $\sum_{x_i}b_{ij}(x_i,x_j)=b_{j}(x_j)$, $\sum_{x_i,x_j}b_{ij}(x_i,x_j)=1$ and $b_{ij}(x_i,x_j) \geq 0$.
The {\it \Bfe function} is defined on this set by
\begin{align}
  F(b)= 
&-\sum_{ij \in E}\sum_{x_i,x_j} b_{ij}(x_i,x_j)\log\Psi_{ij}(x_i,x_j)  
-\sum_{i \in V}\sum_{x_i} b_{i}(x_i)\log\Psi_{i}(x_i)  
\nonumber \\
&+ \sum_{ij \in E}\sum_{x_i,x_j} b_{ij}(x_i,x_j) \log b_{ij}(x_i,x_j)  
+ \sum_{i \in V}(1-d_i)\sum_{x_i}b_i(x_i)\log b_i(x_i), \label{eq:introBfe}
\end{align}
where $d_i=|N_i|$ is the number of neighboring vertices of $i$.
Note that this function is not convex in general and possibly has multiple minima.
The result of \cite{YFWGBP} says that the stationary points of this problem correspond to the solutions of the \lbp.
(The positive definiteness of the Hessian of the \Bfe function will be discussed in Section \ref{sec:PDC}.
The uniqueness of LBP fixed point will be discussed in Chapter \ref{chap:unique}.)

Similar to the case of the mean field approximation,
this variational problem can be viewed as a KL-divergence minimization \cite{YFWconstructing}, i.e.,
if we take (not necessarily normalized) distribution
\begin{equation}
 q(x) = \prod_{ij} \frac{b_{ij}(x_i,x_j)}{b_i(x_i)b_j(x_j)} \prod_{i \in V} b_i(x_i),
\end{equation}
then $D(q||p) \approx F(b)+ \log Z$.
Since we are expecting the KL-divergence is nearly zero at a stationary point $b^{*}$,
this relation motivates to define the approximation $Z_B$ of the partition function $Z$ by
\begin{equation}
 \log Z_B := -F(b^{*}). 
\end{equation}

\subsection{Applications of LBP algorithm}
\label{sec:applicationsLBP}
Since LBP is essentially equivalent to the Bethe approximation,
its application dates back to the 1930's when Bethe invented the Bethe approximation \cite{Bethe}. 
In 1993, Berrou et al \cite{BGTnearShannon}, proposed a novel method of error correcting codes
and found its excellent performance.
This algorithm was later found to be a special case of LBP by McEliece and Cheng \cite{MCturbo}.
This discovery made the LBP algorithm popular.
Soon after that, the LBP algorithm is successfully applied to other problems including 
computer vision problems and medical diagnosis \cite{FPClowlevelvision,MWJempiricalstudy}. %
Since then, scope of the application of the LBP algorithm is expanding.  
For example, LBP has many application in image processing
such as super-resolution \cite{BSsuper}, estimation of human pose \cite{HYWhumanpose} and image reconstruction \cite{Tstatisticalimage}.
Gaussian \lbp is also used for 
solving linear equations \cite{SBgaussianlinear} and linear programming \cite{BTlinearprogramming}.

\subsection{Past researches and our approaches}
In this subsection, we review the past theoretical researches on LBP, motivating further analysis.
A large number of researches have been performed by many researchers to make better understanding of the LBP algorithm.

Behavior of LBP is complicated in general in accordance with non-convexity of the \Bfe.
LBP has possibly many fixed points, and furthermore, may not converge.
For discrete variable models, because of the lower-boundedness of \Bfe,
at least one fixed point is guaranteed to exist \cite{YFWconstructing}. %
Fixed points are not necessarily unique in general, but for trees and one-cycle graphs,
the fixed point is guaranteed to be unique \cite{W1loop}.
This fact motivates analysis on classes of graphs that have a unique fixed point.
Each LBP fixed point is a solution of a nonlinear equation associated with the graph.
Therefore, the problem of the uniqueness of LBP is the uniqueness of the solution of this equation.
In the next section we discuss the history of this kind of problems in mathematics to show
an alternative origin of our research.

As mentioned above, the algorithm does not necessarily converge and often shows oscillatory behaviors.
Concerning the discrete variable model, Mooij \cite{MKsufficient} gives a sufficient condition for convergence in terms of the spectral radius of 
a certain matrix related to the graph geometry.
This matrix is the same as the matrix that appears in the (multivariate) Ihara-graph zeta function \cite{STzeta1}.
The graph zeta function is a popular characteristic of a graph; it is originally introduced by Ihara \cite{Idiscrete}.
Mooij's result has not been considered from the view of the graph zeta function nor graph geometry.
In this thesis, developing a new formula, we show a partial answer why this matrix appears in the sufficient condition of convergence.

The approximation performance has been also a central issue for understanding empirical success of LBP.
Since the approximation of marginals for a discrete model is also an NP-hard problem \cite{DLapproximate}, it seems difficult
to obtain high quality error bounds. 
Therefore, rather than rigorous bounds, we need to develop intuitive understanding of errors.
For binary models,
Chertkov and Chernyak \cite{CCloopPRE,CCloop} derived an expansion called \ls 
that expresses the ratio of $Z$ and its Bethe approximation  
in a finite sum labeled by a set of subgraphs called generalized loops.
We also derive an expansion of marginals in a similar manner.
An interesting point of the \ls is that the graph geometry explicitly appears in the error expression
and non-existence of generalized loop in a tree immediately implies the exactness of the Bethe approximation and LBP.
Concerning the error of marginals, Ikeda et al \cite{ITAsto,ITAinfo} have derived perturbative analysis of marginals 
based on the information geometric methods \cite{ANmethods}.
For Gaussian models, though the problem is not NP-hard, 
Weiss and Freeman have shown that the approximated means by LBP are exact but not for covariances \cite{WFcorrectness}.

For understanding of LBP errors, we follow the \ls approach initiated by Chertkov and Chernyak.
One reason is that the full series is easy to handle because it is a finite sum.
Though the expansion is limited to binary models, it covers important applications such as error correcting codes.

As discussed in the previous subsection,
LBP is interpreted as a minimization problem, where the objective function is the \Bfe. 
Empirically, Yedidia et al \cite{YFWGBP} found that locally stable fixed points of LBP
are local minima of the \Bfe function.
For discrete models, Heskes \cite{Hstable} has shown that stable fixed points of LBP are local minima of the \Bfe function.
This fact suggests that LBP finds a locally good stationary point.
From a theoretical point of view, this relation suggests that there is a covered relation between the LBP 
update and the local structure of the \Bfe.

Analysis of the \Bfe itself is also an important issue for understanding of LBP.
Pakzad et al \cite{PAstat} have shown that the \Bfe is convex if the underlying graph is a tree or one-cycle graph.
But for general graphs, (non) convexity of the \Bfe has not been comprehensively investigated.
As observed from Eq.~(\ref{eq:introBfe}), the Hessian (the matrix of second derivatives) of the \Bfe does not
depend on the given compatibility function, i.e., only determined by the graph geometry.

The variational formulation naturally derives an extension of the LBP algorithm called Generalized \BP (GBP)
that is equivalent to an extension of the Bethe approximation: Kikuchi approximation \cite{YFWGBP,Kikuchi}.
Inspired by this result, many modified variational problems have been proposed.
For example, Wiegerinck and Heskes \cite{WHfractional} %
have proposed a generalization of the \Bfe by introducing tuning parameter in coefficients.
This free energy yields fractional \bp algorithm. 
Since these extended variational problems include the variational problem of the \Bfe function,
it is still important to understand the \Bfe as a starting point.

Finally, we summarize our approach to analysis of LBP motivated by past researches.
For tree structured graphs, LBP has desirable properties such as the uniqueness of solution, exactness and convergence at finite step.
However, as observed in past researches, existence of cycles breaks down such properties.
Organizing these fragmented observations,
our analysis tries to make comprehensive understanding on the relation between the \lbp and graph geometry.
Indeed past researches of LBP have not treated ``graph geometry'' in a satisfactory manner;
few analysis has derived clear relations going beyond tree/not tree classification.
Malioutov et al \cite{MJWwalk} have shown that errors of Gaussian \bp are related to walks of the graph and the universal covering tree, 
but it is limited to the Gaussian case.

\section{Discrete geometric analysis}
In this thesis, we emphasize the discrete geometric viewpoint, which utilizes graph characteristics
such as graph zeta function and graph polynomials.
First, we introduce another mathematical background of this thesis: the interplay between geometry and equation.
This viewpoint puts our analysis of LBP in a big stream of mathematics.
Then, we discuss what kind of discrete geometry we should consider.

\subsection{Geometry and Equations}
The fixed point equation of the LBP algorithm Eq.~(\ref{eq:bplbpupdate}) involves messages.
The messages are labeled by the directed edges of the graph and satisfy local relations.
Therefore the structure of the equation is much related to the graph.
Since it is an equation, it is natural to ask whether there is a solution.
And if so, how many are there and what kind of structure do they have?
As mentioned in the previous section, if the underlying graph is a tree or one-cycle graph,
the uniqueness of the LBP solution is easily shown \cite{Pearl,W1loop}.

Equations that have variables labeled by points in a geometric object often have appeared in mathematics and formed a big stream \cite{USSgift}.
There are many examples that involve deep relations between the topology of a geometric object and 
the properties of equations on it such as solvability. %
In this thesis, we emphasize this aspect of the LBP equation and add a new example of this story.

Here we explain such an interplay by elementary examples.
We can start with the following easy observation.
\begin{equation}
(A) 
\begin{cases}
x= 2x+ 3y +z+3 \\
y= -y+z+3 \\
z= 2z-1 
\end{cases}
\qquad
(B)
\begin{cases}
x= 2x+ 3y +z+1 \\
y= x-y+z+2+1 \\
z= -x+y+2z+3
\end{cases}
\end{equation}
One may immediately find the solution of Eq.~(A),
by computing $z, y$ and $x$ successively. 
But one may not find a solution of Eq.~(B) without paper and a pencil.
The difference is easily realized by a graphical representation of these equations
(Fig.~\ref{LinearEqGraph}).
The first one does not include a directed cycle, i.e. a sequence of directed edges that ends the starting point,
but the latter has.

\begin{figure}
\begin{center}
\includegraphics[scale=0.35]{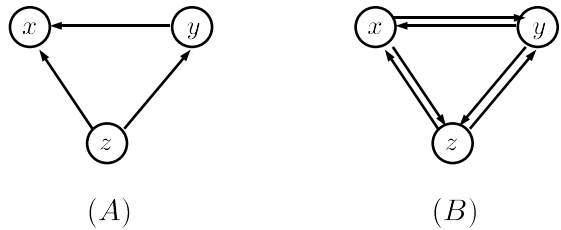} 
\caption{Graph representation of the equations} \label{LinearEqGraph}
\end{center}
\end{figure}

This type of difference is understood by the following setting.
Consider a linear equation $x=Ax+c$.
If $A$ is an upper diagonal matrix, this equation is solved in order of $O(N^2)$ computations by solving one by one.
However, for general matrix $A$, the required computational cost is $O(N^3)$ \cite{GVmatrix}.
Existence of directed cycles makes the equation difficult.

Another easy example of interrelation is the {\it Laplace equation};
\begin{equation}
 \nabla^{2} \phi(x) =0 \quad  x \in \Omega; \quad \phi(x)=0 \quad x \in \partial \Omega. 
\end{equation}
Here, the ``geometric object'' is the region $\Omega \subset \mathbb{R}^{n}$.
This differential equation is characterized by the variational problem of the energy functional:
\begin{equation}
 E[\phi]:=  \frac{1}{2} \int_{\Omega}\norm{\nabla \phi}^{2} dx.
\end{equation}
This characterization reminds us the relation between the LBP fixed point equation and
the variational problem of the \Bfe functional.
In this analogy, the region $\Omega$ corresponds to the graph $G=(V,E)$ and
the function $\phi$ corresponds to messages, or equivalently pseudomarginals, on $G$.
The variation of the energy functional and the \Bfe gives equation
$\nabla^{2} \phi(x) =0$ and the LBP fixed point equation, respectively.
The equation $\nabla^{2} \phi(x) =0$ is local because it is a differential equation.
Similarly, the LBP fixed point equation is also local in a sense that it only involves neighboring messages.
An important difference is that the geometric object $G$ is discrete.

The space of solution is much related to the geometry of $\Omega$.
For example let $n=2$ and $D_1=\{(x,y)|x^2+y^2 < 1\}$ and
$D_2=\{(x,y)|1 < x^2+y^2 \}$.
If $\Omega=D_1$, the only solution is $\phi=0$ from the maximum principle \cite{Acomplex}.
But if $\Omega=D_2$, there is a nonzero solution 
\begin{equation*}
 \phi(x,y)= \ln(x^2+y^2).
\end{equation*}
The Laplace equation can be generalized to be defined on Riemannian manifolds.
The spectrum of the Laplace operator is intimately related to the zeta function of the manifold
which is defined by a product of prime cycles \cite{Sriemannian}.
Furthermore, there is an analogy between Riemannian manifolds and finite graphs,
and the graph zeta function is known as a discrete analogue of the zeta function of Riemannian manifold \cite{AStwisted}.
The spectrum of a discrete analogue of the Laplacian is investigated by \cite{Cspectral}. 

It is noteworthy that the LBP fixed point equation is a non-linear equation
though aforementioned examples are linear equations.
Analysis of LBP does not reduce to finite dimensional linear algebra, e.g., eigenvalues and eigenvectors.
This fact potentially produces a new aspect of analysis on graphs compared to 
linear algebraic analysis on graphs \cite{Godsil,EKKanalysis}.

\subsection{What is the geometry of graphs?}
\label{sec:whatis}
Let us go back to the question:
what kind of discrete geometry should we employ to understand LBP and the Bethe approximation?
We have to think of graph quantities that are consistent with properties of
LBP and the partition function.
In other words, if there is some theory that relates the graph geometry and properties of 
the partition function and LBP, they must share some common properties.
Such requirements give hints to our question.

One may ask for a hint of topologist. 
A graph $G=(V,E)$ is indeed a topological space
when each edge is regarded as an interval $[0,1]$ and they are glued together at vertices. 
However, the basic topology theory can not treat rich properties of the graph,
because it can not distinguish homotopy equivalent spaces and
the homotopy class of a graph is only determined by the number of {\it connected components} $k(G)$
and the {\it nullity} $n(G)=|E|-|V|-k(G)$.
In this sense, graph theory is not in a field of topology but rather a combinatorics \cite{Mconcise}.

For the computation of the partition function, the nullity is much related to its difficulty.
Let $K$ be the number of states of each variable.
We can compute the partition function by $K^{n(G)}$ sums,
because if we cut $n(G)$ edges of $G$, we obtain a tree.
But the partition function on a bouquet graph, i.e. a graph that has one vertex with multiple edges,
is easily computed in $K$ steps.
Therefore, for the understanding of the computation and the behavior of LBP,
we need more detailed information of graph geometry
that distinguishes graphs with the same nullity.

Therefore, we should ask for graph theorists. 
Graph theory has been investigating graph geometry in many senses.
There are many graph characteristics, which are invariant with respect to graph isomorphisms \cite{EMinter1,EMinter2}.
The most famous example is the Tutte polynomial \cite{Tring},
which plays an important role in graph theory, a  broad field of mathematics and theoretical computer science. 
However, this thesis does not discuss the Tutte polynomial because it does not meet criteria discussed below.

In the LBP equation,
one observes that vertices of degree two can be eliminated
without changing the structure of the equation,
because this is just a variable elimination.
On the other hand, for the problem of computing the true partition function,
one also observes that vertices of degree two can be eliminated
with low computational cost keeping the partition function.

The operation of eliminating edges with a vertex of degree one
also keeps the problem essentially invariant, i.e.,
LBP solutions are invariant under the operation and
the true partition function does not change up to a trivial factor.

In this thesis, we consider two objects associated with the graph:
graph zeta function and $\Theta$ polynomial.
We use them in a multivariate form;
in other words, we define them on (directed) edge-weighted graphs.
These quantities are desired properties consistent with the above observations. That is,
\begin{enumerate}
\item{} ``invariant'' to removal of a vertex of degree one and the connecting edge.
\item{} ``invariant'' to erasure of a vertex of degree two.
\end{enumerate}
For the second property, we need to explain more.
Assume that there is a vertex $j$ of degree two and its neighbors are $i$ and $k$.
If we have directed edge weights $u_{\edij}$ and $u_{\edjk}$,
we can erase the vertex $i$ taking $u_{\edik}=u_{\edji}u_{\edij}$ 
keeping the graph zeta function invariant.
A similar result holds for the $\Theta$ polynomial though its weights are
associated with undirected edges.

The first property immediately implies that these quantities are in some sense
trivial if the graph is a tree.
This property reminds us that LBP gives the exact result if the graph is a tree.

Indeed, these properties do not uniquely determine the quantities associated to graphs.
However, it gives a clue to answer our question: what kind of graph geometry is related to LBP?
\section{Overview of this thesis}
\label{sec:overview}

The remainder of the thesis is  organized in the following manner.

\subsubsection{Chapter 2: Preliminaries}
This chapter sets up the problem formally,
introducing hypergraphs, graphical models and exponential families.
The \LBP (LBP) algorithm is also introduced
utilizing the language of the exponential families.
Our characterizations of LBP fixed points gives an understanding of the \Bzf,
as discussed in the first section of Chapter \ref{chap:Bzf}.

\subsubsection{Part I: Graph zeta in \Bfe and \lbp}
The central result of this part is the relation between the Hessian of the \Bfe
and the multivariate graph zeta function.
The multivariate graph zeta function is a computable characteristic of an edge weighted graph
because it is represented by the determinant of a matrix indexed by edges.

The focus of this part is mainly an intrinsic nature of the LBP algorithm and the \Bfe.
Namely, we do not treat the true partition function and the Gibbs free energy.
Interrelation of such exact quantities and their Bethe approximations is discussed in the next part.

The contents of this part is an extension of the result in \cite{WFzeta} 
where only pairwise and binary models are discussed.

\subsubsection{Chapter 3: Graph zeta function}
This chapter develops the graph zeta function and related formulas.
First, we introduce our graph zeta function unifying known types of graph zeta functions. 
Secondly, we show the \IB type determinant formula which plays an essential role in the next chapter.
Some basic properties of the univariate zeta function, such as places of the poles, are also discussed.

\subsubsection{Chapter 4: Bethe-zeta formula}
This chapter presents a new formula, called \Bzf,
which establishes the relation between the Hessian of the Bethe free energy function and the graph zeta function.  
This formula is the central result in Part I.
The proof of the formula is based on the \IB type determinant formula and Schur complements of 
the (inverse) covariance matrices.  
Demonstrating the utility of this formula,
we discuss two applications of this formula.
The first one is the analysis of the positive definiteness and convexity of the \Bfe function;
the second one is the analysis of the stability of the LBP algorithm.

\subsubsection{Chapter 5: Uniqueness of LBP fixed point}
This chapter develops a new approach to the uniqueness problem of the LBP fixed point.
We first establish an index sum formula and combine it with the \Bzf.
Our main contribution of this chapter is the uniqueness theorem for unattractive (frustrated) models on graphs with nullity two.
Though these are toy problems,
the analysis exploits the graph zeta function and is theoretically interesting.
This chapter only discusses the binary pairwise models but our approach can be basically generalized to
multinomial models.

\subsubsection{Part II: Loop Series}
In this part, focusing on binary models, 
we analyze the relation between the exact quantity, such as the partition function and marginal distributions,
and their Bethe approximations using the \ls technique.
The expansion provides graph geometric intuitions of LBP errors.

\subsubsection{Chapter 6: Loop series}
Loop Series (LS), which is developed by Chertkov and Chernyak \cite{CCloopPRE,CCloop}, 
is an expansion that expresses the approximation error in a finite sum in terms of a certain class of subgraphs.
The contribution of each term is the product of local contributions, which are easily calculated by the LBP outputs.
First we explain the derivation of the LS in our notation, which is suitable for the graph polynomial treatment
in the next chapter.
In a special case of the perfect matching problems,
we observe that the \ls has a special form and is related to the graph zeta function.
We also review some applications of the \ls.

\subsubsection{Chapter 7: Graph polynomials from Loop Series}
This chapter treats the \ls as a weighted graph characteristics called theta polynomial, $\Theta_{G}(\bs{\beta},\bs{\gamma})$.
Our motivation for this treatment is to ``divide the problem in two parts.''
The \ls is evaluated in two steps: 
1.~the computation of $\bs{\beta}=(\beta_{ij})_{ij \in E}$ and $\bs{\gamma}=(\gamma_{i})_{i \in V}$ by an LBP solution;
2.~the summation of all subgraph contributions. 
Since the first step seems difficult, we focus on the second step.
If there is an interesting property in the form of the sum, or the $\Theta$-polynomial,
the property should be related to the behavior of the error of the partition function approximation.

Though we have not been successful in deriving properties of $\Theta$-polynomial
that can be used to derive properties of the Bethe approximation,
we show that the graph polynomials $\theta_{G}(\beta,\gamma)$ and $\omega_{G}(\beta)$, 
which are obtained by specializing $\Theta_G$, have interesting properties: \dcr.
We also discuss partial connections to the Tutte polynomial and the monomer-dimer partition function.
We believe that these results give hints for future investigations of $\Theta$-polynomial.

\subsubsection{Chapter 8: Conclusion}
This chapter concludes this thesis and suggests some future researches.

\subsubsection{Appendix}
In Appendix A, we summarize useful mathematical formulas, which are used
in proofs of this thesis.
In Appendix B, we put topics on LBP which are not necessary for the logical thread of this thesis,
but helpful for further understandings.

%% file: chapter2.tex

In this chapter, we introduce objects and methods studied in this thesis.
Probability distributions that have ``local'' factorization structures 
appear in many fields including physics, statistics and engineering.
Such distributions are called graphical models. 
\LBP (LBP) is an efficient approximation method applicable to inference problems on graphical models.
The focus of this thesis is an analysis of this algorithm applied to any graph-structured distributions.
We begin in Section~\ref{sec:prob} with elements of hypergraphs as well as graphical models
because the associated structures with these graphical models are, precisely speaking, hypergraphs.
Section~\ref{sec:LBP} introduces the LBP algorithm
on the basis of the theory of exponential families.
A collection of exponential families, called \ifa, is utilized to formulate the algorithm.
The Bethe free energy, which gives alternative language for formulating the approximation by the LBP algorithm,
is discussed in Section~\ref{sec:BFE}, providing characterizations of LBP fixed points.

\section{Probability distributions with graph structure}\label{sec:prob}
Probability distributions that are products of ``local'' functions
appears in a variety of fields, including
statistical physics \cite{Pcluster,Ggibbs}, %
statistics \cite{Wgraphical},
artificial intelligence \cite{Pearl},
coding theory \cite{MCturbo,Mgood,GLDPC},
machine learning \cite{Jlearning},
and combinatorial optimizations \cite{MPZanalytic}.
Typically, such distributions come from
system modeling of random variables that only have ``local'' interactions/constraints.  
These factorization structures are well visualized by graph representations, called factor graphs.
Furthermore, the structures are cleverly exploited in the algorithm of LBP.

We start in Subsection \ref{sec:basicsgraph} with an introduction of hypergraphs
because factor graphs are indeed hypergraphs.
Further theory of hypergraphs is found in \cite{Bhypergraphs}.
Subsection \ref{sec:fgrep} formally introduces the factor graph of graphical models with some examples.

\subsection{Basic definitions of graphs and hypergraphs}\label{sec:basicsgraph}
We begin with the definition of (ordinary) graphs.
A {\it graph} $G=(V,E)$ consists of the vertex set $V$ joined by edges of $E$. 
Generalizing the definition of graphs, we define hypergraphs. 
A {\it hypergraph} $H=(V,F)$ consists of a set of {\it vertices} $V$ and a set of {\it hyperedges} $F$.
A hyperedge is a non-empty subset of $V$.
Fig.~\ref{fig:sethypergraph} illustrates a hypergraph $H=(\{1,2,3\},\{\alpha_1,\alpha_2,\alpha_3\})$,
where $\alpha_1=\{1,2\}$, $\alpha_2=\{1,2,3,4\}$ and $\alpha_3=\{4\}$.
In order to describe the message passing algorithm in Section \ref{sec:basicLBP},
it is convenient to identify a relation $i \in \alpha$ with a directed edge $\edai$.
The left of Fig.~\ref{fig:DBhypergraph} illustrates this representation of the above example, where squares represent hyperedges.
Therefore, explicitly writing the set of directed edges $\vec{E}$,
a hypergraph $H$ is also denoted by $\fgdefn$.

It is also convenient to represent a hypergraph as a bipartite graph.
A graph $G=(V,E)$ is {\it bipartite} if 
the vertices are partitioned into two set, say $V_1$ and $V_2$, and
all edges join the vertices of $V_1$ and $V_2$.
A hypergraph $H=(V \cup F, \vec{E})$ is identified with a bipartite graph
$B_H=(V \cup F, E_{B_H} )$, where $E_{B_H}$ is obtained by forgetting the directions of $\vec{E}$.
(See the right of Fig.~\ref{fig:DBhypergraph}.)

\begin{figure}
\begin{minipage}{.33\linewidth}
\begin{center}
\includegraphics[scale=0.27]{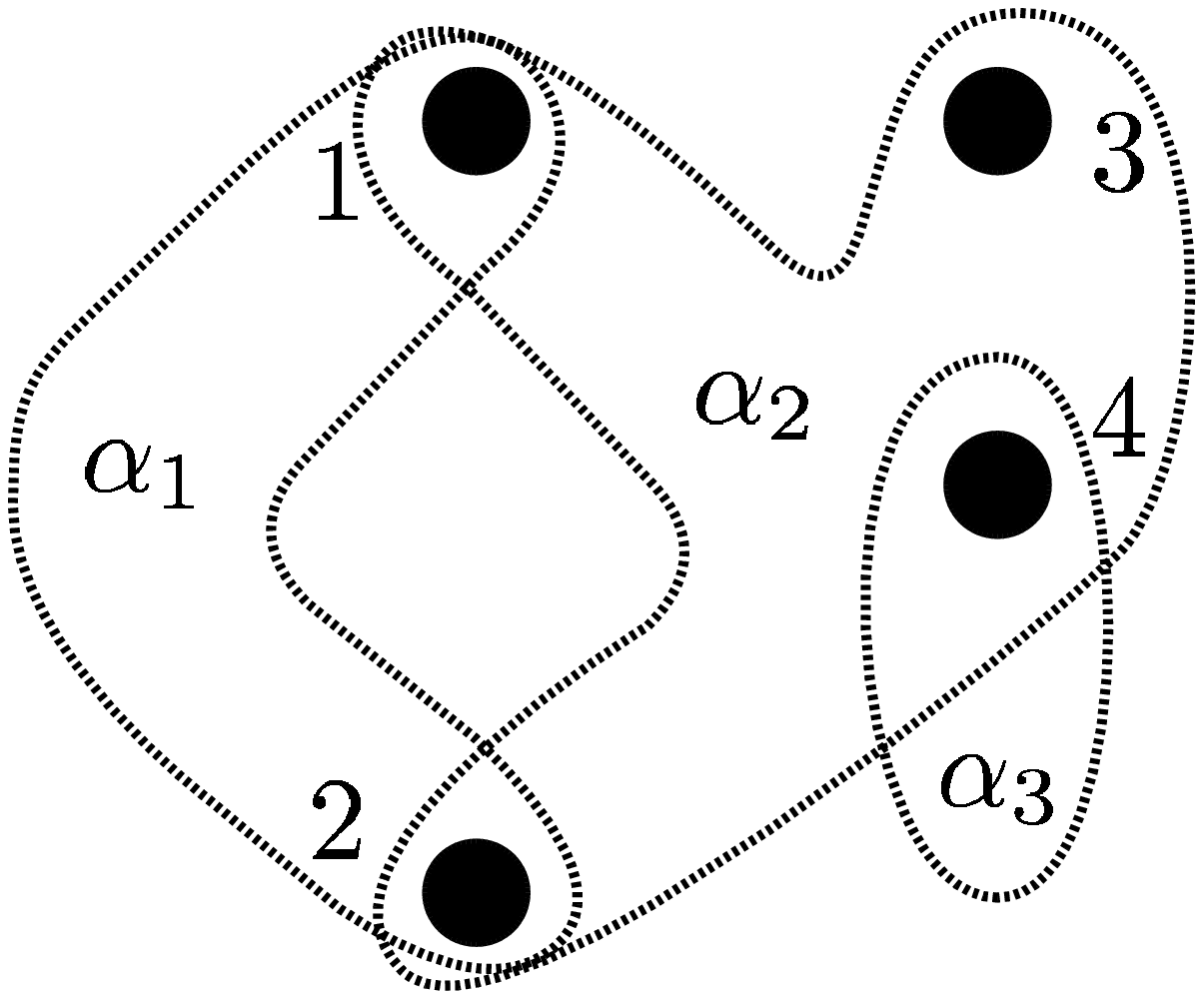}
\vspace{-2mm}
\caption{Hypergraph $H$. \label{fig:sethypergraph}}
\end{center}
\end{minipage}
\begin{minipage}{.66\linewidth}
\begin{center}
\includegraphics[scale=0.27]{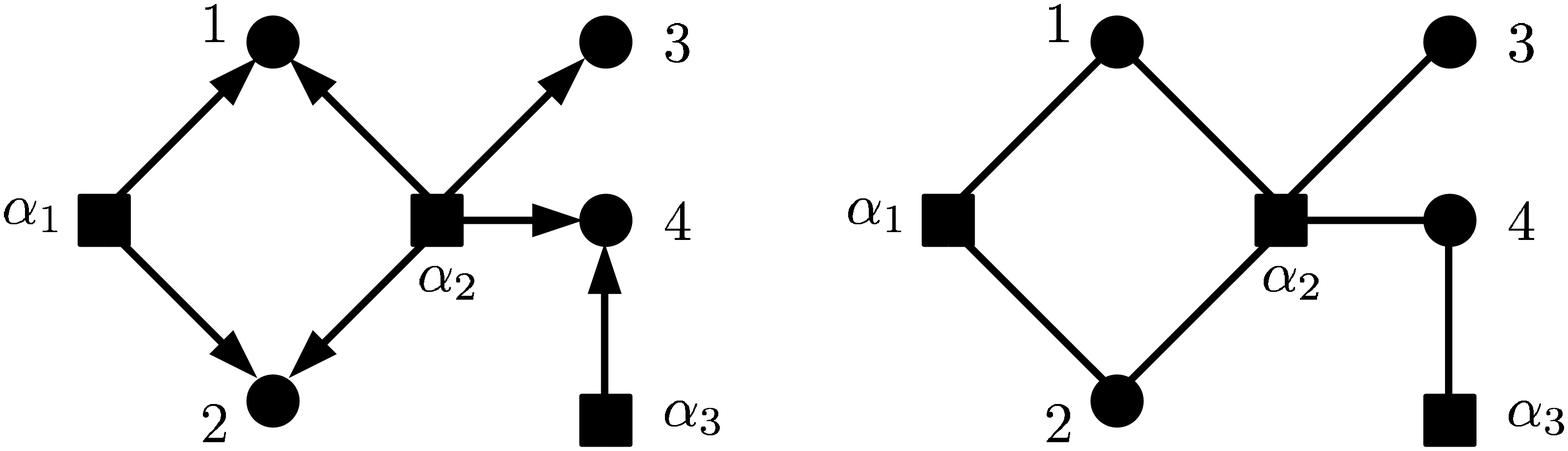}
\vspace{-2mm}
\caption{Two representations. \label{fig:DBhypergraph}}
\end{center}
\end{minipage}
\end{figure}

For any vertex $i \in V$, the {\it neighbors} of $i$ is defined by
$N_i:=\{\alpha \in F | i \in \alpha \}$.
Similarly, for any hyperedge $\alpha \in F$, the {\it neighbors} of $\alpha$ is defined by
$N_{\alpha}:=\{ i \in V | i \in \alpha \}=\alpha$.
The {\it degrees} of $i$ and $\alpha$ are given by
$d_i:=|N_i|$ and $d_{\alpha}:=|N_{\alpha}|=|\alpha|$, respectively.
A hypergraph $H=(V,F)$ is called {\it $(a,b)$-regular} if $d_i=a$ and $d_{\alpha}=b$ for all $i \in V$
and $\alpha \in F$. 
If all the degrees of hyperedges are equal to two,
a hypergraph is naturally identified with a graph.

A {\it walk} $W=(i_1,\alpha_1,i_2,\ldots,\alpha_{n},i_{n+1})$ of a hypergraph
is an alternating sequence of vertices and hyperedges that satisfies
$\alpha_{k} \supset \{i_k,i_{k+1}\}$, $i_k \neq i_{k+1}$ for $k=1,\ldots,n$.
We say that $W$ is a walk from $i_1$ to $i_{n+1}$ and has length $n$.
A walk $W$ is said to be {\it closed} if $i_1 =i_{n+1}$.
A {\it cycle} is a closed walk of distinct hyperedges.
A hypergraph $H$ is {\it connected} if for every pair of distinct vertices $i,j$
there is a walk from $i$ to $j$.
Obviously, a hypergraph is a disjoint union of connected components.
The number of connected components of $H$ is denoted by $k(H)$.
The {\it nullity} of a hypergraph $H$ is defined by
$n(H):=|V|+|F|-|\vec{E}|$.

\begin{defn}
A hypergraph $H$ is a {\it tree} if it is connected and has no cycle.
\end{defn}
This condition is equivalent to $n(H)=0$ and $k(H)=1$.
Other characterization will be given in Propositions \ref{prop:treecore} and \ref{prop:treeprime}.
Note that this definition of tree is different from {\it hypertree} known in graph theory and computer science \cite{SSaxioms,GLShypertree}.
For example, the hypergraph in Fig~\ref{fig:sethypergraph} is a hypertree though it is not a tree in our definition.

\subsubsection{Core of hypergraphs}
Here, we discuss the core of hypergraphs,
which gives another characterization of trees.

The {\it core}\footnote{This term is taken from \cite{Stopology} where the core of graphs is defined.
Note that this notion is different from the core in \cite{Godsil}.} 
of a hypergraph $H=(V,F)$, denoted by $\core(H)$, 
is a hypergraph that is obtained by the union of the cycles of $H$. 
In other words, $\core(H)=(V',F')$ is given by
$F'=\{\alpha \in F| \alpha \text{ is in some cycles of }H \}$ and
$V'=\{i \in V| i  \text{ is in some cycles of }H\}$.
A hypergraph $H$ is said to be a {\it coregraph}
if $H=\core(H)$.
See Fig.~\ref{fig:corehypergraph} for an example.

\begin{figure}
\begin{center}
\includegraphics[scale=0.25]{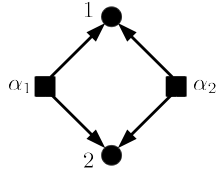}
\vspace{0mm}
\caption{Core of hypergraph in Fig.~\ref{fig:DBhypergraph}. \label{fig:corehypergraph}}
\end{center}
\end{figure}

Intuitively, the core of a hypergraph is obtained by removing vertices and hyperedges
of degree one until there is neither such vertices nor hyperedges.
More precisely, the operation for obtaining the core is as follows.
First, for $\fgdefn$, find a directed edge $(\edai) \in \vec{E}$ that satisfies $d_{\alpha}=1$ or $d_{i}=1$.
If both of the condition is satisfied, remove $\alpha, i \text{ and } (\edai)$.
If either of them is satisfied, remove $(\edai)$ and the degree one vertex/hyperedge.

The following characterization of tree is trivial from the above definitions.
\begin{prop}
\label{prop:treecore}
 A connected hypergraph $H$ is a tree if and only if $\core(H)$ is the empty hypergraph.
\end{prop}

\subsection{Factor graph representation}\label{sec:fgrep}
Our primary interest is probability distributions
that have factorization structures represented by hypergraphs.

\begin{defn}
Let $H=(V,F)$ be a hypergraph.
For each $i \in V$,
let $x_i$ be a variable that takes values in $\mathcal{X}_i$.
A probability density function $p$ on $x=(x_i)_{i \in V}$ is said to be {\it graphically factorized}
with respect to $H$ if it has the following factorized form
\begin{equation}
p(x)=\frac{1}{Z}
\prod_{\alpha \in F} \Psi_{\alpha}(x_{\alpha}), \label{defp}
\end{equation}
where $x_{\alpha}=(x_i)_{i \in \alpha}$, $Z$ is the normalization constant
and $\Psi_{\alpha}$ is non-negative valued function called {\it compatibility function}.
A set of compatibility functions, giving a graphically factorized distribution,
is called a {\it graphical model}.
The associated hypergraph $H$ is called the {\it factor graph} of the graphical model.
\end{defn}
We often refer to a hypergraph as a factor graph,
implicitly assuming that it is associated with some graphical model.
For a factor graph, a hyperedge is usually called a {\it factor}.
Factor graph is explicitly introduced in \cite{KFLfactor}.

Any probability distribution on $\mathcal{X}= \prod_{i} \mathcal{X}_i$ is trivially 
graphically factorized with respect to the ``one-factor hypergraph,'' where the unique factor includes all vertices.  
It is more informative if the factorization involves factors of small sizes.
Our implicit assumption in the subsequent chapters is that
for all factors $\alpha$, $\mathcal{X}_{\alpha}=\prod_i \mathcal{X}_i$ is small, in sense of cardinality or dimension, 
enough to be handled by computers.

A Markov Random Field (MRF) is an example that have such a factorization structure.
Let $G=(V,E)$ be a graph and $\mathcal{X}=\prod_{i \in V} \mathcal{X}_i$ be a discrete set.
A positive probability distribution $p$ of $\mathcal{X}$ is said to be a {\it Markov random field} on $G$
if it satisfies
\begin{equation}
 p(x_i | x_{ V \smallsetminus i}) = p(x_i| x_{N_i}) \text{ for all } i \in V.
\end{equation} 
See, e.g., \cite{KSmarkov} for further materials.
A {\it clique} is a subset of vertices every two of which are connected by an edge.
Hammersley-Clifford theorem says that 
\begin{equation}
 p(x) \propto \prod_{ C \in \mathcal{C}} \Psi_{C}(x_C),
\end{equation}
where $\mathcal{C}$ is a set of cliques.
A proof of this theorem, using the M\"obius inversion technique, is found in \cite{Grandom}.

Bayesian networks provide another class of examples of factorized distributions.
The scope of applications of Bayesian networks includes
expert system \cite{CDLSprobabilistic}, speech recognition \cite{Jspeech} %
and bioinformatics \cite{DEKMbiological}. 
Consider a Directed Acyclic Graph (DAG), i.e., a directed graph without directed cycles. 
A {\it Bayesian network} is specified by local conditional probabilities associated with the DAG \cite{Pearl,CDLSprobabilistic}.
Namely, it is given by the following product 
\begin{equation}
 p(x) = \prod_{i \in V} p(x_i | x_{\pi (i)}),
\end{equation}
where $\pi(i)$ is the set of {\it parents} of $i$, i.e., the set of vertices
from which an edge is incident on $i$.
If $\pi(i)= \emptyset$, we take $p(x_i|\emptyset)=p(x_i)$.
The factor graph, associated with this distribution,
consists of factors $\alpha=\{i \} \cup \pi(i)$.

We often encounters a situation that
the ``global constraint'' of variables is given as a logical conjunction
of ``local constraints.''
A product of local functions can naturally represent such a situation.
For example, in {\it linear codes}, a sequence of binary (0 or 1) variables $x$ has
constraints of the following form, called {\it parity check}:
\begin{equation}
 x_{i_1} \oplus x_{i_2} \oplus \ldots \oplus x_{i_k}=0, \label{eq:pritycheck}
\end{equation}
where $\oplus$ denotes the sum in $\mathbb{F}_2$.
For a given set of parity checks, a sequence of binary variables $x$ is called a {\it codeword}
if it satisfies all the conditions.
A parity check can be implemented by a local function $\Psi_{\alpha}$
that is equal to zero if $x_{\alpha}$ violates Eq.~(\ref{eq:pritycheck})
$(\alpha= \{ i_1, \ldots, i_k\})$.
Furthermore, the product of the local functions implies the condition for the linear code.
Satisfiability problem (SAT), coloring problem and matching problem, etc,
have the same structure.
\section{Loopy Belief Propagation algorithm}\label{sec:LBP}
\BP (BP) is an efficient method that calculates exact marginals of the given 
distribution factorized according to a tree-structured factor graph \cite{Pearl}. 
\LBP (LBP) is a heuristic application of the algorithm for factor graphs with cycles,
showing successful performance
in various problems.
We mentioned examples of applications in Subsection \ref{sec:applicationsLBP}.

In this thesis, we refer to a {\it family} as a collection of probability distributions.
We carefully distinguish a family and a model, which gives a single probability distribution.
First, in Subsection \ref{sec:expfamily}, we introduce a class of families, called exponential families,
because an {\it inference family}, which is needed for the LBP algorithm and introduced in Subsection \ref{sec:infmodel}, 
is a set of exponential families.
The detail of the LBP algorithm and its exactness on trees
are described in Subsections \ref{sec:basicLBP} and \ref{sec:LBPtree}.
Subsection \ref{sec:LBPstability} derives the differentiation of the LBP update at LBP fixed points,
which determines the stability of the algorithm.

\subsection{Introduction to exponential families}\label{sec:expfamily}
Exponential families are the simplest and the most famous class of probability distributions.
Many important stochastic models such as multinomial, Gaussian, Poisson and gamma distributions are all included in this class.
Here, we provide a minimal theory on exponential families.
The core of the theory is the Legendre transform of the log partition function
and the bijective transform between dualistic parameters, called \npara and \epara.
These techniques are exploited especially in the derivation of the
\Bzf in Section \ref{sec:detIhara}.
More details of the theory about the exponential families is found in books \cite{BNinformation,Bfundamentals}
and a composition from the information geometrical viewpoint is found in \cite{ANmethods}.

The following definition of the exponential families is not completely rigorous,
but it would be enough for the purpose of this thesis.
\begin{defn}
Let $\mathcal{X}$ be a set and $\nu$ be a base measure on it.
For given $n$ real valued functions $\bs{\phi}(x)=(\phi_1(x),\ldots,\phi_n(x))$ %
on $\mathcal{X}$, 
a parametric family of probability distributions on $\mathcal{X}$ is given by
\begin{equation*}
 p(x;\bs{\theta})=
\exp \left( 
\sum_{i=1}^{N} \theta_i \phi_i(x)  - \psi(\bs{\theta})
\right), \qquad \quad
\psi(\bs{\theta}):= \log \int \exp \left( \sum_{i=1}^{N} \theta_i \phi_i(x) \right) {\rm d} \nu(x)
\end{equation*}
and is called an {\it exponential family}.
The parameter $\bs{\theta}$, called {\it natural parameter},
ranges over the set $\Theta:= {\rm int} \{\bs{\theta} \in \mathbb{R}^{N}; 
\int 
\exp ( \sum_{i=1}^{N} \theta_i \phi_i(x) ) \di \nu < \infty
\}$, where int denotes the interior of the set.
The function $\bsphi (x)$ is called the {\it sufficient statistic} and
$\psi(\bs{\theta})$ is called the {\it log partition function}.
\end{defn}

An affine transform of the \nparas gives another exponential family; we identify it with the original family.

It is known that one can differentiate the log partition function at any number of times
by interchanging differential and integral operations \cite{Bfundamentals}.
One easily observes that
$\Theta$ is a convex set and $\psi(\bs{\theta})$ is a convex function on it.
Actually, the convexity of $\Theta$ is derived from the convexity of the exponential function.
The Hessian of $\psi$ 
\begin{equation}
 \pds{\psi}{\theta_i}{\theta_j}=\cov{p_{\bs{\theta}}}{\phi_i}{\phi_j} \quad i,j=1,\ldots,N \label{covasdiff}
\end{equation}
is obviously positive semidefinite and thus $\psi$ is convex.

In this thesis, we require the following regularity condition for exponential families.
\begin{asm}
\label{asm:expregular}
 The $N$ by $N$ matrix Eq.~(\ref{covasdiff}) is positive definite.
\end{asm}

\subsubsection{Legendre transform}
The heart of the theory of exponential family is the duality coming from the Legendre transform, 
which is applicable to any convex function and derives the dual parameter set.
A comprehensive treatment of the theory of the Legendre transform is found in \cite{BLconvex}.

First, we introduce a transform of the natural parameter to the dual parameter.
For the sufficient statistics $\bsphi$, let ${\rm supp} \bsphi$ be the minimal closed set $S \subset \mathbb{R}^N$
for which $\nu(\bsphi^{-1}( \mathbb{R}^N \smallsetminus S))=0$.
The dual parameter set, called the {\it expectation parameters},
is defined by $Y:= {\rm int}( {\rm conv}( {\rm supp} \bsphi) )$.
Obviously, $Y$ is an open convex set.
If $\mathcal{X}$ is a finite set, $Y$ is explicitly expressed as follows:
\begin{equation*}
 Y = \{ \sum_{x \in \mathcal{X}'} \alpha_{x} \phi (x)| \sum_{x \in \mathcal{X'}} \alpha_{x}=1, \alpha_x > 0\},
\end{equation*}
where $\mathcal{X'}=\{ x \in \mathcal{X}| \nu(\{x\}) > 0 \}$.

The following theorem is the fundamental result establishing the transform to this dual parameter set.
\begin{thm}
[\cite{Bfundamentals}] %
A map 
\begin{equation*}
 \Lambda : \Theta \ni \bstheta \longmapsto 
\pd{\psi}{\bstheta}(\bs{\theta})  \in Y
\end{equation*}
is a bijection.
\end{thm}
\begin{proof}
We only prove the injectivity of the map $\Lambda$.
Take distinct points $\bstheta$ and $\bstheta'$ in $\Theta$.
Define
\begin{equation}
 f(t):= \inp{\bstheta' -\bstheta}{ \Lambda ( \bstheta+t(\bstheta'-\bstheta))} \quad t \in [0,1],
\end{equation}
where $\inp{\cdot}{\cdot}$ is the standard inner product.
Since the covariance matrix is positive definite from Assumption \ref{asm:expregular},
$f(t)$ is strictly increasing.
Therefore, 
\begin{equation*}
 f(1)-f(0)=\inp{\bstheta' -\bstheta}{ \Lambda (\bstheta') -\Lambda( \bstheta)} > 0.
\end{equation*} 
This yields $\Lambda (\bstheta') \neq \Lambda( \bstheta)$.

The proof of the surjectivity is found in Theorem 3.6. of \cite{Bfundamentals}.
\end{proof}
The map $\Lambda$, which is referred to as a {\it moment map}, is also written as the expectation of the sufficient statistic
$\Lambda (\bstheta)= \E{p_{\bstheta}}{\bsphi}$.

The Legendre transform of $\psi(\bs{\theta})$ on $\Theta$, which gives a convex function on the dual parameter set,
is defined by
\begin{equation}
 \varphi(\bseta)= \sup_{\bstheta \in \Theta} ( \inp{\bstheta}{\bseta} - \psi(\bstheta)),
\qquad \bs{\eta} \in Y, \label{LT}
\end{equation}
where $\inp{\bstheta}{\bseta}= \sum_i \theta_i \eta_i$ is the inner product.
This function is convex with respect to $\bseta$, because it is a supremum of linear functions. 
Since the expression in the supremum in Eq.~(\ref{LT}) is concave with respect to $\bstheta$,
the supremum is uniquely attained at $\hat{\bstheta}(\bseta)$ that satisfies
$\bseta =  \Lambda ( \hat{\bstheta}(\bseta) )$.
This equation implies that a map $\bseta \mapsto \hat{\bstheta}(\bseta)$ is the inverse of $\Lambda$.
Note that $\varphi$ is actually a negative entropy
\begin{equation}
\varphi (\bseta) =\E{p_{\hat{\bstheta} (\bseta)}}{\log p_{\hat{\bstheta} (\bseta)}}. \label{eq:negativeentropy}
\end{equation}
Note also that derivative of $\varphi$ gives the inverse of the map $\Lambda$, i.e.
$\pd{\varphi}{\bseta}(\bseta)=\Lambda^{-1}( \bseta )$,
which is easily checked by the differentiation of the equation
$\varphi(\bseta)=  \inp{ \hat{\bstheta}(\bseta )}{ \bseta} -   \psi(   \hat{\bstheta}(\bseta) )$.
Therefore, the Hessian of $\varphi$ is the inverse of the covariance matrix and thus
$\varphi$ is a strictly convex function.

The inverse transform of Eq.~(\ref{LT}) is obtained by an identical equation
\begin{equation}
 \psi(\bstheta)= \sup_{\bseta \in Y} ( \inp{\bstheta}{\bseta} -  \varphi(\bseta) ),
\qquad \bs{\theta} \in \Theta, \label{ILT}
\end{equation}
because the supremum in Eq.~(\ref{ILT}) is attained at
$\hat{\bseta}(\bstheta)$ that satisfies
$\bstheta = \Lambda^{-1}( \hat{\bseta} (\bstheta))$.

In summary, strictly convex functions $\psi$ and $\varphi$ are the Legendre transform of each other
and the \nparas and the \eparas are transformed by $\Lambda$ and $\Lambda^{-1}$, which are given by the derivatives of the functions.

\subsubsection{Examples of exponential families}
\begin{example}[Multinomial distributions]
\label{example:multinomial}
 Let $\mathcal{X}=\{1,\ldots,N\}$ be a finite set
with the uniform base measure.
We define sufficient statistics as
\begin{equation}
 \phi_k(x)=
\begin{cases}
 1 \text{ \quad  if } x=k \\
 0 \text{ \quad  otherwise}
\end{cases}
\end{equation}
for $k=1,\ldots,N-1$.
Then the given exponential family is called {\it multinomial distributions} and 
coincides with the all probability distributions on $\mathcal{X}$
that have positive probabilities for all elements of $\mathcal{X}$.

By definition, the region of \nparas is $\Theta=\mathbb{R}^{N-1}$.
The region of \eparas is the interior of probability simplex.
That is,
\begin{equation}
 Y=\{ (y_1,\ldots,y_N); \sum_{k=1}^{N}y_k=1, y_k > 0 \}.
\end{equation}

\end{example}

\begin{example}[Gaussian distributions]
Let $\mathcal{X}=\mathbb{R}^{n}$ with the Lebesgue measure
and
let $\phi_i(x_i)=x_i$ and $\phi_{ij}(x_i ,x_j)=x_i x_j$.
The exponential family given by the sufficient statistics $\bsphi(x)=(\phi_i(x_i),\phi_{jk}(x_j,x_k))_{1\leq i \leq n, 1 \leq j \leq k \leq n}$,
is called {\it Gaussian distributions}, consists of probability distributions of the form
\begin{equation}
 p(\bsx;\bstheta)= \exp \big(
\sum_{i \leq j}\theta_{ij}x_i x_j + \sum_{i}\theta_i x_i
-\psi (\bstheta)
\big).
\end{equation}
If we set $J_{ij}=J_{ji}= - \theta_{ij} ~~ (i \neq j)$, $J_{ii}=-2 \theta_{ii}$ and $h_i=\theta_i$,
it comes to
\begin{align}
& p(\bsx;\bstheta)= \exp \big(
-\frac{1}{2}\bsx^T J \bsx + \bsh^T \bsx - \psi (\bstheta)
\big), \\
& \psi(\bstheta)=
\frac{n}{2}\log 2 \pi - \frac{1}{2}\log \det J
+\frac{1}{2} \bsh^T J^{-1} \bsh.
\end{align}
Obviously, the set of \nparas is $\Theta= \{ \bstheta \in \mathbb{R}^{N} | J \text{ is positive definite.}\}$,
where $N= n + \frac{n(n+1)}{2} $.
As is well known, the mean and covariance of $\bsx$ are given by
$\bsmu= J^{-1}\bsh$ and $\Sigma= J^{-1}$, respectively.
The transform to the dual parameter is given by the expectation of sufficient statistic:
$\Lambda (\bstheta) = ( \mu_i, \Sigma_{ij}+ \mu_i \mu_j)$.
Therefore, the set of \eparas is 
$Y = \{ \bseta \in \mathbb{R}^N | \Sigma(\bseta):=(\eta_{ij} - \eta_i \eta_j)_{ 1 \leq i , j \leq n  } \text{ is positive definite.}  \}$.
The dual convex function $\varphi$ is 
\begin{equation}
 \varphi (\bseta)= - \frac{n}{2}(1+\log 2 \pi ) - \frac{1}{2} \log \det \Sigma (\bseta).
\end{equation}

For a given mean vector $\bsmu=(\mu_i)$,
the {\it fixed-mean Gaussian distributions} is the exponential family obtained by the sufficient
statistics $\bsphi(x)=\{(x_i-\mu_i)(x_j-\mu_j)\}_{1 \leq i \leq j \leq n}$.
Moreover, if $\bsmu=\bs{0}$, the family is called the {\it zero-mean Gaussian distributions}. 

\end{example}

\subsection{\Ifa for LBP}\label{sec:infmodel}
In this Subsection, we construct a set of exponential families used in the LBP algorithm.
In order to perform inferences using LBP for
a given graphical model, we have to fix a ``family'' that includes the probability distribution.

Let $H=(V,F)$ be a hypergraph.
In succession, we follow the notations in Subsection \ref{sec:fgrep}.
For each vertex $i$, we consider an exponential family $\mathcal{E}_i$ with a sufficient statistic $\phi_i$
\footnote{In the previous subsection, we used bold symbols to represent vectors,
but from here we simplify the notation.} 
and a base measure $\nu_i$ on $\mathcal{X}_i$. 
A \npara, \epara, the log partition function and its Legendre transform are denoted by $\theta_i$, $\eta_i$, 
$\psi_i$ and $\varphi_i$ respectively.
Furthermore, for each factor $\alpha=\fai$,
we give an exponential family $\mathcal{E}_{\alpha}$ on $\mathcal{X}_{\alpha}= \prod_{i \in \alpha} \mathcal{X}_i$
with the base measure $\nu_{\alpha}=\prod_{i \in \alpha} \nu_i$ and a sufficient statistic $\fa{\phi}$ of the form
\begin{equation}
 \fa{\phi}(x_{\alpha})=(\pa{\phi}(x_{\alpha}),\phi_{i_1}(x_{i_1}),\ldots,\phi_{i_{d_{\alpha}}}(x_{i_{d_{\alpha}}}) ).
\end{equation}
An important point is that $\fa{\phi}$ includes the sufficient statistics of $i \in \alpha$ as components.
The \npara, \epara, log partition function and its Legendre transform of this model are denoted by
\begin{equation}
 \fa{\theta}=(\pa{\theta},\va{\theta}{i_1},\ldots,\va{\theta}{i_{d_{\alpha}}} ) \in \Theta_{\alpha}, \quad
 \fa{\eta}=(\pa{\eta},\va{\eta}{i_1},\ldots,\va{\eta}{i_{d_{\alpha}}}) \in Y_{\alpha}, \quad
 \psi_{\alpha} \text{ and } \varphi_{\alpha}.
\end{equation}

In order to use these exponential families $\mathcal{E}_{\alpha}$ and $\mathcal{E}_i$ for LBP, we need an assumption.
\begin{asm}
[Marginally closed assumption] \label{asm:marginallyclosed}
For all pair of $i \in \alpha$,
\begin{equation}
 \int  p(x_{\alpha}) {\rm d} \nu_{\alpha \smallsetminus i} ( x_{\alpha \smallsetminus i})  
\in \mathcal{E}_i \quad \text{ for all } p \in \mathcal{E}_{\alpha}.
\end{equation}
\end{asm}

\begin{defn}
\label{defn:ifa}
A collection of the exponential families $\mathcal{I}:=\{\mathcal{E}_{\alpha}, \mathcal{E}_i \}$
given by sufficient statistics $( \pa{\phi}(x_{\alpha}), \phi_i (x_i) )_{ \alpha \in F,i \in V}$ as above,
satisfying Assumptions \ref{asm:expregular} and \ref{asm:marginallyclosed}
is called an {\it \ifa} associated with a hypergraph $H$.
A \ifa is called {\it pairwise} if the associated hypergraph is a graph.
\end{defn}
Inference model has a parameter set $\Theta=\prod_{\alpha} \Theta_{\alpha} \times \prod_i \Theta_i$, which is bijectively mapped
to the dual parameter set $Y=\prod_{\alpha} Y_{\alpha} \times \prod_i Y_i$ by the maps of respective components.

An inference model naturally defines an exponential family on $\mathcal{X}=\prod_i \mathcal{X}_i$
of the sufficient statistic $( \pa{\phi}(x_{\alpha}), \phi_i (x_i) )_{ \alpha \in F,i \in V}$.
This exponential family is called the {\it global exponential family} and denoted by $\mathcal{E}(\mathcal{I})$.

\begin{example}[Multinomial]
Let $\mathcal{E}_i$ be an exponential family of multinomial distributions.
Choosing functions  $\fa{\phi}(x_{\alpha})$ suitably,
we can make the $\mathcal{E}_{\alpha}$
being multinomial distributions on $\mathcal{X}_{\alpha}$.
Then the \ifa is called a {\it multinomial inference family}.
\end{example}

\begin{example}[Gaussian]
Let $\mathcal{X}_i = \mathbb{R}$.\footnote{Extension to high dimensional case, i.e. $\mathcal{X}_i=\mathbb{R}^{r_i}$, is straight forward.}
For Gaussian case,
the sufficient statistics are given by
\begin{equation}
 \phi_{i}(x_i)=(x_i,x_i^2), \qquad
 \pa{\phi}(x_{\alpha})=(x_i x_j)_{i , j \in \alpha, i \neq j}
\end{equation}
Then the \ifa $\mathcal{I}$ is called {\it Gaussian \ifa}.
Assumption \ref{asm:marginallyclosed} is satisfied because a marginal of a Gaussian distribution is a Gaussian distribution.
Fixed-mean cases are completely analogous.
Usually, $H$ is a graph rather than hypergraphs.
In this thesis, we only consider Gaussian \ifas on graphs, but extensions of our results to hypergraphs are straightforward.
\end{example}

\subsection{Basics of the LBP algorithm}\label{sec:basicLBP}
The LBP algorithm calculates approximate marginals of the given graphical model $\Psi = \{ \Psi_{\alpha}\}$
using a fixed \ifa $\mathcal{I}$. 
We always assume that the \ifa includes the given probability density function.
\begin{asm}
\label{asm:modelindludes}
For all factors $\alpha \in F$, there exists $\bar{\theta}_{\alpha}$ s.t.
\begin{equation}
 \Psi_{\alpha}(x_{\alpha}) = \exp \left( \inp{\bar{\theta}_{\alpha} }{ \phi_{\alpha} (x_{\alpha})}  \right) \label{eq:asm:modelindludes}
\end{equation}
\end{asm}
This assumption is equivalent to the assumption
\begin{equation}
p(x)=  \frac{1}{Z} \prod_{\alpha} \Psi_{\alpha}(x_{\alpha}) \in \mathcal{E}(\mathcal{I})
\end{equation}
up to trivial constant re-scalings of $\Psi_{\alpha}$, which do not affect the LBP algorithm.

The procedures of the LBP algorithm is as follows \cite{KFLfactor}.
For each pair of a vertex $i \in V$ and a factor $\alpha \in F$ satisfying $i \in \alpha$, 
an initialized message is given in a form of
\begin{equation}
 m_{\edai }^{0}(x_i) = \exp ( \inp{\mu_{\edai}^{0}}{\phi_i(x_i)} ), \label{messageform}
\end{equation}
where the choice of $\mu_{\edai}^{0}$ is arbitrary.
The set $\{  m_{\edai }^{0} \}$ or $\{ \mu_{\edai}^{0}\}$ is called an {\it initialization} of the LBP algorithm.
At each time $t$, the messages are updated by the following rule:
\begin{equation}
m^{t+1}_{\edai}(x_i)
=\omega 
\int
\Psi_{\alpha}(x_{\alpha})
\hspace{-1mm}
\prod_{j \in \alpha, j \neq i}
\prod_{\beta \ni j, \beta \neq \alpha}
\hspace{-1mm} m^{t}_{\edbj}(x_j)
\hspace{1mm} {\rm d}\nu_{\alpha \smallsetminus i}({x_{\alpha \smallsetminus i}})
\qquad (t \geq 0), \label{LBPupdate}
\end{equation}
where $\omega$ is a certain scaling constant.\footnote{
Here and below, we do not care about the integrability problem.
For multinomial and Gaussian cases, there are no problems.}
See Fig~\ref{fig:LBPupdate} for the illustration of this message update scheme.
From Assumptions \ref{asm:marginallyclosed} and \ref{asm:modelindludes},
messages can keep the form of Eq.~(\ref{messageform}).
\begin{figure}
\begin{center}
\includegraphics[scale=0.3]{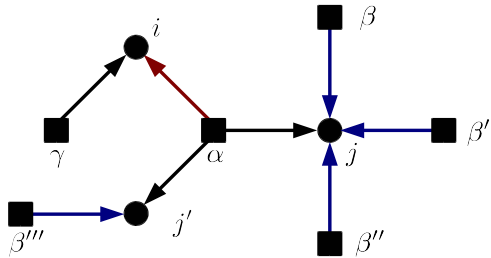}
\vspace{-2mm}
\caption{The blue messages contribute to the red message at the next time step. \label{fig:LBPupdate}}
\end{center}
\end{figure}

One may realize that Eq.~(\ref{LBPupdate}) looks slightly different from Eq.~(\ref{eq:bplbpupdate}),
which involves compatibility functions associated with vertices.
However, the compatibility functions associated with vertices can be included in that of factors
and such operations do not affect the LBP algorithm essentially.
Therefore, our treatment is general.

Since this update rule simultaneously generates all messages of time $t+1$ by that of 
time $t$, it is called a {\it parallel update}.
Another possibility of the update is a {\it sequential update}, where, at each time step, 
one message is chosen according to some prescribed or random order of directed edges. 
In this paper, we mainly discuss the parallel update.

We repeat the update Eq.~(\ref{LBPupdate}) until the messages converge to a fixed point,
though this procedure is not guaranteed to converge.
Indeed, it sometimes exhibits oscillatory behaviors \cite{MWJempiricalstudy}. 
The set of LBP fixed points does not depend on the choices of the update rule,
but converging behaviors, or {\it dynamics}, does depend on the choices.

If the algorithm converges, we obtain the fixed point messages $\{m^{*}_{\edai}\}$
and {\it beliefs}\footnote{
Belies are often defined for middle time messages $\{m^{t}_{\edai}\}$ by
Eqs.~(\ref{eq:defbelief1}) and (\ref{eq:defbelief2}).
However, in this thesis, beliefs are only defined by fixed point messages.
} that are defined by
\begin{align}
&b_{i}(x_i):= \omega
\prod_{\alpha \ni i} m_{\edai}^{*}(x_i) \label{eq:defbelief1}\\
&b_{\alpha}(x_{\alpha})
:=\omega 
\Psi_{\alpha}(x_{\alpha})
\prod_{j \in \alpha } 
\prod_{\beta \ni j, \beta \neq \alpha} m^{*}_{\edbj}(x_j),  \label{eq:defbelief2}
\end{align}
where $\omega$ denotes normalization constants requiring
\begin{equation}
 \int b_i(x_i) {\rm d} \nu_i =1 \quad \text{ and } \quad \int b_{\alpha}(x_{\alpha}) {\rm d} \nu_{\alpha}=1.
\end{equation} 
Note that beliefs automatically satisfy conditions $b_{\alpha}(x_{\alpha}) > 0$,
and 
\begin{equation}
\int b_{\alpha}(x_{\alpha}) {\rm d} \nu_{\alpha \smallsetminus i}({x_{\alpha \smallsetminus i}}) =b_i(x_i). \label{eq:localconsistency}
\end{equation}
Beliefs are used as approximation of the true marginal distributions $p_{\alpha}(x_{\alpha})$ and $p_i(x_i)$.
We will give the approximation of the partition function by LBP, called the Bethe approximation,
in the next section.

\subsection{BP on trees}\label{sec:LBPtree}
For the understanding of the LBP algorithm,
tree is always a good starting point.
Historically, the message update scheme of the algorithm is designed to calculate the exact marginals of tree-structured distributions
and called {\it Belief Propagation} (BP).
Here, we review the fact.
\begin{prop}
If $H$ is a tree, 
the LBP algorithm stops at most $|\vec{E}|$ updates and
the calculated beliefs are equal to the exact marginal distributions of $p$.
\end{prop}
\begin{proof}
We omit a detailed proof.
Basically, the assertion is checked by extending observations in Subsection \ref{sec:introLBPalgorithm}. 
\end{proof}

\subsection{LBP as a dynamical system}\label{sec:LBPstability}
At each time $t$, the state of the algorithm is specified by the set of 
messages $\{ m_{\edai}^t \}$, which is identified with 
its \nparas $\bsmu^t=\{ \mu_{\edai}^t \} \in \mathbb{R}^{\vec{E}}$.
In terms of the parameters, the update rule Eq.~(\ref{LBPupdate}) is written as follows.
\begin{equation}
\label{eq:eparaLBPupdate}
 \mu_{\edai}^{{t+1}}=
\Lambda_{i}^{-1}
\Big(
   \Lambda_{\alpha} ( \pa{\bar{\theta}},
                              \va{\bar{\theta}}{i_1}  + \hspace{-2mm}  \sum_{\beta \in N_{i_1} \smallsetminus \alpha}\hspace{-2mm} \mu^t_{\edbione},     
                      \ldots, \va{\bar{\theta}}{i_k} + \hspace{-2mm}  \sum_{\beta \in N_{i_k} \smallsetminus \alpha}\hspace{-2mm} \mu^t_{\edbik}
                    )_i
\Big)
- \sum_{\gamma \in N_i \smallsetminus \alpha} \hspace{-2mm} \mu^t_{\edgi},
\end{equation}
where $\alpha = \fai$, $d_{\alpha}=k$ and
$\Lambda_{\alpha}(\cdots)_i$ is the $i$-th component ($i \in \alpha$).
To obtain this equation, multiply Eq.~(\ref{LBPupdate}) by
\begin{equation*}
 \prod_{\gamma \in N_i \smallsetminus \alpha} m^t_{\edgi}(x_i)
\end{equation*}
and normalize it to be a probability distribution.
Then take the expectation of $\phi_i$.

The update rule can be viewed as a transform $T$ on the set of \nparas of messages $M$.
Formally,
\begin{equation*}
 T:  M \longrightarrow  M, \qquad
 \bsmu^t = T( \bsmu^{t-1} ).
\end{equation*}
In this formulation, the fixed points of LBP are $\{ \bsmu^* \in M |  \bsmu^*= T( \bsmu^* )\}$.

In order to get familiar with the computation techniques, here we compute the 
differentiation of the update map $T$ around an LBP fixed point.
This expression derived in \cite{ITAinfo,ITAinfo} for the cases of turbo and LDPC codes. %
\begin{thm}
[Differentiation of the LBP update]
\label{thm:diffofLBP}
At an LBP fixed point, the differentiation (linearization) of the LBP update is
\begin{equation}
 \pd{T(\bsmu)_{\edai}}{\mu_{\edbj}}= 
\begin{cases}
\var{b_i}{\phi_i}^{-1} \cov{b_{\alpha}}{\phi_i}{\phi_j}  
&\text{ if } j \in N_{\alpha} \smallsetminus i \text{ and } \beta \in N_j \smallsetminus \alpha, \\
0 &\text{ otherwise.}
\end{cases}
\end{equation}
\end{thm}
\begin{proof}
First, consider the case that $j \in N_{\alpha} \smallsetminus i \text{ and } \beta \in N_j \smallsetminus \alpha$.
The derivative is equal to
\begin{equation}
 \pd{\Lambda^{-1}_i}{\eta_i}   \pd{(\Lambda^{}_{\alpha})_i}{\theta^{\alpha}_j}    
=
\var{b_i}{\phi_i}^{-1} \cov{b_{\alpha}}{\phi_i}{\phi_j} .
\end{equation}
Another case is $i=j$ and $ \alpha , \beta \in N_i ~(\alpha \neq \beta)$.
Then, the derivative is 
\begin{equation}
 \pd{\Lambda^{-1}_i}{\eta_i}   \pd{(\Lambda^{}_{\alpha})_i}{\theta^{\alpha}_i}-I    
=
0
\end{equation}
because $\var{b_i}{\phi_i}= \var{b_{\alpha}}{\phi_i}$ from Eq.~(\ref{eq:localconsistency}).
In other cases, the derivative is trivially zero.
\end{proof}

The relation $j \in N_{\alpha} \smallsetminus i \text{ and } \beta \in N_j \smallsetminus \alpha$
will be written as $(\edbj) \rightharpoonup (\edai)$ in Subsection \ref{sec:defgraphzeta}.
We will discuss the relations between the differentiation $T$
and stability properties of the LBP algorithm in Section \ref{sec:stability}.

It is noteworthy that the elements of the linearization matrix is 
explicitly expressed by the fixed point beliefs.

\section{Bethe free energy}\label{sec:BFE}
The Bethe approximation was initiated in the paper of Bethe \cite{Bethe}
to analyze physical phases of two atom alloy.
Roughly speaking, the Bethe approximation captures short range fluctuations
computing states in small clusters in a consistent manner.
The Bethe approximation is known to be exact for distributions on tree-structured graphs.
The modern formulation for presenting the approximation is a variational problem
of the {\it Bethe free energy} \cite{Anote}.
In the end of this section, we see that this approximation
is equivalent to the LBP algorithm.
This relation was first clearly formulated by Yedidia et al \cite{YFWGBP}.

In this section, we introduce two types of Bethe free energy functions,
both of them yield variational characterization of the Bethe approximation. %
These two functions are basically similar and have the same values on the points corresponding to the LBP fixed points.
The first type is essentially utilized to show the equivalence
of the Bethe approximation and LBP by Yedidia et al \cite{YFWGBP}. 
In \cite{ITAsto}, Ikeda et al discusses relations between these two types of \Bfe functions
on a constrained set.%

\subsection{Gibbs free energy function}
First, we should introduce the Gibbs free energy function
because the Bethe free energy function is a computationally tractable approximation of the Gibbs free energy function.
For given graphical model $\Psi=\{ \Psi_{\alpha}\}$, the {\it Gibbs free energy} $F_{Gibbs}$ is 
a function over the set of probability distributions $\hat{p}$ on $x=(x_i)_{i \in V}$ defined by 
\begin{equation}
 F_{Gibbs}(\hat{p})= \int \hat{p}(x) \log \left( 
\frac{\hat{p}(x)}{\prod_{\alpha} \Psi_{\alpha}(x_{\alpha})}
\right) {\rm d}\nu(x), \label{def:GibbsFE}
\end{equation}
where $\nu = \prod_{i \in V} \nu_i$ is the base measure on $\mathcal{X}= \prod_{i \in V} \mathcal{X}_i$.
Since $y \log y$ is a convex function of $y$,
$F_{Gibbs}$ is a convex function with respect to $\hat{p}$.
Using Kullback-Leibler divergence $D(q||p)=\int\hat{p}\log(q/p)$,
Eq.~(\ref{def:GibbsFE}) comes to
\begin{equation}
 F_{Gibbs}(\hat{p})= D(\hat{p}||p) - \log Z.
\end{equation}
Therefore, the exact distribution Eq.~(\ref{defp})
is characterized by a variational problem
\begin{align}
 p(x)= \argmin_{\hat{p}} F_{Gibbs}( \hat{p} ),
\end{align}
where the minimum is taken over
all probability distributions on $x$.
As suggested from the name of ``free energy,''
the minimum value of this function is equal to $- \log Z$.

From the Assumption \ref{asm:modelindludes},
$p$ is in the global exponential family $\mathcal{E}(\mathcal{I})$.
Therefore, it is possible to restrict the range of the minimization within $\mathcal{E}(\mathcal{I})$
without changing the outcome of the minimization.

\subsection{\Bfe function}
\label{sec:twoBfes}
At least for discrete variable case, computing values of the Gibbs free energy function is intractable in general
because the integral in Eq.~(\ref{def:GibbsFE}) is indeed a sum of $|\mathcal{X}|= \prod_i |\mathcal{X}_i|$ states.
We introduce functions called \Bfe that does not include such exponential number of state sum.

There are two types of \Bfe functions;
the type 1 is defined on an affine subspace of \eparas
whereas the second type is defined on an affine subspace of \nparas.
In information geometry, such subspaces are called
{\it m-affine space} and {\it e-affine space}, respectively \cite{ANmethods}.

\subsubsection{Type 1}
\begin{defn}
The type 1 \Bfe function is a function of \eparas.
For a given \ifa $\mathcal{I}$,
a set $L(\mathcal{I})$\footnote{For multinomial cases, the closure of this set is called {\it local polytope} \cite{WJgraphical,WJvariational}.} 
is defined by
$L(\mathcal{I}) = \{ \bs{\eta}=\{\fa{\eta},\eta_{i}\}; \va{\eta}{i}=\eta_i \} $.
On this set, the \Bfe function is defined by
\begin{equation}
 F(\bs{\eta}):=
-\sum_{\alpha \in F} \inp{ \fa{\bar{\theta}} }{ \fa{\eta} } +
\sum_{\alpha \in F}\varphi_{\alpha}(\fa{\eta}) +
\sum_{i \in V} (1-d_i)\varphi_i(\eta_i), \label{defn:Bfe}
\end{equation}
where $\fa{\bar{\theta}}$ is the \npara of $\Psi_{\alpha}$.
\end{defn}
Since $Y_i$ and $Y_{\alpha}$ are open convex, $L$ is a relatively open convex set.
This function is computationally tractable
because it is a sum of the order $O(|F|+|V|)$ terms, assuming
the functions $\varphi_{\alpha}$ and $\varphi_i$ are tractable. 

An element of $L$ is called a set of {\it pseudomarginals}. %
The pseudomarginals can be identified with a set of functions $\beliefs$
that satisfies
\begin{enumerate}
 \item $b_{\alpha}(x_{\alpha}) \in \mathcal{E}_{\alpha}$, \label{pseudo1}
 \item $\int b_{\alpha}(x_{\alpha}) {\rm d} \nu_{\alpha \smallsetminus i} = b_i(x_i)$. \label{pseudo2} 
\end{enumerate}
The second condition is called {\it local consistency}.
Under this identification, the \Bfe function is
\begin{align}
 F(\beliefsw)= -\sum_{\alpha \in {F}} \int 
 b_{\alpha}(x_{\alpha})\log\Psi_{\alpha}(x_{\alpha})   {\rm d}\nu_{\alpha}
& + \sum_{\alpha \in {F}} \int  b_{\alpha}(x_{\alpha})\log b_{\alpha}(x_{\alpha}) {\rm d} \nu_{\alpha} \nonumber \\
&+ \sum_{i \in V}(1-d_i) \int b_i(x_i)\log b_i(x_i) {\rm d}\nu_i.
\end{align}

\begin{example}[Multinomial \ifa]
Let $\mathcal{I}$ be a multinomial \ifa.
The local polytope is given by
\begin{small}
\begin{equation*}
L=\{ \{b_{\alpha}, b_i\}_{\alpha \in F, i \in V}|
\quad
b_{\alpha}(x_{\alpha}) > 0,
\quad
\sum_{x_{\alpha}}b_{\alpha}(x_{\alpha})=1,
\quad
\sum_{x_{\alpha \smallsetminus i }}b_{\alpha}(x_{\alpha})=b_i(x_i) \quad {}^{\forall} i \in \alpha
\}.
\end{equation*}
\end{small}
The Bethe free energy function is 
\begin{small}
\begin{equation}
 F= -\sum_{\alpha \in {F}}\sum_{x_{\alpha}}
 b_{\alpha}(x_{\alpha})\log\Psi_{\alpha}(x_{\alpha})   
 + \sum_{\alpha \in {F}}\sum_{x_{\alpha}} b_{\alpha}(x_{\alpha})\log
 b_{\alpha}(x_{\alpha})  
+ \sum_{i \in V}(1-d_i)\sum_{x_i}b_i(x_i)\log b_i(x_i). \label{eq:BFEmultinomial}
\end{equation}
\end{small}
\end{example}

In order to see the relation between the Bethe free energy function and the Gibbs free energy function
we construct a map from the domain of the \Bfe as follows:
\begin{equation} 
 \Pi ( \beliefs ):=  \prod_{\alpha}b_{\alpha}(x_{\alpha}) \prod_{i} b_i(x_i)^{1-d_i}.  
\end{equation}
The following fact gives an insight that
the Bethe free energy function approximate the Gibbs free energy function.
\begin{prop}
 If $H$ is a tree, $\Pi$ is a bijective map from $L$ to $\mathcal{E}(\mathcal{I})$.
The inverse map is obtained by the marginals of $p \in \mathcal{E}(\mathcal{I})$.
Under this map, the Gibbs free energy function coincide with the Bethe free energy function: $F=F_{Gibbs} \circ \Pi$.
\end{prop}
\label{prop:Fexact1}
\begin{proof}
 Since $H$ is a tree, one easily observes that $\sum_{x }\prod_{\alpha}b_{\alpha} \prod_{i} b_i^{1-d_i}=1$. 
This implies $\Pi(\bsb) \in \mathcal{E}(\mathcal{I})$ ($\bsb=\beliefsw$).
The injectivity of $\Pi$ is obvious because the marginals of $\Pi(\bsb)$ are $\beliefsw$.
For given $p \in \mathcal{E}(\mathcal{I})$, let $\bs{p}= \{p_{\alpha}, p_i\}$ be the set of marginal distributions.
We see that $\Pi(\bs{p})=p$ because the \eparas $\{\pa{\eta}, \eta_i \}$ of the global exponential family are equal.
Thus the first part of the assertion is proved.

Next, we check that $F=F_{Gibbs} \circ \Pi$.
Since the marginals of $\Pi(\bsb)$ are $\beliefsw$, we obtain
\begin{align*}
 F_{Gibbs} \circ \Pi (\bsb)
&=
-\sum_{\alpha \in {F}} \int \Pi(\bsb) \log\Psi_{\alpha}(x_{\alpha})   {\rm d}\nu(x)
 + \sum_{\alpha \in {F}} \int \Pi(\bsb) \log b_{\alpha}(x_{\alpha})  {\rm d}\nu(x) \\
& \quad + \sum_{i \in V}(1-d_i) \int \Pi(\bsb) \log b_i(x_i)  {\rm d}\nu(x) \\
&=
F(\bsb).
\end{align*}
\end{proof}
For general factor graphs, $\Pi(\bsb)$ is not necessarily normalized.
This property is related to the approximation error of the partition function
(See Lemma \ref{lem:ZZB} for details).
Note also that, for general factor graphs,
marginal distributions of an element in $\mathcal{E}(\mathcal{I})$ are not necessarily
elements in $\mathcal{E}_{\alpha}$ or $\mathcal{E}_i$.
However, for multinomial and (fixed-mean) Gaussian \ifas,
it is the case even if $H$ is not a tree.

Though the Bethe free energy function $F$ approximates the convex function $F_{Gibbs}$,
it is not necessarily convex nor has unique minima.
Though functions $\varphi_{\alpha}$ and $\varphi_i$ are convex,
the negative coefficients $(1-d_i)$ makes the function complex.
In general, the convexity of $F$ is broken as the nullity of the underlying hypergraph grows.
The positive-definiteness of the Hessian of the Bethe free energy will be analyzed in
Section \ref{sec:PDC} using the \Bzf.

\subsubsection{Type 2}
The second type of \Bfe function is a function of natural parameters.

\begin{defn}
Define an affine space of \nparas by
\begin{equation*}
 A (\mathcal{I},\Psi):= 
\{\bstheta=\{\fa{\theta},\theta_i \} | \pa{\theta} = \pa{\bar{\theta}} ~{}^{\forall} \alpha \in F,
\ \sum_{\alpha \ni i} \va{\bar{\theta}}{i} =(1-d_i)\theta_i+\sum_{\alpha \ni i} \va{\theta}{i} 
~{}^{\forall} i \in \alpha \}.
\end{equation*}
The type 2 Bethe free energy function\footnote{
In this thesis, we mean ``\Bfe function'' by the type 1 unless otherwise stated. 
}
$\mathcal{F}$ is a function on $A (\mathcal{I},\Psi)$ defined by
\begin{equation}
 \mathcal{F}(\bstheta)= -\sum_{\alpha \in F} \psi_{\alpha}(\fa{\theta})-\sum_{i \in V}(1-d_i) \psi_{i}(\theta_i).  \label{eq:defn:type2BFE}
\end{equation}
\end{defn}
Note that $\mathcal{F}$ itself does not depend on the given distribution $\Psi$
in contrast to $F$.
Note also that $L$ and $A$ are subsets of the same set $Y \simeq \Theta$,
where the identification is given by the map $\prod_{\alpha} \Lambda_{\alpha} \times \prod_{i} \Lambda_i$.
As we see in the next subsection, the values of $F$ and $\mathcal{F}$ coincide
at intersections of $L$ and $A$.

\subsection{Characterizations of the LBP fixed points}
\label{sec:LBPcharacterizations}
We present several characterization of LBP fixed points.
As we will discuss in Section \ref{sec:Bzfintro},
this presentation gives intuitive understanding of the \Bzf. 
For the characterizations, we use a formal definition of {\it beliefs}.
We will see that it is the same thing given in Subsection \ref{sec:basicLBP}, after knowing the
result of the Theorem \ref{thm:LBPcharacterizations}.

\begin{defn}
For given \ifa $\mathcal{I}$ and graphical model $\Psi=\{\Psi_{\alpha}\}$,
A set of {\it beliefs} $\beliefs$ is a set of pseudomarginals that satisfies 
\begin{equation}
 \prod_{\alpha}b_{\alpha}(x_{\alpha}) \prod_{i} b_i(x_i)^{1-d_i} \propto \prod_{\alpha} \Psi_{\alpha}(x_{\alpha}).  \label{eq:productcondition}
\end{equation}
\end{defn}

\begin{thm}
 \label{thm:LBPcharacterizations}
Let $\mathcal{I}$ be a \ifa and $\Psi=\{\Psi_{\alpha}\}$ be a graphical model.
The following sets are naturally identified each other.
\begin{enumerate}
 \item The set of fixed points of \lbp.
 \item The set of the beliefs.
 \item The set of intersections of $L(\mathcal{I})$ and $A(\mathcal{I},\Psi)$.
 \item The set of stationary points of $F$ over $L(\mathcal{I})$.
 \item The set of stationary points of $\mathcal{F}$ over $A(\mathcal{I},\Psi)$.
\end{enumerate}
Furthermore, for an LBP fixed point $\thetasw$, the corresponding beliefs are given by
\begin{align}
 &b_i(x_i)= \exp ( \inp{\theta_i}{\phi_i (x_i )} - \psi_{i}(\theta_i)),  \label{eq:beliefbyepara1} \\
 &b_{\alpha}(x_{\alpha})=  \exp ( \inp{\fa{\theta}}{\phi_{\alpha} (x_{\alpha})} - \psi_{\alpha}(\fa{\theta})). \label{eq:beliefbyepara2}
\end{align}
\end{thm}
\begin{proof}
{\bf 2 $\Leftrightarrow$ 3:}
Since a set of beliefs is a set of pseudomarginals,
the beliefs is identified with $\thetasw$ that satisfies the local consistency conditions and Eq.~(\ref{eq:productcondition}). 
These conditions are obviously equivalent to
the constraints of $L(\mathcal{I})$ and $A(\mathcal{I},\Psi)$ respectively. \\
{\bf 1 $\Leftrightarrow$ 2,3:}
For a given LBP fixed point messages, a belief is given by Eqs.~(\ref{eq:defbelief1}) and (\ref{eq:defbelief2}).
By definition, it is easy to check that Eq.~(\ref{eq:productcondition}) holds.
For the converse direction, we define the messages by the beliefs by
\begin{equation}
 m_{\edai}(x_i) = \exp ( \inp{\theta_i + \va{\bar{\theta}}{i} - \va{\theta}{i}}{\phi_i} ).
\end{equation} 
From the constraints of $L(\mathcal{I})$ and $A(\mathcal{I},\Psi)$, one observes that
\begin{align*}
& \prod_{\beta \in N_i} m_{\edbi}(x_i) = \exp ( \inp{\theta_i}{\phi_i (x_i )} ) \propto b_i(x_i),\\
& \Psi_{\alpha} (x_{\alpha}) \prod_{i \in \alpha} \prod_{\beta \in N_i \smallsetminus \alpha} \hspace{-2mm} m_{\edbi}(x_i)
= \exp ( \inp{ \pa{\theta} }{\pa{\phi}(x_{\alpha}) } +  \sum_{i \in N_{\alpha}} \inp{\va{\theta}{i} }{\phi_i } )
\propto b_{\alpha}(x_{\alpha}).
\end{align*}
Therefore, the local consistency condition implies that
\begin{equation}
  \prod_{\beta \in N_i} m_{\edbi}(x_i) \propto \int \Psi_{\alpha} 
\prod_{j \in \alpha} \prod_{\beta \in N_j \smallsetminus \alpha} m_{\edbj}(x_j) {\rm d} \nu_{\alpha \smallsetminus i}.
\end{equation}
This is obviously equivalent to the LBP fixed point equation.\\
{\bf 3 $\Leftrightarrow$ 4:}
A point in $L(\mathcal{E})$ is identified with $\{ \pa{\eta}, \eta_i \}_{\alpha \in F, i \in V}$.
Taking derivatives of $F$, we see that stationary point conditions are $\pa{\theta}= \pa{\bar{\theta}}$ 
and $\sum_{\alpha \ni i} \va{\bar{\theta}}{i} = \sum_{\alpha \ni i} \va{\theta}{i} + (1-d_i) \theta_i$\\
{\bf 3 $\Leftrightarrow$ 5:}
This equivalence is also checked by taking derivatives of $\mathcal{F}$ in $A(\mathcal{E},\Psi)$.
\end{proof}

The condition 3 is an alternative exposition of the characterization of the LBP fixed points given 
by Ikeda et al \cite{ITAsto}.
In the paper, the fixed points are characterized by ``e-condition'' and ``m-condition,''
which partly correspond to the constraint of $A$ and $L$ respectively.
Wainwright et al \cite{Wrepara} derives another characterization of the LBP fixed points utilizing spanning trees of the graph.
Their redundant representation of exponential families is related to our \nparas of exponential families in the \ifa. 

Finally, we check that the values of $F$ and $\mathcal{F}$ are equal.
Actually, using the third characterization,
\begin{align*}
 F(\bseta)-\mathcal{F}(\bstheta)
&=
-\sum_{\alpha \in F} \inp{ \fa{\bar{\theta}} }{ \fa{\eta} } +
\sum_{\alpha \in F}(\varphi_{\alpha}(\fa{\eta})+ \psi_{\alpha}(\fa{\theta} ) )+
\sum_{i \in V} (1-d_i)(\varphi_i(\eta_i)+\psi_i(\theta_i) ) \\
&=
-\sum_{\alpha \in F}  \inp{\fa{\bar{\theta}} }{ \fa{\eta}} +
\sum_{\alpha \in F}  \inp{ \fa{\theta} }{ \fa{\eta} }
+\sum_{i \in V} (1-d_i) \inp{ {\theta}_i}{ \eta_i}  \\
&=0.
\end{align*}

\begin{defn}
\label{defn:BFE}
For given LBP fixed point, the {\it Bethe approximation} $Z_B$ of the partition function $Z$ is defined by
\begin{equation}
 - \log Z_{B}= F(\bstheta)= \mathcal{F}(\bseta).
\end{equation}
\end{defn}

\subsection{Additional remarks}
\label{sec:PreAdd}

\subsubsection{Extensions and variants of LBP}
Generalizing the \Bfe function,
the approximation method can be further extended to Cluster Variational Method (CVM) \cite{Kikuchi},
which leads a message passing algorithm called generalized belief propagation \cite{YFWGBP,YFWconstructing}.
Another generalization of the \Bfe function is proposed in \cite{WHfractional}.
The derived message passing algorithm is called fractional belief propagation.
Expectation propagation, introduced in \cite{Mexpectationpropagation},
is derived by easing the local consistency condition to consistency of expectations 
called weak consistency \cite{HOWWZapproximate}.
All these message passing algorithm include the LBP algorithm as a special case.

The fourth condition in Theorem \ref{thm:LBPcharacterizations} says that
the LBP algorithm finds a stationary point of the \Bfe function $F$.
This viewpoint motivates direct optimization approaches to the \Bfe function.
Welling and Teh \cite{WTbelief} have derived an iterative algorithm that decrease the \Bfe function
at each step.
Yuille \cite{Ycccp} also developed CCCP algorithm for the optimization.
One advantage of these algorithm is that they are guaranteed to converge to an LBP fixed point.

\subsubsection{Related algorithms}
Max-product algorithm is a similar algorithm to LBP algorithm.
It is obtained by replacing the sum operator in the LBP update Eq.~(\ref{LBPupdate})
with the max operator. From arithmetic laws satisfied by max and product,
the max-product algorithm is defined parallel to LBP \cite{AMdistributive}.
It is also obtained as a ``zero temperature limit'' of LBP algorithm.

%% file: chapter3.tex

\section{Introduction}
Zeta functions, such as Riemann, Weil and Selberg types, appear in many fields of mathematics.
In 1966, Y.~Ihara introduced an analogue of Selberg zeta function
and proved its rationality establishing a determinant formula \cite{Idiscrete}.
Though his zeta function was associated to a certain algebraic object,
it was abstracted and extended to be defined on arbitrary finite graphs
by works of J.~P.~Serre \cite{Strees}, Sunada \cite{SL-functions} and Bass \cite{Bass}.
This zeta function is referred to as the {\it Ihara zeta function}. 
There are some generalization of the Ihara zeta function.
{\it The edge zeta function} is a multi-variable generalization of the Ihara zeta function,
allowing arbitrary scalar weight for each directed edge \cite{STzeta1}.
{\it L-function} is also an extension using a finite dimensional unitary representation
of the fundamental group of the graph.
Another direction of an extension is the zeta function of hypergraphs \cite{Shypergraph}.

In this chapter, unifying these generalizations,
we introduce a graph zeta function defined on hypergraphs with matrix weights. 
We show an \IB type determinant formula %
based on a simple determinant operations (Proposition \ref{app:detdet}).
This formula plays an important role in establishing the relations between this zeta and LBP in the next chapter.

The remainder of this chapter is organized as follows.
In Section \ref{sec:basiczeta}, we provide the definition of our graph zeta function
as well as necessary definitions of hypergraphs such as prime cycles.
In Section \ref{sec:detIhara}, we show the \IB type determinant formula,
requiring additional structure on the matrix weights.
Miscellaneous properties of one-variable hypergraph zeta is discussed in Section \ref{sec:miszeta}.
We conclude in Section \ref{sec:discussion:zeta} with a summary
and discussion of the role of these results in the remainder of the thesis.

\section{Basics of the graph zeta function}\label{sec:basiczeta}
\subsection{Definition of the graph zeta function}
\label{sec:defgraphzeta}
In the first part of this subsection, in addition to Subsection \ref{sec:basicsgraph}, we introduce basic definitions and notations
of hypergraphs required for the definition of our graph zeta function.

Let $H=(V,F)$ be a hypergraph.
As commented in Subsection \ref{sec:basicsgraph},
it is also denoted by $\fgdefn$.
For each edge $e=(\edai) \in \vec{E}$,
$s(e)=\alpha \in F$ is the 
{\it starting factor} of $e$ and
$t(e)=i \in V$ is the {\it terminus vertex} of $e$. 
If two edges $e,e'\in \vec{E}$ satisfy conditions
$t(e) \in s(e')$ and $t(e) \neq t(e')$, this pair is denoted by $\ete{e}{e'}$.
(See Figure \ref{fig:edgerelation}.)
A sequence of edges $(e_1,\ldots,e_k)$ is said to be a 
{\it closed geodesic} if $\ete{e_l}{e_{l+1}}$
for $l \in \mathbb{Z}/k\mathbb{Z}$. 
For a closed geodesic $c$, we may form the {\it m-multiple} $c^{m}$ by repeating $c$ $m$-times. 
If $c$ is not a multiple of strictly shorter closed geodesic, $c$ is said to be {\it prime}. 
Two closed geodesics are said to be {\it equivalent} if one is obtained by cyclic
permutation of the other.
An equivalence class of a prime closed geodesic is called a {\it prime cycle}.
The set of prime cycles of $H$ is denoted by $\mathfrak{P}_H$.

\begin{figure}
\begin{center}
\includegraphics[scale=0.25]{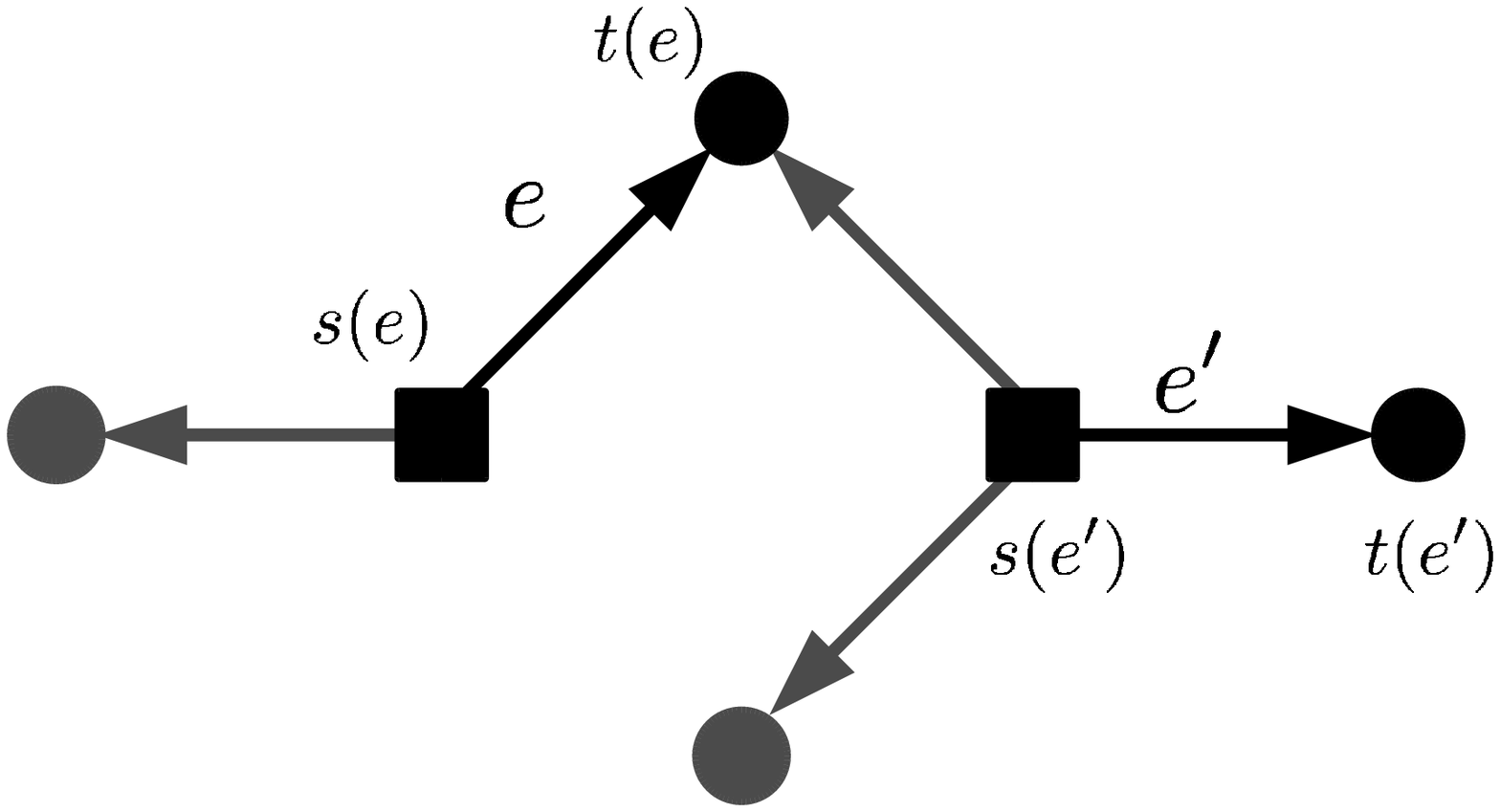}
\vspace{-1mm}
\caption{The relation $\ete{e}{e'}$. \label{fig:edgerelation}}
\end{center}
\end{figure}

If $H$ is a graph (i.e. $d_{\alpha}=2$ for all $\alpha \in F$),
these definitions reduce to standard definitions \cite{KSzeta}.
(We will explicitly give them in Subsection \ref{sec:specialIB}.)
In this case, a factor $\alpha=\{i,j\}$ is identified with an undirected edge $ij$, and
$(\edaj), (\edai)$ are identified with $(\edij), (\edji)$ respectively.

Usually, in graph theory, Ihara's graph zeta function is a univariate function and associated with a graph.
Our graph zeta, which is needed for the subsequent development of this thesis, is much more complicated.
It is defined on a hypergraph having weights of matrices. 
To define matrix weights, we have to prescribe its sizes;
we associate a positive integer $r_e$ with each edge $e \in \vec{E}$.
Note that the set of functions on $\vec{E}$ that take values on $\mathbb{C}^{r_e}$
for each $e \in \vec{E}$ is denoted by $\vfe$.
Note also that the set of $n_1 \times n_2$ complex matrices is denoted by $\mat{n_1}{n_2}$. 

\begin{defn}
For each pair of $\etea$, a matrix weight $u_{\etea} \in \mat{r_e}{r_{e'}}$ is associated.
For this given matrix weights $\bs{u}=\{u_{\etea}\}$,
the graph zeta function of $H$ is defined by
\begin{equation}
 \zeta_{H}(\bs{u}):=
\prod_{\mathfrak{p} \in \mathfrak{P}_H}
\frac{1}{\det \big( I- \pi(\mathfrak{p}) \big) },
\end{equation}
where $\pi(\mathfrak{p}):=$
$u_{e_k \rightharpoonup e_1}\ldots u_{e_2 \rightharpoonup e_3}  u_{e_1 \rightharpoonup e_2}$
for $\mathfrak{p}=(e_1,\ldots,e_k)$.
\end{defn}
Since $\det(I_n-AB)=\det(I_m-BA)$ for $n \times m $ and $m \times n$ matrices $A$ and $B$,
$\det( I- \pi(\mathfrak{p}))$ is well defined for the equivalence class $\mathfrak{p}$.
Rigorously speaking, we have to care about the convergence;
we should restrict the definition for sufficiently small matrix weights $\bsu$.
However, as we will discuss in the next subsection, the zeta function have analytical continuation
to the whole space of matrix weights.

If $H$ is a graph and $r_{e}=1$ for all $e \in \vec{E}$,
this zeta function reduces to the edge zeta function \cite{STzeta1}.
Furthermore, if all these scalar weights are set to be equal, i.e. $u_{\etea}=u$, 
the zeta function reduces to the Ihara zeta function.
On the other hand, for general hypergraphs,
we obtain the one-variable hypergraph zeta function 
by setting all matrix weights to be the same scalar $u$ \cite{Shypergraph}.
These reductions will be discussed in Subsection \ref{sec:specialIB}.

\begin{example} \label{example1}
$\zeta_{H}(\boldsymbol{u})^{}=1$ if and only if $H$ is a tree.
(See Proposition \ref{prop:treeprime} is Subsection \ref{sec:primecycle}.)
For 1-cycle graph $C_N$ of length
$N$, the prime cycles are $(e_1,e_2,\ldots,e_N)$ and
$(\bar{e}_N,\bar{e}_{N-1},\ldots,\bar{e}_1)$. (See Figure \ref{fig:prime1cycle}.)
The zeta function is
\begin{small}
\begin{equation*}
 \zeta_{C_N}(\boldsymbol{u})=
\det(I_{r_{e_1}}- u_{e_N \rightharpoonup e_1}\ldots u_{e_2 \rightharpoonup e_3}  u_{e_1 \rightharpoonup e_2}        )^{-1}
\det(I_{r_{\bar{e}_N}}- u_{\bar{e}_1 \rightharpoonup \bar{e}_N}\ldots u_{\bar{e}_{N-1} \rightharpoonup \bar{e}_{N-2}}  u_{\bar{e}_N \rightharpoonup \bar{e}_{N-1}}  )^{-1}.
\end{equation*} 
\end{small}
\end{example}

Except for the above two types of hypergraphs, 
the number of prime cycles is infinite.

\begin{figure}
\begin{center}
\includegraphics[scale=0.25]{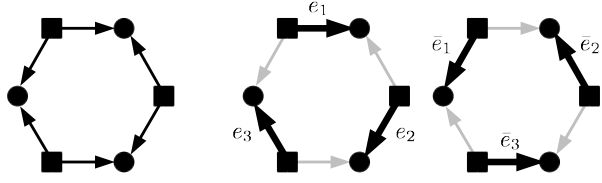}
\vspace{-1mm}
\caption{$C_3$ and its prime cycles. \label{fig:prime1cycle}}
\end{center}
\end{figure}

\subsection{The first determinant formula}
The following determinant formula gives analytical continuation to the whole strength of matrix weights. 

\begin{thm}
[The first determinant formula of zeta function]
We define a linear operator $\matmu : \vfe \rightarrow \vfe$ by
\begin{equation}
 \matmu f (e) =  \sum_{e':\etea} u_{\etea}f(e') \qquad f \in \vfe.
\end{equation}
Then, the following formula holds
\begin{equation}
  \zeta_{H}(\bs{u})^{-1}=
\det(I- \matmu).
\end{equation}  
\end{thm}
Note that matrix representation of the operator $\matmu$ is 
\begin{equation}
\matmu_{e,e'}=
\begin{cases}
u_{\etea}  \qquad \text{if } \etea \\
0          \qquad \qquad \text{otherwise.} \label{def:matmu}
\end{cases}
\end{equation}
The simplification of this matrix (i.e. on a graph, $r_e=1$, $u=1$)
is called {\it directed edge matrix} in \cite{STzeta1}
or  {\it Perron-Frobenius operator} in \cite{KSzeta}.
A noteworthy difference, in our and their definitions, is that directions of edges are opposite,
because we choose directions to be consistent with illustrations of the LBP algorithm.

The following proof proceeds in an analogous manner with Theorem 3 in \cite{STzeta1}. 
It is also possible to use Amitsur's theorem \cite{Acharacteristic} as in \cite{Bcounting}. 
\begin{proof}
First define a differential operator 
\begin{equation}
 \mathcal{H}:=
\sum_{\etea}\sum_{a_{e}, a_{e'}} (u_{\etea})_{a_{e}, a_{e'}}\pd{}{ (u_{\etea})_{a_{e}, a_{e'}} }
\end{equation}
where $(u_{\etea})_{a_{e}, a_{e'}}$ denotes the $(a_{e}, a_{e'})$ element of the matrix $u_{\etea}$.
If we apply this operator to a $k$ product of $u$ terms, it is multiplied by $k$.
Since  $\log \zeta_{H}(\bs{0})=0$ and $\log \det (I -\mathcal{M}(\bs{0}))^{-1}=0$,
it is enough to prove that $\mathcal{H}\log \zeta_{H}(\bsu)=\mathcal{H} \log \det (I -\matmu)^{-1}$.
Using equations $\log \det X = \tr \log X$ and 
$- \log (1-x)=\sum_{k \geq 1} \frac{1}{k}x^k$, we have
\begin{align}
 \mathcal{H}\log \zeta_{H}(\bsu)
&=\mathcal{H} \sum_{ \mathfrak{p} \in \mathfrak{P}_H }
- \log \det ( I - \pi(\mathfrak{p}) )  \nonumber \\ 
&=\mathcal{H} \sum_{ \mathfrak{p} \in \mathfrak{P}_H }
\sum_{k \geq 1} \frac{1}{k} \tr ( \pi(\mathfrak{p})^{k}) \label{eq:thm:det3}\\
&=\sum_{ \mathfrak{p} \in \mathfrak{P}_H } 
\sum_{k \geq 1}  |\mathfrak{p}| \tr ( \pi(\mathfrak{p})^{k}) \label{eq:thm:det4}\\
&=\sum_{C: \text{closed geodesic} } \tr ( \pi(C)) 
\quad = \sum_{k \geq 1}  \tr ( \matmu^{k}). \nonumber
\end{align}
From Eq.~(\ref{eq:thm:det3}) to Eq.~(\ref{eq:thm:det4}), notice that 
$\mathcal{H}$ acts as a multiplication of $k|\mathfrak{p}|$ for each summand. 
This is because the summand is a sum of degree  $k|\mathfrak{p}|$ terms counting each $(u_{\etea})_{a_{e}, a_{e'}}$ degree one.
From Eq.~(\ref{eq:thm:det4}) to the next equation, we used a property of closed geodesic:
it is uniquely represented as a repeat of the minimal period.

On the other hand, one easily observes that
\begin{align*}
 \mathcal{H} \log \det (I -\matmu)^{-1}
&= \mathcal{H} \sum_{k \geq 1} \frac{1}{k} \tr ( \matmu^{k}) \\
&= \sum_{k \geq 1}  \tr ( \matmu^{k}).
\end{align*}
Then, the proof is completed.
\end{proof}

\section{Determinant formula of \IB type}\label{sec:detIhara}
In the previous section, we showed that the zeta function is expressed as a determinant
of a size $\sum_{e \in \vec{E} } r_{e}$ matrix.
In this section, we show another determinant expression, requiring an additional assumptions on matrix weights.
The formula is called \IB type determinant formula and indispensably used in the derivation of
the \Bzf in the next chapter.

\subsection{The formula}
In the rest of this subsection, we fix a set of positive integers $\{r_i\}_{i \in V}$ 
associated with vertices. 
Let $\{u^{\alpha}_{\edij}\}_{\alpha \in F, i,j \in \alpha}$ be a set of matrices of size
$u^{\alpha}_{\edij} \in \mat{r_j,r_i}$.
Our additional assumption on the set of matrix weights, which is an argument of zeta function, is that
\begin{equation}
 r_e:=r_{t(e)} \text{~ and  ~} u_{\etea}:=u^{s(e)}_{t(e') \rightarrow t(e)}.
\end{equation}
Then the graph zeta function can be seen as a function of $\bs{u}=\{u^{\alpha}_{\edij}\}$.
With slight abuse of notation, it is also denoted by $\zeta_{H}(\bsu)$.
Later in Chapter \ref{chap:Bzf},
we see that $r_i$ corresponds to the dimension of the sufficient statistic $\phi_{i}$ and
$u^{\alpha}_{\edij}$ comes to a matrix $\var{b_j}{\phi_j}^{-1} \cov{b_{\alpha}}{\phi_j}{\phi_i}$.

To state the \IB type determinant formula,
we introduce a linear operator $\bs{\iota}(\bs{u}): \vfe \rightarrow \vfe$
defined by
\begin{equation}
 (\bs{\iota}(\bs{u})f)(e):=
\sum_{e': {s(e')=s(e) \atop t(e')\neq t(e) }} u^{s(e)}_{\ed{t(e')}{t(e)}}f(e')  \qquad f \in \vfe. \label{def:iota}
\end{equation}
\begin{figure}
\begin{center}
\includegraphics[scale=0.25]{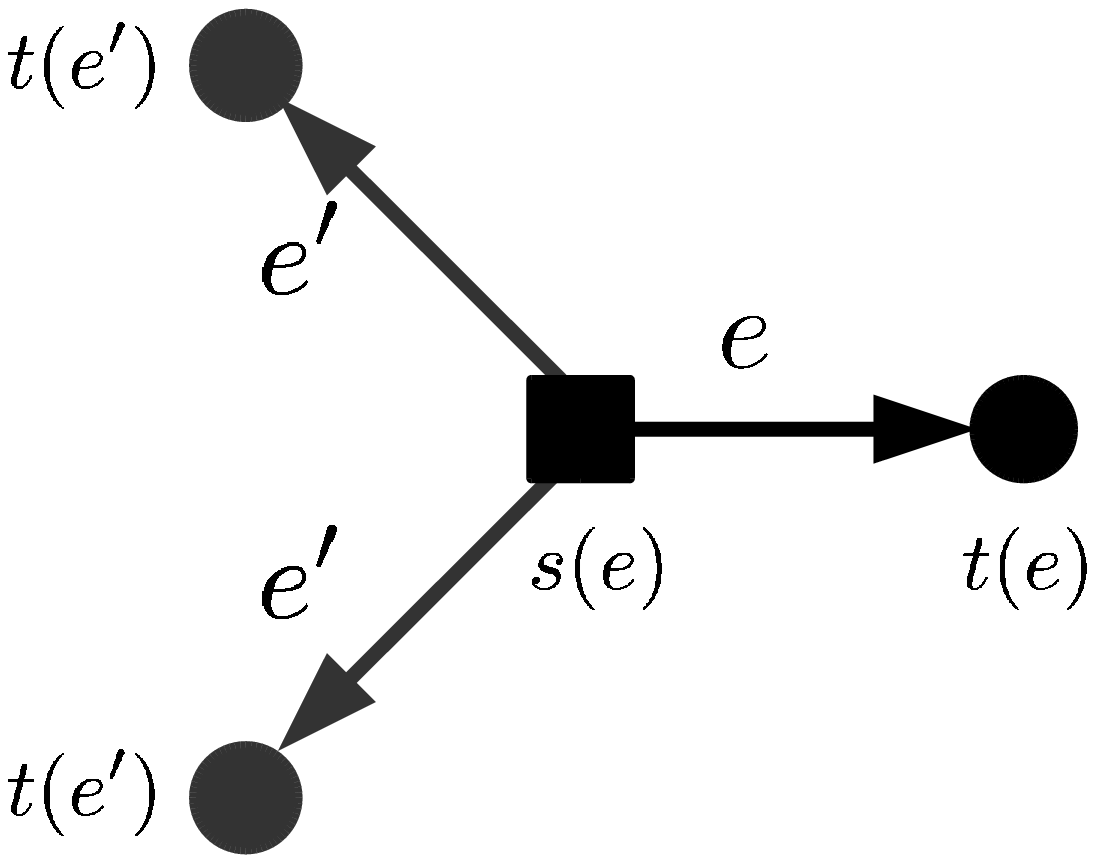}
\vspace{0mm}
\caption{Illustration for the definition of $\bs{\iota}(\bsu)$. \label{fig:iotaillustration}}
\end{center}
\vspace{0mm}
\end{figure}
The matrix representation of $\bs{\iota}(\bs{u})$ is a block diagonal matrix
because it acts for each factor separately.
Therefore $I+\bs{\iota}(\bs{u})$ is also a block diagonal matrix.
Each block is indexed by $\alpha \in F$ and denoted by $U_{\alpha}$.
Thus, for $\alpha=\{i_1,\ldots,i_{d_{\alpha}}\}$, 
\begin{small}
\begin{equation}
 U_{\alpha}= 
\begin{bmatrix}
 I_{r_{i_1}}    & u^{\alpha}_{i_2 \rightarrow i_1} & \cdots     &  u^{\alpha}_{i_{d_{\alpha}} \rightarrow i_1} \\
 u^{\alpha}_{i_1 \rightarrow i_2}   &  I_{r_{i_2}}         &  \cdots   & u^{\alpha}_{i_{d_{\alpha}} \rightarrow i_2} \\
   \vdots              &    \vdots      &  \ddots          & \vdots \\
u^{\alpha}_{i_{d_1} \rightarrow i_{d_{\alpha}}} &  u^{\alpha}_{i_{d_2 } \rightarrow i_{d_{\alpha}}} &  \cdots & I_{r_{i_{d_{\alpha}}}}\\
\end{bmatrix}. \label{eq:defU}
\end{equation}
\end{small}
We also define $w^{\alpha}_{\ed{i}{j}}$ by the elements of $W_{\alpha}=U_{\alpha}^{-1}$.
\begin{small}
\begin{equation}
 W_{\alpha}= 
\begin{bmatrix}
w^{\alpha}_{i_1 \rightarrow i_1}  & w^{\alpha}_{i_2 \rightarrow i_1} & \cdots     & w^{\alpha}_{i_{d_{\alpha}} \rightarrow i_1} \\
w^{\alpha}_{i_1 \rightarrow i_2}  & w^{\alpha}_{i_2 \rightarrow i_2} &  \cdots    & w^{\alpha}_{i_{d_{\alpha}} \rightarrow i_2} \\
   \vdots              &    \vdots      &  \ddots          & \vdots \\
w^{\alpha}_{i_{d_1 }\rightarrow i_{d_{\alpha}}} &  w^{\alpha}_{i_{d_2} \rightarrow i_{d_{\alpha}}} &  \cdots 
& w^{\alpha}_{i_{d_{\alpha}} \rightarrow i_{d_{\alpha}}    }  \\
\end{bmatrix}. \label{eq:defW}
\end{equation}
\end{small}

Similar to the definition of $\vfe$ in Subsection \ref{sec:defgraphzeta},
we define $\vfv$ as a set of functions on $V$ that takes value on $\mathbb{C}^{r_i}$ for each $i \in V$.

\begin{thm}[Determinant formula of Ihara-Bass type]
\label{thm:Ihara}
Let $\mathcal{D}$ are $\mathcal{W}$ are linear transforms on $\vfv$
defined by 
\begin{equation}
 (\mathcal{D}g)(i):= d_i g(i), \qquad
 (\mathcal{W}g)(i):= \sum_{{e,e' \in \vec{E} \atop {t(e)=i, s(e)=s(e')}  }} 
w^{s(e)}_{\ed{t(e')}{i}} g(t(e')).  \label{eq:defDW}
\end{equation}
Then, we have the following formula
\begin{equation}
  \zeta_{H}(\bs{u})^{-1}=
\det \big(
I_{r_V}- \mathcal{D} + \mathcal{W}
\big)
\prodf \det U_{\alpha},
\end{equation}
where $r_V:=\sum_{i \in V}r_i$ 
\end{thm}
The rest of this subsection is devoted to the proof of this formula.
The proof is based on the decomposition in the following Lemma \ref{lem:decM} and
the formula of Proposition \ref{app:detdet}.
We define a linear operator by
\begin{align*}
\mathcal{T}: \vfv \rightarrow \vfe,  \qquad
&(\mathcal{T}g)(e):=g(t(e)) \\
\end{align*}
The vector spaces $\vfe$ and $\vfv$ have inner products naturally.
We can think of the adjoint of $\mathcal{T}$ which is given by
\begin{equation*}
 \mathcal{T}^{*}: \vfe \rightarrow \vfv,  \qquad
(\mathcal{T}^{*}f)(i):= \sum_{e: t(e)=i} f(e).
\end{equation*}

The linear operators have a following relation.
\begin{lem}
\label{lem:decM}
\begin{equation}
 \matmu= \bs{\iota}(\bs{u})\mathcal{T}\mathcal{T}^* - \bs{\iota} (\bs{u})
\end{equation}
\end{lem}
\begin{proof}
Let $f \in \vfv$.
\begin{align*}
\Big( \bs{\iota}(\bs{u})\mathcal{T}\mathcal{T}^* - \bs{\iota} (\bs{u}) \Big) f(e) 
&= \sum_{e': {s(e')=s(e) \atop t(e') \neq t(e)}} u^{s(e)}_{\ed{t(e')}{t(e)}}
\sum_{e'': t(e'')=t(e')}f(e'')
- \sum_{e'': {s(e'')=s(e) \atop t(e'') \neq t(e)}} u^{s(e)}_{\ed{t(e'')}{t(e)}} f(e'') \\
&= \sum_{e': {s(e')=s(e) \atop t(e') \neq t(e)}} u^{s(e)}_{\ed{t(e')}{t(e)}}
\sum_{e'': {t(e'')=t(e') \atop e''\neq e'}}f(e'') \\
&=(\matmu f)(e).
\end{align*}
\end{proof}

\begin{proof}[Proof of Theorem \ref{thm:Ihara}]
 \begin{align*}
  \zeta_{H}(\bs{u})^{-1}&= \det(I- \matmu)\\
  &=\det(I-   \bs{\iota}(\bs{u})\mathcal{T}\mathcal{T}^* + \bs{\iota} (\bs{u})) \\
  &=\det(I-   \bs{\iota}(\bs{u})\mathcal{T}\mathcal{T}^* (I+\bs{\iota} (\bs{u}))^{-1})
   \det( (I+\bs{\iota} (\bs{u})) ) \\
  &=\det(I_{r_V}- \mathcal{T}^* (I+\bs{\iota} (\bs{u}))^{-1} \bs{\iota}(\bs{u})\mathcal{T} )
  \prod_{\alpha \in F} \det(U_{\alpha})
 \end{align*}
It is easy to see that
$I_{r_V}- \mathcal{T}^* (I+\bs{\iota} (\bs{u}))^{-1} \bs{\iota}(\bs{u})\mathcal{T}
=I_{r_V}- \mathcal{T}^*\mathcal{T}+ \mathcal{T}^* (I+\bs{\iota} (\bs{u}))^{-1} \mathcal{T}$.
We observe that
\begin{equation*}
 (\mathcal{T}^*\mathcal{T}g)(i)= \sum_{e: t(e)=i} g(t(e))=d_i g(i)
\end{equation*}
and
\begin{equation*}
 (\mathcal{T}^*(I+\bs{\iota} (\bs{u}))^{-1} \mathcal{T}g)(i)
=\sum_{e: t(e)=i} ((I+\bs{\iota} (\bs{u}))^{-1} \mathcal{T}g)(e)
=(\mathcal{W}g)(i).
\end{equation*}
\end{proof}

\subsection{Special cases of \IB type determinant formula}
\label{sec:specialIB}
In this subsection, we rewrite the above formula for two special cases.
The first case is the Storm's hypergraph zeta function \cite{Shypergraph},
where all matrix weights are set to be the same scalar value $u$.
In the second case, the zeta function is associated to a graph, not a hypergraph.
This case corresponds to the pairwise \ifa when we discuss the relations to the LBP algorithm in the next chapter.
In both of the cases, the matrix $W_{\alpha}$, which was defined in Eq.~(\ref{eq:defW})
as the inverse of $U_{\alpha}$, has explicit expressions.

\subsubsection{One variable hypergraph zeta}
Let $r_i=1$ and $u^{\alpha}_{\edij}=u$.
The set of functions of $\vec{E}$ and $V$ are denoted by 
$\sfe$ and $\sfv$ instead of $\vfe$ and $\vfv$.
We define the directed matrix by $\mathcal{M}= \mathcal{M}(1)$, i.e.,
\begin{equation}
\mathcal{M}_{e,e'}=
\begin{cases}
1          \qquad \text{if } \etea \\
0          \qquad \text{otherwise.}
\end{cases}  
\end{equation}
Then, $\matmu = u \mathcal{M}$.
Theorem \ref{thm:Ihara} is reduced to the following form.
\begin{cor}
\label{cor:hyperIharaBass}
\begin{equation*}
 \zeta_{H}(u)^{-1}=
\det\Big(
(1-u)I+u^2 \tilde{D}(u)-u\tilde{A}(u)
\Big)
(1-u)^{|\vec{E}|-|V|-|F|}
\prod_{\alpha \in F} 
\big(1+(d_{\alpha}-1)u \big),
\end{equation*}
where 
\begin{small}
 \begin{equation*}
 (\tilde{\mathcal{D}}(u)g)(i):= \Big( \sum_{\alpha \ni i}\frac{(d_{\alpha}-1)}{1+(d_{\alpha}-1)u} \Big) g(i), \quad
 (\tilde{\mathcal{A}}(u)g)(i):=\sum_{\alpha \supset \{i,j\} \atop j \neq i }\frac{1}{1+(d_{\alpha}-1)u} g(j), \quad
g \in \sfv.
 \end{equation*}
\end{small}
\end{cor}
\begin{proof}
 $(U_{\alpha})_{i,i}=1$ and $(U_{\alpha})_{i,j}=u$ implies that
$\det U_{\alpha}= (1-u)^{d_{\alpha}-1}(1+(d_{\alpha}-1))$ and
$(W_{\alpha})_{i,i}=(1+(d_{\alpha}-2)u)(1-u)^{-1}(1+(d_{\alpha}-1)u)$ and 
$(W_{\alpha})_{i,j}=-u(1-u)^{-1}(1+(d_{\alpha}-1)u)$ 
Therefore,
\begin{align*}
(I-\mathcal{D}+\mathcal{W})g(i)
&= g(i)-d_i g(i)+ 
\sum_{\alpha \supset \{i,j\} \atop j \neq i } \sfrac{-u}{(1-u)(1+(d_{\alpha}-1)u)}g(j)  
 + \sum_{\alpha \ni i}\sfrac{1+(d_{\alpha}-2)u}{(1-u)(1+(d_{\alpha}-1)u)}g(i) \\
&= (I- \frac{u}{1-u} \tilde{\mathcal{A}}(u) + \frac{u^2}{1-u}\tilde{\mathcal{D}}(u) )g(i).
\end{align*}
\end{proof}

Corollary \ref{cor:hyperIharaBass} extends Theorem 16 of \cite{Shypergraph},
where this type of formula is only discussed for $(d,r)$-regular hypergraphs.

\subsubsection{Non-hyper graph zeta}
Here and in the below, we consider the case that $H=(V,F)$ is a graph, i.e. all degrees of hyperedges are equal to two.
Then it is identified with an (undirected) graph. 
First, we define the zeta function $Z_G$ of a general graph $G=(V,E)$.
For each undirected edge of $G$, we make a pair of oppositely
directed edges, which form a set of {\it directed edges} $\vec{E}$.
Thus $|\vec{E}|=2|E|$. For each directed edge $e \in \vec{E}$, $o(e)
\in V$ is the {\it origin} of $e$ and $t(e) \in V$ is the {\it
terminus} of $e$.   For $e \in \vec{E}$, the {\it inverse edge} is
denoted by $\bar{e}$, and the corresponding undirected edge by
$[e]=[\bar{e}] \in E$.

A {\it closed geodesic} in $G$ is a sequence $(e_1,\ldots,e_k)$ of
directed edges such that 
$t(e_i)=o(e_{i+1}), e_i \neq \bar{e}_{i+1}$ for $i \in \mathbb{Z}/k\mathbb{Z}$.
Prime cycles are defined in a similar manner to that of hypergraphs.
The set of prime cycles is denoted by $\mathfrak{P}_{G}$.

\begin{defn}
Let $G=(V,E)$ a graph.
For given positive integers $\{r_i\}_{i \in V}$ and 
matrix weights $\bsu=\{ u_{e} \}_{e \in \vec{E}}$ of sizes
$u_e \in \mat{r_{t(e)}}{r_{o(e)}}$,
 \begin{equation}
Z_{G}(\boldsymbol{u}):=\prod_{\mathfrak{p} \in \mathfrak{P}_G }
\det (1-\pi(\mathfrak{p}))^{-1},
\quad
\pi(\mathfrak{p}):=u_{e_1} \cdots u_{e_k}
\quad
\text{ for }
\mathfrak{p}=(e_1,\ldots,e_k), 
\end{equation}
\end{defn}
This zeta function is the matrix weight extension of the edge zeta function in \cite{STzeta1}
where the edge weights are scalar values.
Since $\mathfrak{P}_{\ug{H}}$ is naturally identified with $\mathfrak{P}_{H}$,
$Z_{\ug{H}}=\zeta_{H}$.

\begin{cor}
\label{cor:IBfornonhyper}
For a graph $G=(V,E)$,
 \begin{equation}
  Z_{G}(\bs{u})^{-1}=
  \det ( I + \hat{\mathcal{D}}(\bsu) - \hat{\mathcal{A}}(\bsu) )
  \prod_{[e] \in {E}} \det(I - u_e u_{\bar{e}}). \label{eq:IBfornonhyper}
 \end{equation}
where $\hat{\mathcal{D}}$ and $\hat{\mathcal{A}}$ are defined by 
\begin{align}
 &(\hat{\mathcal{D}}(\bsu)g)(i):= \Big( \sum_{e: t(e)=i}(I_{r_i}-u_eu_{\bar{e}})^{-1}u_eu_{\bar{e}} \Big)g(i), \label{eq:modDop} \\
 &(\hat{\mathcal{A}}(\bsu)g)(i):= \sum_{e: t(e)=i}(I_{r_i}-u_eu_{\bar{e}})^{-1}u_e g(o(e)). \label{eq:modAop}
\end{align}
\end{cor}
\begin{proof}
For $e=(\edij)$, the $U_{[e]}$ block is given by 
\begin{equation}
 U_{[e]}=
\begin{bmatrix}
 I_{r_i} & u_e \\
u_{\bar{e}} & I_{r_j} \\
\end{bmatrix}
\end{equation}
Therefore $\det  U_{[e]}= \det(I_{r_i}-u_eu_{\bar{e}})$ and
the inverse $W_{[e]}$ is
\begin{equation}
 W_{[e]}=
\begin{bmatrix}
(I_{r_i}-u_eu_{\bar{e}})^{-1} & 0 \\
0 & (I_{r_j}-u_{\bar{e}}u_{e})^{-1} \\
\end{bmatrix}
\begin{bmatrix}
 I_{r_i} & -u_e \\
-u_{\bar{e}} & I_{r_j} \\
\end{bmatrix}.
\end{equation}
Plugging these equations into Theorem \ref{thm:Ihara}, we obtain the assertion. 
\end{proof}

In \cite{MSweighted} and \cite{HSTweighted}, a weighted graph version of \IB type determinant formula is derived
assuming scalar weights $\{u_e\}_{e \in \vec{E}}$ satisfy conditions of $u_e u_{\bar{e}} = u^2$.
In this case, the factors $(1-u_e u_{\bar{e}})^{-1}$ in Eqs.~(\ref{eq:modDop}) and (\ref{eq:modAop}) do not depend on $e$ and 
Eq.~(\ref{eq:IBfornonhyper}) is simplified.
Corollary \ref{cor:IBfornonhyper} gives the extension of the result to arbitrary weighted graph.
A direct proof, without discussing hypergraph case, of Corollary \ref{cor:IBfornonhyper} is found in Theorem 2 of \cite{WFzeta}.

\subsubsection{Ihara-Bass formula}
Reduced from these two special cases,
we obtain the following formula
which is known as Ihara-Bass formula:
\begin{equation*}
 Z_{G}(u)^{-1}
=(1-u^2)^{|E|-|V|}
\det(I-u\mathcal{A}+ u^2 (\mathcal{D}-I) ),
 \end{equation*}
where
$\mathcal{D}$ is the {\it degree matrix}
and
$\mathcal{A}$ is the {\it adjacency matrix}
defined by
\begin{equation}
(\mathcal{D}f)(i):=
d_i f(i),
\quad
(\mathcal{A}f)(i):=
\sum_{e \in \vec{E} ,  t(e)=i}
f(o(e)),
\quad
\text{    }
f \in C(V). \nonumber
\end{equation}
Many authors have been discussed the proof of the \IB formula. 
The first proof was given by Bass \cite{Bass} and others are found in \cite{KSzeta,STzeta1}.
A combinatorial proof by Foata and Zeilberger is found in \cite{FZcombinatorial}.

\section{Miscellaneous properties}\label{sec:miszeta}
This section provides miscellaneous topics.
In the first subsection, the prime cycles are discussed relating hypergraph properties. 
In the second subsection, we present additional properties of 
the directed edge matrix and the one-variable hypergraph zeta function.
These properties are utilized in the subsequent developments.

\subsection{Prime cycles}
\label{sec:primecycle}
\begin{prop}
\label{prop:invprime}
 $\mathfrak{P}_H$ = $\mathfrak{P}_{\core (H)}$
\end{prop}
\begin{proof}
Let $\fgdefn$.
The proof is by induction on $|\vec{E}|$.
If $H$ is a coregraph, the statement is trivial;
if not, there is a directed edge $e \in \vec{E}$ that satisfies $d_{s(e)}=1$ or $d_{t(e)}=1$.
Obviously, there is no geodesic that goes through $e$.
Therefore, removal of $e$ from $H$ does not affect the set of prime cycles. 
\end{proof}
This Proposition immediately implies that $\zeta_{H}=\zeta_{\core(H)}$.

The following proposition claims that trees can be characterized in terms of prime cycles.
\begin{prop}
\label{prop:treeprime}
 Let $H$ be a connected hypergraph.
$H$ is a tree if and only if $\mathfrak{P}_{H}=\emptyset$.
\end{prop} 
\begin{proof}
The ``only if'' part is trivial from Propositions \ref{prop:invprime} and \ref{prop:treecore}. 
For ``if'' part, assume that $H$ is not a tree.
Then there is a cycle and the cycle gives a closed geodesic.
Therefore, $\mathfrak{P}_{H} \neq \emptyset$.
\end{proof}

\subsection{Directed edge matrix}
\label{sec:directededgematrix}
First, we derive a simple formula for the determinant of the directed edge matrix $\mathcal{M}$.
This type of expression appears in the \ls expansion of the perfect matching problem in Section \ref{sec:perfectmatching}.
\begin{thm}
\label{thm:detofM}
\begin{equation}
  \det \mathcal{M} =
\prodv (1-d_i)
\prodf (1-d_{\alpha})
\end{equation}
\end{thm}
\begin{proof}
From Lemma \ref{lem:decM}, we have $\mathcal{M}= \iota \mathcal{T}\mathcal{T}^{*}-\iota$,
where $\iota= \iota(\bs{1})$. 
From Eq.~(\ref{def:iota}) and Proposition \ref{app:detuniform}, we see that $\det(\iota)=\prod_{\alpha \in F}(-1)^{d_{\alpha -1}}(d_{\alpha}-1)$. 
Using Proposition \ref{app:detdet}, the assertion follows.
\end{proof}
This formula implies that the matrix is invertible if and only if the hypergraph has a nonempty coregraph.
Since $\zeta_{H}(u)=\zeta_{\core (H)}(u)$, the spectrum (i.e. the set of eigenvalues) 
of $H$ and $\core (H)$ only differs by zero eigenvalues.

Next, we consider the irreducibility of the non-negative matrix $\mathcal{M}$.
\begin{prop}
For a connected hypergraph $H$,
$\mathcal{M}$ is irreducible if and only if $H$ is a coregraph and $n(H) \geq 2$. 
\end{prop}
\begin{proof}
By definition, $\mathcal{M}$ is irreducible iff, for arbitrary $e$ and $e' \in \vec{E}$,
there is a sequence of directed edges $(e_1,e_2,\ldots,e_k)$ s.t. 
$e_1=e$, $e_l  \rightharpoonup \e_{l+1}~(l=1,\ldots,k-1)$ and $e_k=e'$.
If $H$ is a connected coregraph with $n(H) \geq 2$, we can construct such a sequence if $H$ is a connected coregraph and has more than one cycles.
If not, we can not do that.
(Detail is omitted.)
\end{proof}

Another important question regarding the directed matrix $\mathcal{M}$ is the spectral radius,
or the Perron-Frobenius eigenvalue.
\begin{prop}
\label{prop:PFboundM}
For $e \in \vec{E}$, let 
$k_e :=| \{ e'\in \vec{E}; \etea    \} |$,
$k_m= \min k_e$ and $k_{M}=\max k_e$. Then
\begin{equation}
 k_m \leq \specr{\mathcal{M}} \leq k_M. \label{eq:luboundM}
\end{equation}
Therefore, if $\core (H) \neq \emptyset$, then $\specr{\mathcal{M}} \geq 1 $.
If $H$ is $(a,b)$-regular, $\specr{{\mathcal{M}}}=(a-1)(b-1)$.
If $H$ is a graph,
\begin{equation}
 \min_{i \in V} d_i -1 \leq \specr{{\mathcal{M}}} \leq   \max_{i \in V} d_i -1.  \label{eq:PFboundMgraph}
\end{equation}
\end{prop}
\begin{proof}
Since $k_e = \sum_{e'} \mathcal{M}_{e,e'}$, the bound Eq.~(\ref{eq:luboundM}) is trivial from Theorem \ref{app:luboudofspecr}.
The second statement comes from $k_m \geq 1$ for non-empty coregraphs.
\end{proof}

Finally let us consider the pole of $\zeta_G(u)$.
Obviously the pole closest to the origin is $u= \specr{\mathcal{M}}^{-1} \geq k_M^{-1}$
and is simple if $\mathcal{M}$ is irreducible.
Furthermore, the following theorem implies that
$\zeta_G(u)$ has a pole at $u=1$ with multiplicity $n(H)$ if $H$ is connected and $n(H) \geq 2$.

\begin{thm}[Hypergraph Hashimoto's theorem]
\label{thm:Hashimoto}
Let $\chi(H):=|V|+|F|-|\vec{E}|$ be the {\it Euler number} of $H$.
\begin{equation}
 \lim_{u \rightarrow 0}
\zeta_H(u)^{-1}(1-u)^{- \chi(H)+1}= \chi(H) \kappa(B_H),
\end{equation}
where $\kappa(B_H)$ is the number of spanning trees of the bipartite graph $B_H$.
(See Subsection \ref{sec:basicsgraph} for the construction of $B_H$ from $H$.)
\end{thm}
\begin{proof}
For a graph $G=(V,E)$, Hashimoto proved that \cite{Hpadic,Hzeta}
\begin{equation}
\lim_{u \rightarrow 1}
Z_{G}(u)^{-1}(1-u)^{-|E|+|V|-1}
=
-2^{|E|-|V|+1}(|E|-|V|)
\kappa(G), \nonumber
\end{equation}
where $\kappa (G)$ is the number of spanning tree of $G$.
A simple proof by Northshield is found in \cite{Nnote}.
Since there is a one-to-one correspondence between
$\mathfrak{P}_H$ and $\mathfrak{P}_{B_H}$,
we have $\zeta_{H}(u)=Z_{B_H}(\sqrt{u})$.
Then the assertion is proved by the above formula.
\end{proof}

In \cite{Shypergraph}, Storm showed that, if $H$ is an $(a,b)$-regular Ramanujan hypergraphs, 
all non trivial poles of $\zeta_G(u)$ lie on the circle of radius $[(a-1)(b-1)]^{-1/2}$.
This property is analogous to the Riemann hypothesis (RH) of the Riemann zeta function, 
which claims that all non trivial zeros of the Riemann zeta function have real part of $1/2$.
For the Ihara zeta function, a bound on the modulus of imaginary poles is found in \cite{KSzeta}.

\section{Discussion}\label{sec:discussion:zeta}
In this chapter, we introduced our graph zeta function, generalizing graph zeta functions known in graph theory.
Our main contribution of this chapter is the \IB type determinant formula, which extends
the \IB formula of one-variable graph zeta function.
The proof is based on a simple determinant formula in Proposition \ref{app:detdet},
changing the size of determinant from ${\rm dim}\vfe$ to ${\rm dim}\vfv$.

The \IB type determinant formula plays an important role in developments in the sequel,
especially in the proof of \Bzf in Chapter \ref{chap:Bzf}.
The formula is also used in the alternative derivation of the \ls in the perfect matching problem,
showing intimate relations between the graph zeta function and the Bethe approximation.

The definition of our zeta function can be extended to Bartholdi type zeta function
where closed geodesics are allowed to have backtracking \cite{Bcounting}.
The \IB type determinant formula in Theorem \ref{thm:Ihara} is also extended to this case without difficulty.
A related work is found in \cite{Ibartholdi}.

In this thesis, we discuss the graph zeta function only due to the connection to LBP algorithm and \Bfe function.
However, there are other contexts where Ihara zeta function appears.
We refer a paper \cite{Sdiscrete} for a review of the Ihara zeta function and related topics.

%% file: chapter4.tex

\section{Introduction}
\label{sec:Bzfintro}
The aim of this chapter is to show the ``\Bzf''
and to demonstrate its applications.
This formula provides a relation between the Hessian of the \Bfe function
and the graph zeta function.

In Section \ref{sec:BZ}, we prove the main formula using the \IB type determinant formula
proved in the previous chapter.
In Section \ref{sec:PDC}, as an application of the main formula,
we analyze the region where the Hessian of \Bfe function $F$ is positive definite.
Section \ref{sec:stability} discusses the stability of the LBP fixed points,
extending results of Heskes \cite{Hstable} and Mooij et al \cite{MKproperty}.
The main formula is further applied to the uniqueness problem of the LBP fixed points
in the next chapter.

\subsection{Intuition for the Bethe zeta formula}
Beforehand, we describe the underlying mathematical structure
that let the \Bzf hold.
It is the ``duality'' between the two variational characterizations of the LBP fixed points
given in Theorem \ref{thm:LBPcharacterizations}.

Recall that the LBP fixed points are the intersections of
the submanifold $L(\mathcal{I})$ and $A(\mathcal{I},\Psi)$.
See Figure \ref{fig:intersection}.
The whole space is $Y \simeq \Theta$;
an element of this set is a vector of all expectation/natural parameters of local exponential families.
Each line stands for $L$ and $A$ respectively
and the intersection is an LBP fixed point.
In the first figure, the submainfolds intersect transversally, while
those intersect tangentially in the second figure. 
In the second case, both the Hessians of $F$ and $\mathcal{F}$
degenerate.
Therefore, one can expect that
\begin{equation}
 \det ( \nabla^2 F )=0  \iff \text{ The intersection is tangential } \iff   \det ( \nabla^2 \mathcal{F} )=0, 
\end{equation}
where $\nabla^2$ denotes the Hessian with respect to the coordinates of $L$ and $A$, respectively.
\begin{figure}
 \begin{center}
  \includegraphics[scale=0.3]{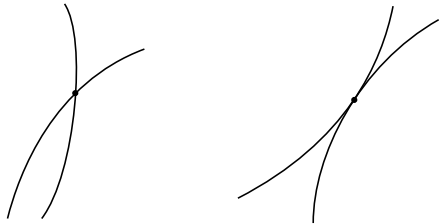}
\caption{Two types of intersections. \label{fig:intersection}}
 \end{center}
\end{figure}
After calculations (See Appendix \ref{sec:HesseF}),
one can see that 
\begin{equation}
 \nabla^{2} \mathcal{F} =X [I -\matmu ]Y \label{eq:nablaFtype2}
\end{equation}
holds at an LBP fixed point with certain matrices $X$ and $Y$.

These observations suggest that
there is a relation like $ \det(\nabla^2 F)=\det(I-\matmu)  \times \text{(factor)}$
at LBP fixed points.
Furthermore, since $A(\mathcal{I},\Psi)$ moves depending on $\Psi$,
one can expect that such relations hold at all points of $L(\mathcal{I})$.

Based on the techniques developed in the previous chapter,
we will formulate the relations as an identity on $L$ rather than statements on LBP fixed points, i.e.,
the \Bzf.
The first advantage of this approach is its powerful applicability.
In fact, the \Bzf will be utilized as a continuous function on $L$ in the proof of Theorem \ref{thm:positive}.
The second advantage is the simplicity of the proof.
This approach only involves linear algebraic calculations
and is much easier than just making the above observations rigorous.

\section{\Bzf}\label{sec:BZ}
In order to make the assertion clear, we first recall the definitions and notations.
Let $H=(V,F)$ be a hypergraph and
let $\mathcal{I}=\{ \mathcal{E}_{\alpha},\mathcal{E}_i\}$ be an \ifa on $H$.
Exponential families $\mathcal{E}_i$ and $\mathcal{E}_{\alpha}$ have sufficient statistics $\phi_i$
and $\fa{\phi}$ as discussed in Subsection \ref{sec:infmodel}. 
Furthermore, as discussed in Subsection \ref{sec:twoBfes},
a point $\bseta=\{\pa{\eta},\eta_i\} \in L$ is identified with
a set of pseudomarginals $\beliefs$.
\begin{thm}
[Bethe-zeta formula]
\label{thm:BZ}
At any point of $\bseta=\{\pa{\eta},\eta_i\} \in L$ the following equality holds. 
\begin{equation*}
\zeta_{H}(\bsu)^{-1}
\hspace{-1mm}
=
\det(\nabla^2 F)
\prod_{\alpha \in F}
\hspace{-0mm}
\det(\var{b_{\alpha}}{\fa{{\phi}} }) 
\prod_{i \in V}
\hspace{-0mm}
\det(\var{b_{i}}{\phi_{i}})^{1-d_i}, 
\end{equation*}
where
\begin{equation}
 u^{\alpha}_{\ed{i}{j}}:=
\var{b_j}{{\phi}_{j}}^{-1}
\cov{b_{\alpha}}{{\phi}_{j}}{{\phi}_{i}} \label{def:u}
\end{equation}
is an $r_j \times r_i$ matrix.
\end{thm}
Note that $\nabla^2 F$ is the Hessian matrix with respect to the coordinate $\{\pa{\eta},\eta_i\}$.
The Hessian does not depend on the given compatibility functions $\Psi_{\alpha}$ because
those only affect linear terms in $F$.
So, the formula is accompanied with the \ifa $\mathcal{I}$.

By the definition of the \ifa,
all local exponential families $\mathcal{E}_{\alpha}$ and $\mathcal{E}_i$
satisfy Assumption \ref{asm:expregular}. 
Therefore, the determinants of variances appear in the formula are always positive.

Note that the zeta function is given by the products of weights Eq.~(\ref{def:u}) along prime cycles.
This type of expression also appears in the covariance expression of distant vertices on a tree structured hypergraph.
(See Appendix \ref{app:sec:infontree}.)

\begin{proof}
[Proof of Theorem \ref{thm:BZ}]
From the definition of the \Bfe function Eq.~(\ref{defn:Bfe}), 
the (V,V)- block of $\Hesse F$ is given by
\begin{equation*}
 \pdseta{F}{i}{i}= \sum_{\alpha \ni i} \pdseta{\varphi_{\alpha}}{i}{i}+(1-d_i)\pdseta{\varphi_i}{i}{i},
\quad
 \pdseta{F}{i}{j}= \sum_{\alpha \supset \{i,j\}} \pdseta{\varphi_{\alpha}}{i}{j} \quad (i \neq j).
\end{equation*}
The (V,F)-block and (F,F)-block are given by 
\begin{equation*}
 \pds{F}{\eta_i}{\pa{\eta}}= \pds{\varphi_{\alpha}}{\eta_i}{\pa{\eta}},
\quad \quad
 \pds{F}{\pa{\eta}}{\pb{\eta}}= \pds{\varphi_{\alpha}}{\pa{\eta}}{\pb{\eta}}\delta_{\alpha,\beta}.
\end{equation*}
Using the diagonal blocks of (F,F)-block, we erase (V,F)-block and 
(F,V)-block of the Hessian.
In other words, we choose a square matrix $X$ such that $\det X =1$ and 
\begin{equation}
X^T (\nabla^2 F) X
=
\begin{bmatrix}
\quad Y & 0 \\
\quad 0 & 
\Big( \pds{F}{\pa{\eta}}{\pb{\eta}} \Big) 
\end{bmatrix}.   \label{eq:XHesseFX}
\end{equation}  
Then we obtain
\begin{align}
 Y_{i,i}
&=\sum_{\alpha \ni i} \left\{ \pdseta{\varphi_{\alpha}}{i}{i}
- \pds{\varphi_{\alpha}}{\eta_i}{\pa{\eta}} \left(\pds{\varphi_{\alpha}}{\pa{\eta}}{\pa{\eta}}\right)^{-1}
\pds{\varphi_{\alpha}}{\pa{\eta}}{\eta_i}  \right\}  + (1-d_i) \pdseta{\varphi_{i}}{i}{i}, \label{eq:Yii}\\
Y_{i,j}
&=\sum_{\alpha \supset \{i,j\} } \left\{ \pds{\varphi_{\alpha}}{\eta_i}{\eta_j}
-\pds{\varphi_{\alpha}}{\eta_i}{\pa{\eta}} \left(\pds{\varphi_{\alpha}}{\pa{\eta}}{\pa{\eta}}\right)^{-1}
\pds{\varphi_{\alpha}}{\pa{\eta}}{\eta_j}  \right\} .  \label{eq:Yij}
\end{align}
On the other hand,
since $u^{\alpha}_{\ed{i}{j}}:=$ $\var{b_j}{{\phi}_{j}}^{-1}$ $\cov{b_{\alpha}}{{\phi}_{j}}{{\phi}_{i}}$,
the matrix $U_{\alpha}$ defined in Eq.~(\ref{eq:defU}) is
\begin{equation}
 U_{\alpha}
=
\diag ( \var{}{\phi_i}^{-1} | i \in \alpha)
\hspace{1mm} \var{b_{\alpha}}{(\phi_i)_{i \in \alpha}}.
\end{equation}
Since the matrix $\var{b_{\alpha}}{(\phi_i)_{i \in \alpha}}$ is a submatrix of $\var{b_{\alpha}}{\fa{\phi}}$,
its inverse can be expressed by submatrices of 
$\var{b_{\alpha}}{\fa{\phi}}^{-1}= \pds{\varphi_{\alpha}}{\fa{\eta}}{\fa{\eta} }$
using Proposition \ref{app:schur}.
Therefore, the elements of $W_{\alpha}=U_{\alpha}^{-1}$ is given by 
\begin{equation}
\label{eq:wij}
 w^{\alpha}_{\edji}=
\left\{ \pdseta{\varphi_{\alpha}}{i}{j}
-\pds{\varphi_{\alpha}}{\eta_i}{\pa{\eta}} \left(\pds{\varphi_{\alpha}}{\pa{\eta}}{\pa{\eta}} \right)^{-1}
\pds{\varphi_{\alpha}}{\pa{\eta}}{\eta_j}  \right\} 
\var{}{\phi_j}.
\end{equation}
Combining Eq.~(\ref{eq:Yii}),(\ref{eq:Yij}) and (\ref{eq:wij}),
we obtain
\begin{equation}
 Y \hspace{1mm} \diag \left( \var{}{\phi_i} | i \in V \right)= I - \mathcal{D} + \mathcal{W},
\end{equation}
where $\mathcal{D}$ and $\mathcal{W}$ are defined in Eq.~(\ref{eq:defDW}).
Accordingly, we obtain
\begin{align*}
 \zeta_H(\bsu)^{-1}
&= \det (I-\mathcal{D}+\mathcal{W}) \prodf \det U_{\alpha}  \\
&= \det Y \prodv \det ( \var{}{\phi_i} ) \prodf       
\frac{ \det \left( \var{b_{\alpha}}{(\phi_i)_{i \in \alpha}} \right)}{ \prod_{j \in \alpha} \det \left( \var{}{\phi_j} \right) } \\
&= \det \left( \Hesse F \right) \prodv \det ( \var{}{\phi_i} )^{1-d_i} \prodf 
\frac{ \det \left( \var{b_{\alpha}}{(\phi_i)_{i \in \alpha}} \right)}{\det \left( \pds{\varphi_{\alpha}}{\pa{\eta}}{\pa{\eta}} \right)} \\
&= \det \left( \nabla^2 F \right) \prodf
\det(\var{b_{\alpha}}{ \fa{\phi} }) 
\prodv \det(\var{b_{i}}{{\phi}_{i}})^{1-d_i},
\end{align*}
where we used
$ \det \left( \var{b_{\alpha}}{(\phi_i)_{i \in \alpha}} \right) \det \left( \pds{\varphi_{\alpha}}{\pa{\eta}}{\pa{\eta}} \right)^{-1}$
$=\det \left( \var{}{\fa{\phi}}   \right)$,
which is proved by Proposition \ref{app:schur}.
\end{proof}

\subsection{Case 1: Multinomial \ifa}%
In the rest of this section, 
we rewrite the \Bzf for specific cases.
Especially, we give explicit expressions of the determinants of the variances.
First, we consider the multinomial case.

\begin{lem}
Let $\phi$ be the sufficient statistics of the multinomial distributions
on $\mathcal{X}=\{ 1,2,\ldots,N\}$ defined in Example \ref{example:multinomial}.
Then the determinant of the variance is given by
\begin{equation}
  \det \left( \var{p}{{\phi}} \right)= \prod_{k=1}^{N} p(k) .
\end{equation}
\end{lem}
\begin{proof}
From the definition of the sufficient statistics, one easily observes that
 $\var{}{\phi_i}=p(i)-p(i)^2$ and $\cov{}{\phi_i}{\phi_j}=-p(i)p(j)$.
Therefore,
\begin{align*}
 \det(\var{}{{\phi}} )
&= \det \left(
\begin{small}
\begin{bmatrix}
p(1)&    & 0 \\
& \ddots &   \\
0&       & p(N-1)  
\end{bmatrix} -
\begin{bmatrix}
p(1) \\
\vdots \\
p(N-1)  
\end{bmatrix}
\begin{bmatrix}
p(1) & \cdots & p(N-1)  
\end{bmatrix}
\end{small}
\right) \\
&=
(1- \sum_{k=1}^{N-1} p(k) )
\prod_{k=1}^{N-1}p(k).
\end{align*}
\end{proof}

\begin{cor}
[Bethe-zeta formula for multinomial \ifa]
For any pseudomarginals $ \beliefsw \in L$ the following equality holds. 
 \begin{equation*}
\zeta_{H}(\bsu)^{-1}
\hspace{-1mm}
=
\det(\nabla^2 F)
\prod_{\alpha \in F}
\prod_{x_{\alpha}} b_{\alpha}(x_{\alpha})
\prod_{i \in V}
\prod_{x_i} b_i(x_i)^{1-d_i},
 \end{equation*}
where
$u^{\alpha}_{\ed{i}{j}}:=$
$\var{b_j}{{\phi}_{j}}^{-1}$
$\cov{b_{\alpha}}{{\phi}_{j}}{{\phi}_{i}}$
is an $r_j \times r_i$ matrix.
\end{cor}
For binary and pairwise case, this formula is first shown in \cite{WFzeta}.

\subsection{Case 2: Fixed-mean Gaussian \ifa}%
Let $G=(V,E)$ be a graph.
We consider the fixed-mean Gaussian \ifa on $G$.
For a given vector $\bsmu=(\mu_i)_{i \in V}$, the \ifa
is constructed from sufficient statistics
\begin{equation}
 \phi_i(x_i)=(x_i-\mu_i)^2  \quad \text{ and } \quad \px{\phi}{ij}(x_i,x_j)=(x_i-\mu_i)(x_j-\mu_j).
\end{equation}
The expectation parameters of them are denoted by $\eta_{ii}$ and $\eta_{ij}$, respectively.
The variances and covariances are
\begin{equation}
\var{}{\phi_i}=2 \eta_{ii}^2 , \quad
\var{}{\fx{\phi}{ij}}=
\begin{small}
\begin{bmatrix}
 2 \eta_{ii}^2 & 2 \eta_{ij}^2 & 2 \eta_{ii} \eta_{ij} \\
 2 \eta_{ij}^2 & 2 \eta_{jj}^2 & 2 \eta_{jj} \eta_{ij} \\
 2 \eta_{ii} \eta_{ij} & 2 \eta_{jj}\eta_{ij} &  \eta_{ij}^2 +\eta_{ii} \eta_{jj}\\
\end{bmatrix},
\end{small}   \label{eq:covvarfGauss}
\end{equation}
where 
\begin{equation*}
 \fx{\phi}{ij}(x_i,x_j)= \left( (x_i-\mu_i)^2,(x_j-\mu_j)^2,(x_{i}-\mu_i)(x_j-\mu_j) \right).
\end{equation*}
Therefore, $\det(\var{}{\fx{\phi}{ij}})=4 (\eta_{ii} \eta_{jj} - \eta_{ij}^2)^3 $.
\begin{cor}
[\Bzf for fixed-mean Gaussian \ifa]
For any pseudomarginals $ \{\eta_{ii},\eta_{ij}\} \in L$ the following equality holds. 
 \begin{equation*}
Z_{G}(\bsu)^{-1}
\hspace{-1mm}
=
\det(\nabla^2 F)
\prodv \eta_{ii}^{2(1-d_i)}
\prod_{ij \in E} (\eta_{ii} \eta_{jj} - \eta_{ij}^2)^3
\hspace{2mm} 2^{|V|,}
 \end{equation*}
where
$u^{ij}_{\ed{i}{j}}:= \eta_{ij}^2 \eta_{jj}^{-2}$
is a scalar value.
\end{cor}
One interesting point of this case is that the edge weights $u_{\ed{i}{j}}$ are
always positive.

\section{Application to positive definiteness conditions} 
\label{sec:PDC}
The \Bfe function $F$ is not necessarily convex
though it is an approximation of the Gibbs free energy function, which is convex.
Non convexity of the Bethe free energy can lead to multiple fixed points. 
Pakzad et al \cite{PAstat} and Heskes \cite{Huniquness} have derived sufficient conditions of the convexity and
have shown that the Bethe free energy is convex for trees and graphs with one cycle. 
In this section, instead of such global structure, we
shall focus on the local structure of the Bethe free energy function, i.e. the Hessian. 

As an application of the \Bzf,
we derive a condition for positive definiteness of the Hessian of the \Bfe function.
This condition is utilized to analyze a region where the Hessian is positive definite.

We will use the following notations.
For a given square matrix $X$,
$\spec{X} \subset \mathbb{C}$ denotes the set of eigenvalues (spectra)
and $\rho(X)$ the spectral radius of a matrix $X$,
i.e., the maximum of the modulus of the eigenvalues.

\subsection{Positive definiteness conditions}
\begin{lem}
\label{lem:zerocorr}
Let $\bseta=\{\pa{\eta},\eta_i \} \in L$.
 If $\cov{b_{\alpha}}{\phi_i}{\phi_j}=0$ is for all $\alpha \in F$
and $i,j \in \alpha (i \neq j)$, then
$\Hesse F(\bseta)$ is a positive definite matrix.
\end{lem}
\begin{proof}
We use the notations following Theorem \ref{thm:BZ}.
The assumption of this lemma means $u^{\alpha}_{\edij}=0$.
Since $W_{\alpha}=U_{\alpha}^{-1}=I$, we have $w^{\alpha}_{\edij}=\delta_{i,j}$.
Therefore,
$Y_{ij}=\var{}{\phi_i}^{-1} \delta_{i,j}$ and $Y$ is a positive definite matrix.
Furthermore, $\pds{\varphi_{\alpha}}{\pa{\eta}}{\pa{\eta}}$ is a positive definite matrix
because it is a submatrix of the positive definite matrix $\pds{\varphi_{\alpha}}{\fa{\eta}}{\fa{\eta}}=\var{}{\fa{\phi}}^{-1}$.
Therefore, from Eq.~(\ref{eq:XHesseFX}), $\Hesse F$ is positive definite.
\end{proof}

\begin{thm}
\label{thm:positive}
Let $\mathcal{I}$ be a multinomial or fixed-mean Gaussian \ifa.
Let $\bsu$ be given by $\bseta \in L$ using Eq.~(\ref{def:u}).
Then, 
\begin{equation*}
\spec{ \matmu } \hspace{0.5mm} \subset
\hspace{0.5mm} \mathbb{C} \smallsetminus \mathbb{R}_{\geq 1}
\quad \Longrightarrow \quad
\nabla^2 F (\bseta)
\text{ is a positive definite matrix.}
\end{equation*}
\end{thm}
\begin{proof}We give proofs for each case.\\
{\bf Case 1: Multinomial}\\
The given $\bseta$ is identified with a set of pseudomarginals $\beliefsw$.
We define $\bseta(t) (t \in [0,1])$ by a set of pseudomarginals
$b_{\alpha}(t):= t b_{\alpha}+(1-t)\prod_{i \in \alpha}b_i$ and
$b_i(t):=b_i$.
Obviously, $\bseta(1)=\bseta$, and $\bseta(0)$ has zero covariances.
From Lemma \ref{lem:zerocorr}, it is enough to prove that 
$\nabla^2 F ( \bseta(t)) \neq 0$  on the interval $[0,1]$
because all eigenvalues of the Hessian are real numbers. 

The covariances and variances at $t$ are
\begin{equation}
 \cov{b_{\alpha}(t)}{\phi_i}{\phi_j}
=t \cov{b_{\alpha}}{\phi_i}{\phi_j}, \quad
\var{b_{\alpha}(t)}{\phi_i}=\var{b_{\alpha}}{\phi_i}.
\end{equation}
Therefore, $\mathcal{M}(\bsu(t)) = t \matmu$.
Our assumption of this lemma implies $\det(I-\mathcal{M}(\bsu(t)) ) \neq 0$
on the interval.
From Theorem \ref{thm:BZ}, we conclude that $\nabla^2 F ( \bseta(t)) \neq 0$.\\
{\bf Case 2: Fixed-mean Gaussian}\\
The proof is analogous to the above proof.
We define $\bseta(t) \in L$ by
\begin{equation}
 \bseta(t)_{ii}: = \eta_{ii}, \quad \bseta(t)_{ij}: = t \bseta_{ij}.
\end{equation}
From Eq.~(\ref{eq:covvarfGauss}), we have 
\begin{equation}
  \cov{b_{\alpha}(t)}{\phi_i}{\phi_j}
=t^2 \cov{b_{\alpha}}{\phi_i}{\phi_j}, \quad
\var{b_{\alpha}(t)}{\phi_i}=\var{b_{\alpha}}{\phi_i}.
\end{equation}
Therefore $\mathcal{M}(\bsu(t)) = t^2 \matmu$.
The remainder of the proof proceeds in the same manner.
\end{proof}

\subsection{Region of positive definite}
\label{sec:regionPD}
In this section, we analyze conditions for
the pseudomarginals that guarantees the positive definiteness of the Hessian.
Our result says that if the \ccms
are sufficiently small, then the Hessian is positive definite.
This ``smallness'' criteria depends on graph geometry.

First, we define \ccms.
\begin{defn}
Let $x, y$ be vector valued random variables following a probability distribution $p$.
The {\it \ccm} of $x$ and $y$ is defined by
\begin{equation}
 \corr{p}{y}{x}:= \var{p}{y}^{-1/2} \cov{p}{y}{x} \var{p}{x}^{-1/2}.  \label{eq:defccm}
\end{equation}
\end{defn}

Our approach for obtaining conditions for the positive definiteness is based on
Theorem \ref{thm:positive}.
Thus we would like to bound the eigenvalues of $\matmu$.
The following lemma implies that the spectrum of $\matmu$ is determined by
the \ccms.
\begin{lem}
\label{lem:specutoc}
Let $u^{\alpha}_{\edij}$ be given by Eq.~(\ref{def:u}) and
$c^{\alpha}_{\edij}:= \corr{b_{\alpha}}{\phi_j}{\phi_i}$.
Then
\begin{equation}
 \spec{\matmu }=\spec{\matmc}.
\end{equation}
\end{lem}
\begin{proof}
Define $\mathcal{Z}$ by $(\mathcal{Z})_{e,e'}:= \delta_{e,e'} \var{}{\phi_{t(e)}}^{1/2}$.
Then
\begin{equation}
(\mathcal{Z}\matmu \mathcal{Z}^{-1})_{e,e'}=
\var{}{\phi_{t(e)}}^{1/2}  \matmu_{e,e'} \var{}{\phi_{t(e')}}^{-1/2} 
= \matmc_{e,e'}.
\end{equation}
\end{proof}

Next, we define the operator norm of matrices
because we need to measure ``smallness'' of a \ccm.
\begin{defn}
Let $V_1$ and $V_2$ be finite dimensional normed vector spaces and let $X$ be a linear operator 
from $V_1$ to $V_2$.
The {\it operator norm} of $X$ is defined by
\begin{equation}
 \norm{X} := \max_{\substack{x \in V_1 \\ \norm{x}=1 }} \norm{Xx}.
\end{equation}
\end{defn}
Since $V_1$ is finite dimensional, the maximum exists.
By definition, $\specr{X} \leq \norm{X}$ and $\norm{XY} \leq \norm{X}~\norm{Y}$ holds.

The operator norm depends on the choice of the norms of $V_1$ and $V_2$.
If the norms in $V_1$ and $V_2$ are given by its inner products,
the induced operator norm is denoted by $\norm{\cdot}_2$.
Then $\norm{X}_2$ is equal to the maximum singular value of $X$.
In this case, the norm of a \ccm is smaller than 1.
(See Proposition \ref{app:prop:normccm} in Appendix \ref{app:sec:probability}.)

\begin{lem}
\label{lem:specbound}
Let 
\begin{equation}
\kappa:= \max_{\substack{\alpha \in F \\ i,j \in \alpha}} \norm{\corr{b_{\alpha}}{\phi_i}{\phi_j}} 
\end{equation}
and
let $\alpha$ be the Perron-Frobenius eigenvalue of $\mathcal{M}$.
Then
\begin{equation}
\specr{ \matmu } \leq \kappa \alpha.
\end{equation}
\end{lem}
\begin{proof}
From Lemma \ref{lem:specutoc}, we consider the spectral radius of $\matmc$.
It is enough to prove $\det(I-z \mathcal{M}(\bs{c}))$ does not have any root in 
$ \{ \lambda \in \mathbb{C} |\hspace{2mm} |\lambda| < \kappa^{-1} \alpha^{-1} \}$.
Accordingly, we show that $\zeta_{H}(z \bs{c})$ does not have any pole in the set.
If $H$ is a tree, the statement is trivial.
Thus, from Proposition \ref{prop:PFboundM}, we assume $\alpha \geq 1$ in the following.
Let $\mathfrak{p}$ be a prime cycle and
let $\lambda_1,\ldots,\lambda_{r}$ be the eigenvalues of $\pi(\mathfrak{p}; \bs{c})$.
Then we obtain $\max \lambda_l \leq \kappa^{|\mathfrak{p}|}$ because of the properties of operator norms.
From this inequality, if $|z \kappa|<  \alpha^{-1} \leq 1$, we obtain
\begin{equation}
\left|  \det(I- z^{|\mathfrak{p}|} \pi(\mathfrak{p}; \bs{c})) \right|
= \left| \prod_l (1-z^{|\mathfrak{p}|} \lambda_l) \right|
\geq (1-|z \kappa|^{|\mathfrak{p}|})^{r}.
\end{equation}

Therefore, if $|z |< \kappa^{-1} \alpha^{-1}$, 
\begin{align*}
 \left| \zeta_{H}(z \bs{c}) \right|
= \left| \prod_{ \mathfrak{p} \in P } \det(I- z^{|\mathfrak{p}|} \pi(\mathfrak{p}; \bs{c}))^{-1}  \right|
\leq   \prod_{ \mathfrak{p} \in P } (1- |z \kappa|^{|\mathfrak{p}|})^{-r} 
=\zeta_{H}(|z \kappa|)^{r} < \infty.
\end{align*}
\end{proof}

The following theorem gives an explicit condition of the
region where the Hessian is positive definite in terms of the
\ccms of the pseudomarginals.

\begin{thm} 
\label{cor:positivedefiniteregion}
Let $\mathcal{I}$ be a multinomial or a fixed-mean Gaussian \ifa.
Let $\alpha$ be the Perron-Frobenius eigenvalue of $\mathcal{M}$ and
define 
\begin{equation*}
 L_{\alpha^{-1}}(\mathcal{I}):=\left\{ \beliefsw \in L(\mathcal{I})~|~
\forall \alpha \in F, ~\forall i,j \in \alpha, \hspace{2mm}
\norm{\corr{b_{\alpha}}{\phi_i}{\phi_j}} < \alpha^{-1} \right\}.
\end{equation*}
Then, the
Hessian $\nabla^2 F$ is positive definite on $L_{\alpha^{-1}}(\mathcal{I})$.
\end{thm}
\begin{proof}
Obviously, $\kappa < \alpha^{-1}$ holds in $L_{\alpha^{-1}}$. 
From Lemma \ref{lem:specbound}, 
$\spec{ \matmu } \subset \{ \lambda \in \mathbb{C} |\hspace{2mm} |\lambda| < 1 \}.$
From Theorem \ref{thm:positive}, the Hessian is positive definite.
\end{proof}

Properties of the Perron-Frobenius eigenvalue of $\mathcal{M}$, such as bound, is 
discussed in Subsection \ref{sec:directededgematrix}.
Roughly speaking, as degrees of factors and vertices increase,
the $\alpha$ also increases and thus $L_{\alpha^{-1}}$ shrinks.
The region $L_{\alpha^{-1}}$ depends on the choice of the operator norms.
If we chose the norm $\norm{\cdot}_2$,
this theorem immediately implies the convexity of 
the \Bfe function of tree and one-cycle hypergraphs as we will discuss in the next subsection.

\subsection{Convexity condition}
\label{sec:convexitycondition}
Since the Hessian $\Hesse F$ does not depend on the given compatibility functions $\Psi=\{ \Psi_{\alpha}\}$,
the convexity of $F$ solely depends on the given \ifa and the factor graph.
For multinomial case, Pakzad and Anantharam have shown that the \Bfe function is convex if the factor graph has
at most one cycle \cite{PAstat}.
The following theorem extends the result.
\begin{thm}
 \label{thm:convexcondition}
Let $\mathcal{I}$ be a multinomial or fixed-mean Gaussian \ifa associated with 
a connected factor graph $H$. Then
\begin{equation}
 F \text{ is convex on } L \iff n(H)=0 \text{ or } 1.
\end{equation}
\end{thm}
\begin{proof}
$(\Leftarrow)$ 
Here, we give a proof based on Theorem \ref{cor:positivedefiniteregion}, which assumes the \ifa is multinomial or fixed-mean Gaussian.
However, this direction of the statement is valid for any \ifa. 
(See appendix \ref{app:sec:convexity}.)

From Proposition \ref{prop:PFboundM},
the Perron-Frobenius eigenvalue $\alpha$ is equal to $1$ if $n(H)=0$
and $0$ if $n(H)=1$.
Using Theorem \ref{cor:positivedefiniteregion} with norm $\norm{\cdot}_2$ and Proposition \ref{app:prop:normccm},
we obtain $L_{\alpha^{-1}}=L$.
Therefore, the \Bfe function is convex over the domain $L$.\\
$(\Rightarrow)$ We prove for each \ifa.\\
{\bf Fixed-mean Gaussian case}\\
Let $G=(V,E)$ be a graph.
For $t \in [0,1]$, let us define $\eta_{ii}(t):=1$ and $\eta_{ij}(t):=t$.
Accordingly, $u^{ij}_{\edij}=t^2$ and $\bseta(t) \in L$.
As $t \nearrow 1$, $\bseta(t)$ approaches to a boundary point of $L$.
From Theorem \ref{thm:Hashimoto},
\begin{align*}
 \det(\nabla^2 F(t)) (1-t^2)^{2|E|+|V|-1}
&= 2^{-|V|}
Z_{G}(t^2)^{-1}
(1-t^2)^{-|E|+|V|-1} \\
&
\longrightarrow
-2^{|E|-2|V|+1}(|E|-|V|) \kappa(G) \quad (t \rightarrow 1).
\end{align*}
If $n(G)=|E|-|V|+1 > 1$, the limit value is negative.
Therefore, at least in a neighborhood of the limit point, $\Hesse F$ is not positive definite.\\
{\bf Multinomial case}\\
First, we consider binary case, i.e. $\phi_{i}(x_i)=x_i \in \{ \pm 1\}$.
For $t \in [0,1]$, let us define $\eta_{ij}(t):=t$ and $\eta_{i}(t):=0$.
Accordingly, $u^{\alpha}_{\edij}=t$ and $\bseta(t) \in L$.
As $t \nearrow 1$, $\bseta(t)$ approaches to a boundary point of $L$.
Using Theorem \ref{thm:Hashimoto}, analogous to the fixed-mean Gaussian case,
we see that $\det(\nabla^2 F(t))$ becomes negative as $t \rightarrow 1$ if  $n(H)=1-\chi (H) > 1$.
Therefore, $F$ is not convex on $L$.

For general multinomial \ifas, the non convexity of $F$ is deduced from the binary case.
There is a face of (the closure of) $L(\mathcal{I})$ 
that is identified with the set of pseudomarginals of the binary \ifa on the same factor graph.
From Eq.~(\ref{eq:BFEmultinomial}) and $0 \log 0 =0$, 
we see that the restriction of $F$ on the face is the \Bfe function of the binary \ifa.
Since this restriction is not convex, $F$ is not convex.
\end{proof}

\section{Stability of LBP}
\label{sec:stability}
In this section we discuss the local stability of LBP and the local
structure of the Bethe free energy around an LBP fixed point.
In the celebrated paper \cite{YFWGBP}, which introduced the equivalence between LBP and the Bethe approximation,
Yedidia et al empirically found that locally stable LBP fixed points are local minima of the \Bfe function.
Heskes have shown  that a locally stable fixed point of arbitrary damped LBP is a local minima of the \Bfe function 
for the multinomial models \cite{Hstable}.  
In this section, we extend the property to the fixed-mean Gaussian cases, 
applying our spectral conditions of stability and positive definiteness.
Since the converse of the property is not necessarily true in general, we also elucidate
the gap.

First, we regard the LBP update as a dynamical system.
As discussed in Section \ref{sec:LBPstability},
the (parallel) LBP algorithm can be formulated as 
repeated applications of a map $T$ on the \nparas of the messages.
Explicit expression of this map is given in Eq.~(\ref{eq:eparaLBPupdate}).
The differentiation $T'$ at the fixed point, which is computed in Theorem \ref{thm:diffofLBP},
is rewritten as follows.
\begin{prop}
\label{prop:diffTisM}
At an LBP fixed point $\bseta \in L$,
\begin{equation}
 T' = \matmu,
\end{equation}
where $\bsu=\{  u^{\alpha}_{\edij} \}$ is given by Eq.~(\ref{def:u}).
\end{prop}

\subsection{Spectral conditions}
Let $T$ be the LBP update map.
A fixed point $\bsmu^{*}$ is called {\it locally stable}\footnote{ This property is often referred to as {\it asymptotically stable} \cite{GHnonlinear}.}
if LBP starting with a point sufficiently close to $\bsmu^{*}$ converges to $\bsmu^{*}$.
To suppress oscillatory behaviors of LBP,
{\it damping} of update $T_{\epsilon }:=(1- \epsilon ) T+ \epsilon I$
is sometimes useful, where $0 \leq \epsilon < 1$ is a damping strength
and $I$ is the identity matrix.

As we will summarize in the following theorem,
the local stability is determined by the linearization $T'$ at the fixed point.
Since $T'$ is nothing but $\matmu$ at an LBP fixed point,
the \Bzf naturally derives relations between the local stability and the Hessian of the \Bfe function.

\begin{thm}
Let $\mathcal{I}$ be a multinomial or a fixed-mean Gaussian model.
Let $\bsmu^{*}$ be an LBP fixed point and assume that $T'(\bsmu^{*})$ has no eigenvalues of unit modulus for simplicity.
Then the following statements hold.
\begin{enumerate}
 \item $\spec{T'(\bsmu^{*})}  \subset \{ \lambda \in \mathbb{C}| |\lambda| < 1\}$ $\iff$ LBP is locally stable at $\bsmu^{*}$.
 \item $\spec{T'(\bsmu^{*})}  \subset \{ \lambda \in \mathbb{C}| {\rm  Re}\lambda < 1\}$ \hspace{-2mm} $\iff$ \hspace{-2mm}
       LBP is locally stable at $\bsmu^{*}$ with some damping.
 \item $\spec{T'(\bsmu^{*})}  \subset \mathbb{C} \smallsetminus \mathbb{R}_{\geq 1}$ $\Rightarrow$
       $\bsmu^{*}$ is a local minimum of BFE. 
\end{enumerate}
\end{thm}
\begin{proof}
$1.:$ This is a standard result. (See \cite{GHnonlinear} for example.)
~$2.:$
There is an $\epsilon \in [0,1)$ that satisfy $\spec{T_{\epsilon}'(\bsmu^{*})}  \subset \{ \lambda \in \mathbb{C}| |\lambda| < 1 \}$
if and only if
$\spec{T'(\bsmu^{*})}  \subset \{ \lambda \in \mathbb{C}| {\rm  Re}\lambda < 1\}$.
~$3.:$
This assertion is a direct consequence of Theorem \ref{thm:positive} and Proposition \ref{prop:diffTisM}.
\end{proof}

This theorem immediately resolve the conjecture of Yedidia et al \cite{YFWGBP}:
locally stable LBP fixed points are local minima of the \Bfe.
Since they only discusses multinomial models, their experiments seems to have been performed for the models. 
Extending the statement, 
the theorem implies that, for both multinomial and fixed-mean Gaussian cases,
locally stable fixed points of arbitrary damped LBP are local minima of the \Bfe function.

Heskes \cite{Hstable} have proved that locally stable fixed points are ``local minima of the \Bfe,''
where the \Bfe function is defined on a restricted set. 
We give the following remarks 
to see that their result actually resolves the Yedidia's conjecture.
He considers the \Bfe function on a subset of $L$;
\begin{equation*}
 S := \left\{ 
\{\pa{\eta} , \eta_i\} \in L ~| 
 ~{}^{\forall} \alpha = \{i_1, \ldots, i_k\} \in F,
\ \pa{\eta} = \Lambda_{\alpha} \pa{ \left( \pa{\bar{\theta}}, \Lambda_{i_1}(\eta_{i_1}), \ldots ,  \Lambda_{i_1}(\eta_{i_k}) \right)  }
~\right\},
\end{equation*}
where $\pa{\bar{\theta}}$ is a constant given by the compatibility function from Eq.~(\ref{eq:asm:modelindludes}).
In other words, $S$ is a set of pseudomarginals $\beliefsw$ that satisfies
\begin{equation*}
 ~{}^{\forall} \alpha \in F, ~ {}^{\exists} \{ \va{\theta'}{i} \}_{i \in \alpha},
\quad
 b_{\alpha}(x_{\alpha}) \propto  \exp 
( 
\inp{ \pa{\bar{\theta}} }{\phi_{\alpha}(x_{\alpha})} + \sum_{i \in \alpha} \inp{  \va{\theta'}{i} }{\phi_i(x_i)}
).
\end{equation*}
Obviously, all LBP fixed points are in $S$.
We can take a coordinate $\{\eta_i\}_{i \in V}$ of $S$ because $\eta_{\alpha}$ is a function of $\{\eta_i\}_{i \in \alpha}$.
The restriction of the \Bfe function $F$ to this set is denoted by $\hat{F}$.
It is straightforward to check that the stationary points of $\hat{F}$ correspond to the LBP fixed points, that is,
\begin{equation*}
 \pd{\hat{F}(\{\eta_i\})}{\eta_j}=0  \quad {}^{\forall} j \in V \iff 
 \{ \eta_i \} \in S   \text{ is an LBP fixed point.} 
\end{equation*}
In \cite{Hstable}, for multinomial models, Heskes have shown that if $\{\eta_i\} \in S$ is a locally stable fixed point of (arbitrary damped) LBP,
it is a local minimum of $\hat{F}$.
This statement is equivalent to the statement replaced by ``local minima of $F$.''
In fact, we can easily check that the positive definiteness of the Hessian of $\hat{F}$ is equivalent to
that of $F$.

It is interesting to ask under which condition a local minimum of the
\Bfe function is a locally stable fixed point of (damped) LBP. 
An implicit reason for the empirical success of the LBP algorithm is that
LBP finds a ``good'' local minimum rather than a local minimum nearby the initial point. 
The theorem gives a partial answer to the question, i.e., the difference between
stable local minima and unstable local minima, in terms of the spectrum of $T'(\bsmu^{*})$.
It is noteworthy that we do not know whether the converse of the third statement holds.

\subsection{Pairwise binary case}
Here we focus on binary pairwise attractive models.
In this special case, the stable fixed points of LBP and the local minima of \Bfe function
are less different.

The given graphical model $\Psi=\{\Psi_{ij},\Psi_i\}$ is called attractive if $J_{ij}\geq 0$,
where  $\Psi_{i}(t)=\exp( h_i x_i )$ and $\Psi_{ij}(x_i,x_j)= \exp( J_{ij} x_i x_j)$ $~(x_i,x_j \in \{ \pm 1\})$.
The following theorem implies that if a stable fixed
point becomes unstable by changing $J_{ij}$ and $h_i$, the
corresponding local minimum also disappears.

\begin{thm}
\label{thmattractive}
Let us consider
continuously parameterized attractive models
$\{\Psi_{ij}(t),\Psi_{i}(t)\}$,
e.g. $t$ is a temperature: $\Psi_{ij}(t)=\exp(t^{-1} J_{ij}x_i x_j)$
and
$\Psi_{i}(t)=\exp(t^{-1} h_i x_i )$.
For given $t$, run LBP algorithm and find a (stable) fixed point.
If we continuously change $t$ and
see the LBP fixed point becomes unstable across $t=t_0$,
then the corresponding local minimum of the Bethe free energy becomes
a saddle point across $t=t_0$.
\end{thm}
\vspace{-4mm}
\begin{proof}
From Eq.~(\ref{eq:defbelief2}),
we see that 
$b_{ij}(x_i,x_j) \propto \exp (J_{ij}x_i x_j + \theta_i x_i+ \theta_j x_j)$
for some $\theta_i$ and $\theta_j$.
From $J_{ij} \geq 0$,
we have $\cov{b_{ij}}{x_i}{x_j} \geq 0$, and thus $u_{\edij} \geq 0$.
When the LBP fixed point becomes unstable,
the Perron-Frobenius eigenvalue of $\matmu$ goes over $1$, which means
$\det(I-\matmu)$ crosses $0$.
From Theorem \ref{thm:BZ}, we see that
$\det(\nabla^{2} F)$ becomes positive to negative at $t=t_0$.
\end{proof}
Theorem \ref{thmattractive} extends Theorem 2 of \cite{MKproperty}, which discusses only
the case of vanishing local fields $h_i=0$
and the trivial fixed point (i.e. $\E{b_i}{x_i}=0$).

\section{Discussion}
This chapter developed the Bethe-zeta formula
for general \ifas including multinomial and Gaussian families.
The formula says that the determinant of the Hessian of the \Bfe function is 
the reciprocal of the graph zeta function up to positive factors.
The underlying mathematical structure that makes the \Bzf hold was the two dualistic variational
characterizations of the LBP fixed points.
In the proof of the formula, we utilized the languages of exponential families and the graph zeta function,
including dual convex functions and the \IB type determinant formula.
The key condition satisfied on the set $L$ was $\var{b_{\alpha}}{\phi_i}=\var{b_i}{\phi_i}$.

In Section \ref{sec:PDC} and \ref{sec:stability}, we discussed applications of the formula,
demonstrating its utility.
First, we applied the formula to analyze the positive definiteness of the Bethe free energy function,
focusing on multinomial and fixed-mean Gaussian \ifa.
Our analysis showed that the region, where the Hessian of the \Bfe function is positive definite,
shrinks as the pole of the Ihara zeta function $\alpha^{-1}$ approaches to zero.
Secondly, we analyzed the stability of the LBP algorithm.
At an LBP fixed point, the matrix $\matmu$, appears in the first determinant formula of the graph zeta function,
is equal to the linearization of the LBP update.
This connection derives the relation between the local minimality and the local stability of the fixed point.   
Applying our spectral conditions, we gave a simple proof of the Yedidia's conjecture:
locally stable fixed points of LBP are local minima of the \Bfe function.

The \Bzf shows that the \Bfe function contains information on the graph geometry, especially 
on the prime cycles.
At the same time, the formula helps to extract graph information from the \Bfe function.
We observed that some values derived from the \Bfe function are related to
graph characteristics such as the number of the spanning trees.

For a tree structured hypergraph, LBP algorithm and the \Bfe function are in a sense trivial; the graph zeta function is also trivial for a tree.
The results of this chapter provide concrete mathematical relations between these trivialities.
For example, $\zeta_{H}(u)=1$ implies that there is no pole, and thus $\alpha^{-1}=\infty$.
Therefore, the \Bfe function is convex and the linearization $T'$ at the unique fixed point is a nilpotent matrix,
which is necessary for the finite step termination of the LBP algorithm.

%% file: chapter5.tex

\section{Introduction}
This chapter provides a new approach to analyze the uniqueness of the LBP fixed point.
Since the LBP fixed points are the solutions of the LBP update equation,
it is natural to ask whether there is a solution; if exist, is it unique?
For multinomial cases, at least one fixed point exists
because the \Bfe function is bounded from below and
a stationary point of the \Bfe function is an LBP fixed point \cite{YFWconstructing}.
Furthermore, if the underlying hypergraph is tree or one-cycle, the solution is unique \cite{W1loop};
this result is obvious as we have shown the convexity of the \Bfe functions for these hypergraphs in Section \ref{sec:convexitycondition}.

From the viewpoint of approximate inference, the uniqueness of LBP fixed point is a preferable property.
Since LBP algorithm is interpreted as the variational problem of the \Bfe function,
an LBP fixed point that correspond to the global minimum is believed to be the best one.
If we find the unique fixed point of the LBP algorithm, it is guaranteed to be the global minimum.

For multinomial models, there are several works that give sufficient conditions for the uniqueness property.
In \cite{Huniquness}, Heskes analyzed the uniqueness problem by
considering equivalent minimax problem.
Other authors analyzed the convergence property rather than the uniqueness.
The LBP algorithm is said to be {\it convergent} if the messages converge to the unique fixed point
irrespective of the initial messages.
By definition, this property is stronger than the uniqueness.
Tatikonda et al. \cite{TJgibbsmeasure} utilized the theory of Gibbs measure.
They have shown that the uniqueness of the Gibbs measure implies the convergence of 
LBP algorithm.
Therefore, known sufficient conditions of the uniqueness of the Gibbs measure are
that of the convergence of LBP algorithm.
Ihler et al. \cite{IFW} and Mooij et al. \cite{MKsufficient} derived sufficient conditions
for the convergence by investigating conditions that make the LBP update a contraction.

In this chapter, focusing on binary pairwise models,
we propose a novel differential topological approach to this problem.
Empirically, one of the strongest condition that is applicable to arbitrary graph is the spectral condition by Mooij et al. \cite{MKsufficient}.
Our approach is powerful enough to reproduce the condition.
For hypergraphs with nullity two, we prove the uniqueness property even though LBP is not necessarily convergent.
Although our discussions in this chapter are restricted to binary pairwise models,
the method can be basically extended to multinomial models.

\subsection{Idea of our approach}
In our approach, in combination with the Bethe-zeta formula,
the index sum theorem is the basic apparatus.
Conceptually, the theorem has the following form:
\begin{equation}
 \sum_{\bseta \text{: LBP fixed point}} \hspace{-2mm} {\rm Index}(\bseta)=1, \label{eq:conceptindex}
\end{equation}
where ${\rm Index}(\bseta)$ is $+1$ or $-1$ and determined by local properties of the \Bfe function at $\bseta$.
The sum is taken over all fixed points of LBP, that is,
all stationary points of the \Bfe function $F$.
If we can guarantee that the index
of any fixed point is $+1$
in advance of running LBP,
we conclude that the fixed point of LBP is unique.

The formula (\ref{eq:conceptindex}) might look surprising,
but such formulas, that connect the global and the local structure, 
are often seen in differential topology.
The simplest example that illustrates the idea of the theorem is sketched
in Figure \ref{figureExampleIndexFormula1}.
In this example, the sum is taken over all the stationary points of the function
and the indices are assigned depending on the sign of 
the second derivative at the points.
When we deform the objective function, the sum is still equal to one
as long as the outward gradients are positive at the boundaries
(see Figure \ref{figureExampleIndexFormula2}).
This example suggests that the important feature for the formula is
the behavior of the function near the boundary of the domain.

\begin{figure}
\begin{minipage}{.5\linewidth}
\begin{center}
\includegraphics[scale=0.3]{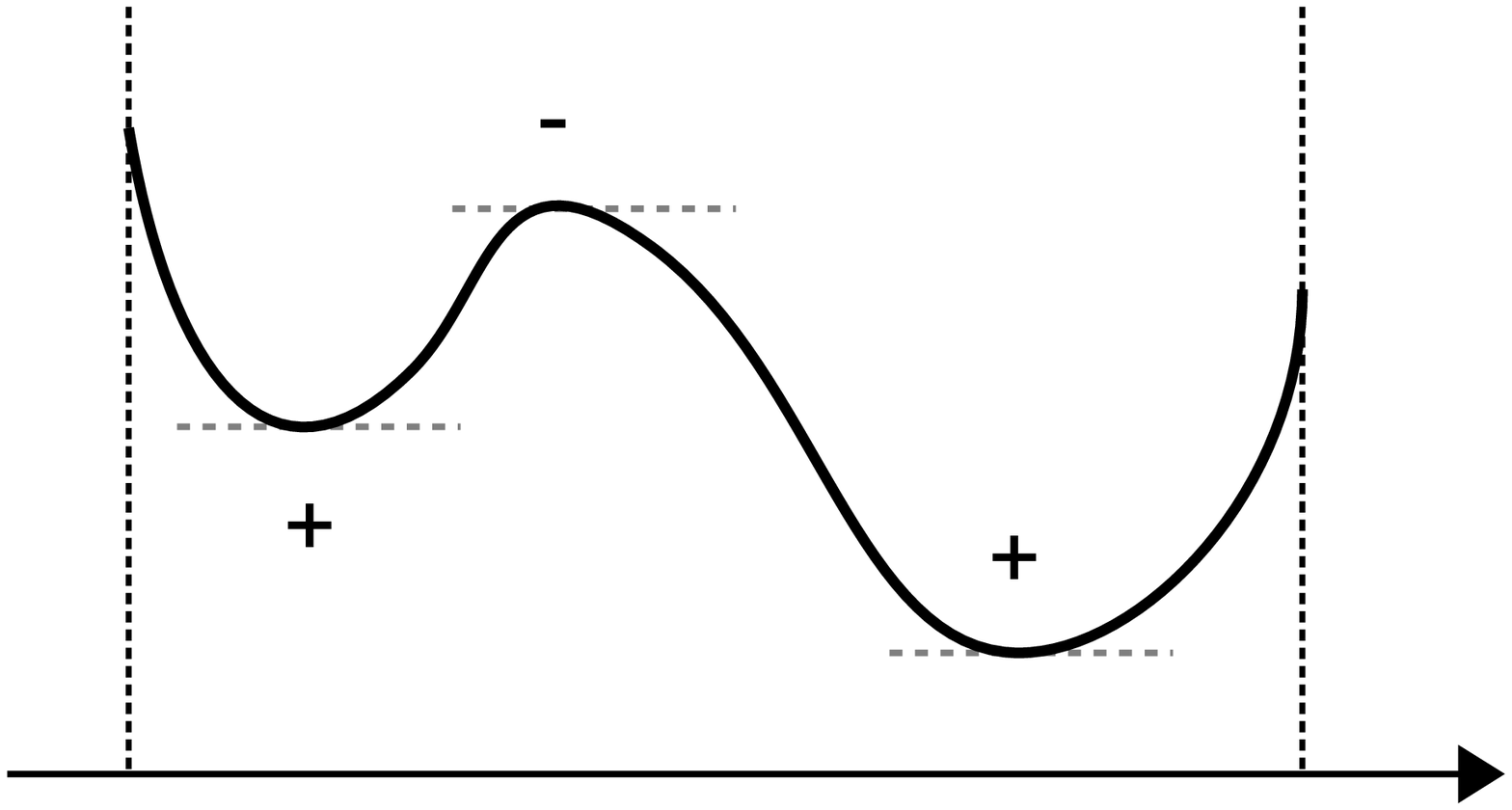} 
\caption{The sum of indices is one. \label{figureExampleIndexFormula1}}
\end{center}
\end{minipage}
\begin{minipage}{.5\linewidth}
\begin{center}
\includegraphics[scale=0.3]{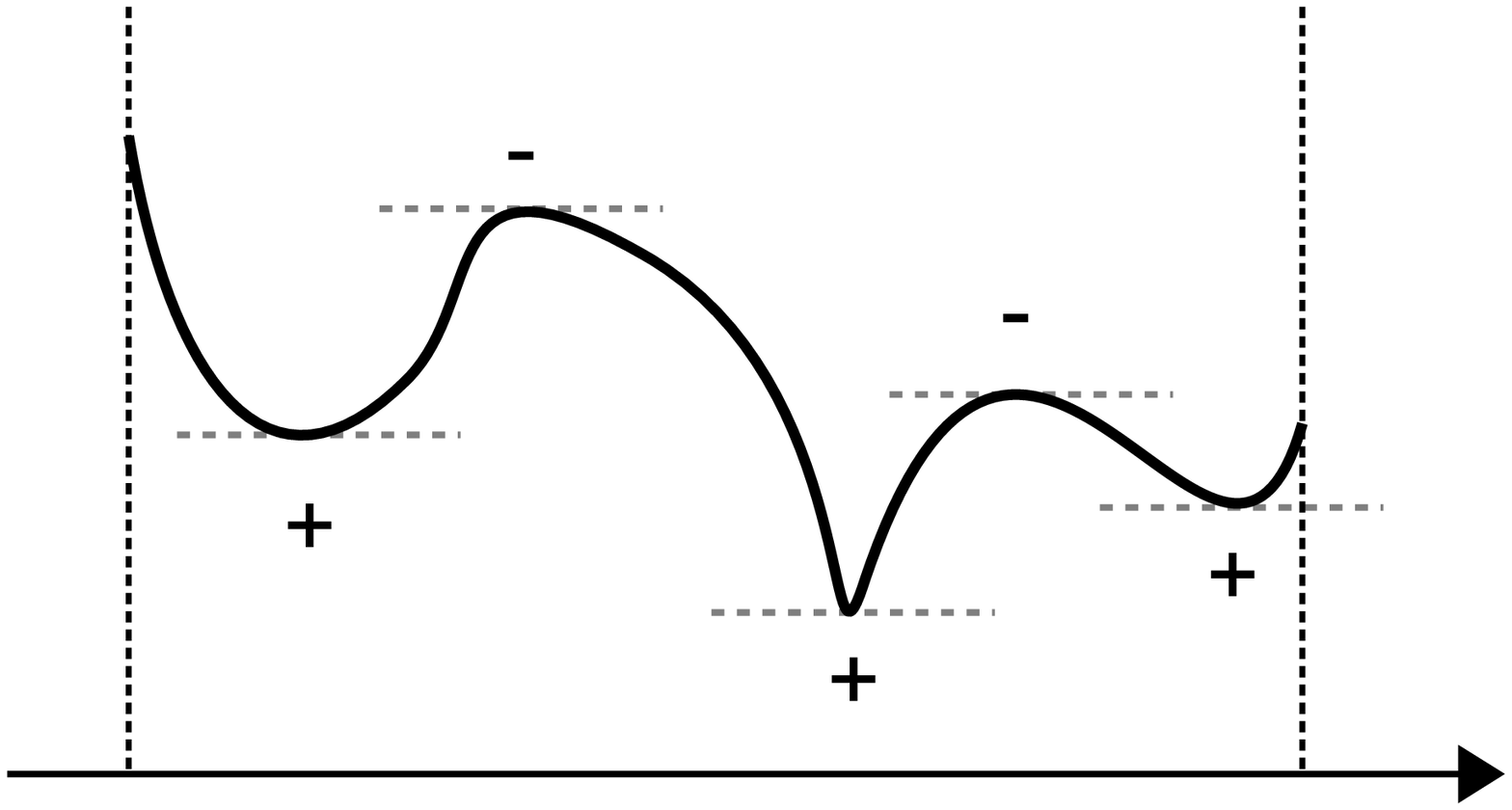} 
\caption{The sum of indices is still one. \label{figureExampleIndexFormula2}}
\end{center}
\end{minipage}
\end{figure}

Simsek et al. have shown a index sum formula, called generalized Poincare-Hopf theorem \cite{SOAgeneralizedPoincare-Hopf}.
In this formula, the indices of the stationary points in the (not necessarily smooth) boundary are summed as well as those in the interior.
The theorem is applied to show the uniqueness of stationary point in non-convex optimization problems such as 
Nash equilibrium \cite{SOAgeneralizedPoincare-Hopf} and
network equilibrium \cite{SOAbox,TWLCequilibrium}.
However, their theorem can not be applied to our \Bfe function because
the behavior of the function near the boundary is so complicated that it is difficult to handle.

We prove the index sum formula utilizing a property of the \Bfe function:
the gradient of the \Bfe function diverges as a point approaches to the boundary of the domain.

\subsection{Overview of this chapter}
This chapter is organized as follows.
In Section \ref{sec:indexsum},
we prove the index sum formula of the \Bfe function using two lemmas.
The first lemma describes an important property of the \Bfe function: 
the divergence of the norm of the gradient vector at the boundary of the domain.
The second lemma is a standard result in differential topology.
The index sum formula, combined with the \Bzf, provides a powerful method of proving the uniqueness;
we will prove the following two results utilizing this method.
Section \ref{sec:uniqueness} proves a uniqueness condition for general graphs, which is a reproduction of 
the condition by Mooij et al. \cite{MKsufficient}. 
In Section \ref{sec:uniqueness2loop}, we focus on graphs of nullity two and prove that the fixed point of LBP is unique
if it is not attractive.

\section{Index sum formula}\label{sec:indexsum}
The purpose of this section is to show the index sum formula presented in the following theorem.
Throughout this chapter, the \ifa is binary pairwise 
and the given graphical model is 
\begin{equation}
   p(x)=\frac{1}{Z} \exp(\sum_{ij \in E} J_{ij} x_i x_j  + \sum_{i \in V} h_i x_i), \label{defIsing2}
\end{equation}
where $G=(V,E)$ is a graph and $x_i = \pm 1$.

\begin{thm}
\label{thm:indexsum}
Let $F$ be the \Bfe function on the set of pseudomarginals $L$.
If $\det \Hesse F(\bseta)\neq0$ is satisfied for all $\bseta \in (\nabla F)^{-1}(0)$,
then
\begin{equation}
 \sum_{\bseta \in (\nabla F)^{-1}(0) } \sgn( \det \Hesse F(\bseta) )=1, \label{eq:indexsum}
\end{equation}
where
\begin{equation*}
 \sgn(x):=
\begin{cases}
1  &\text{ if } x>0, \\
-1 & \text{ if } x<0.
\end{cases}
\end{equation*}
We call each summand, which is $+1$ or $-1$, the {\it index} of $F$ at $\bseta$.
\end{thm}
Note that the set $(\nabla F)^{-1}(0)$, which is the stationary points of the
Bethe free energy function, coincides with the set of fixed points of LBP.
The condition $\det \Hesse F(\bseta)\neq0$ in the statement is not a strong requirement.
In fact, the Hessian is invertible except for a measure-zero region of $L$ and 
an LBP fixed point does not happen to be in the region generally.

The above theorem asserts that the sum of indices
of all the fixed points must be one.
As a consequence, the number of the fixed points of LBP is always odd.
Not rigorously speaking, the formula implies the existence of LBP fixed points
because if there is no LBP fixed point, the L.H.S. of Eq.~(\ref{eq:indexsum}) is equal to zero.
A rigorous proof of the existence is given by bounding the \Bfe function from below \cite{YFWconstructing}.

This formula is generalized to multinomial models straightforwardly.
For Gaussian models, however, this kind of formula does not hold.
Indeed, even for one-cycle graphs, LBP fixed points do not exist 
if the compatibility functions are, in a sense, strong \cite{NWaccuracy}.

\subsection{Two lemmas}
For the proof, we need two lemmas.
The first lemma shows the divergent behavior of the gradient of the \Bfe function near the boundary
of the domain.

\begin{lem}
\label{lem:gradBethe}
If a sequence $\{\bseta_n \} \subset L$ converges to a point
$\bseta_{*} \in \partial L$, then
\begin{equation}
 \norm{\nabla F(\bseta_n)} \rightarrow \infty,
\end{equation}
where $\partial L$ is the boundary of $L \subset \mathbb{R}^{|V|+|E|}$.
\end{lem}
\begin{proof}
Since our model is pairwise and binary,
we choose the sufficient statistics as
$\px{\phi}{ij}(x_i,x_j)=x_i x_j$ and $\phi_{i}(x_i)=x_i$ ($x_i \in \{ \pm 1\}$).
We use notations $m_i:=\eta_i = \E{b_i}{x_i}$ and $\chi_{ij}:=\px{\eta}{ij}=\E{b_{ij}}{x_i x_j}$.
Then the Bethe free energy Eq.~(\ref{defn:Bfe}) is rewritten as
\begin{align}
F(\{m_i,\chi_{ij}\})=
-
\sum_{ij \in E}J_{ij}\chi_{ij}
-
\sum_{i \in V}h_i m_i
+
\sum_{ij \in E} & \sum_{x_i x_j = \pm 1}
\eta \Big(\frac{1+ {m_i x_i} + { m_j x_j }+
 \chi_{ij} x_i x_j}{4} \Big)  \nonumber \\
&+
\sum_{i \in V} (1-d_i) \sum_{x_i = \pm 1}
 \eta \Big(\frac{1+ {m_i x_i}}{2} \Big),
 \label{eq:BetheFreeEnergybinary}
\end{align}
where $\eta (x):=x \log x$. 
The domain of $F$ is written as
\begin{equation}
L:=
\Big\{
\{m_i,\chi_{ij}\} \in \mathbb{R}^{|V|+|E|}|
1+m_i x_i + m_j x_j + \chi_{ij}x_i x_j  > 0
\text{   for all } ij \in E \text{ and } x_i,x_j= \pm 1
\Big\}.\nonumber
\end{equation}
The first derivatives of the \Bfe function are
\begin{align}
\pd{F}{m_i}& = -h_i +
(1-d_i) \frac{1}{2} \sum_{x_i = \pm 1} x_i \log b_i(x_i)
+ \frac{1}{4}\sum_{k \in N_i}\sum_{x_i,x_k=\pm 1}x_i \log
 b_{ik}(x_i,x_k), \label{Bethem} \\
\pd{F}{\chi_{ij}}& = -J_{ij} +
\frac{1}{4} \sum_{x_i,x_j=\pm 1}x_i x_j \log  b_{ij}(x_i,x_j). \label{Bethechi}
\end{align}

Note that it is enough to prove the assertion
when $h_i=0$ and $J_{ij}=0$.
We prove by contradiction.
Assume that $\lVert \nabla F(\bseta_n) \rVert \not \rightarrow \infty$.
Then, there exists $R > 0$ such that 
$\lVert \nabla F(\bseta_n) \rVert \leq R$
for infinitely many $n$.
Let 
$B_0(R)$ be the closed ball of radius $R$ centered at the origin.
Taking a subsequence, if necessary,
we can assume that
\begin{equation}
 \nabla F(\bseta_n) \rightarrow 
{}^{\exists}
\binom{\boldsymbol{\kappa}}{\boldsymbol{\lambda}}
\in 
B_0(R), \label{eq:gradFconverges}
\end{equation}  
because of the compactness of $B_0(R)$.
Let $b_{ij}^{(n)}(x_i,x_j)$ and  $b_{i}^{(n)}(x_i)$ be the pseudomarginals
 corresponding to $\bseta_n$.
Since $\bseta_n \rightarrow  \bseta_{*} \in \partial L$,
there exist $ij \in E$, $x_i$ and $x_j$ such that
\begin{equation*}
b_{ij}^{(n)}(x_i,x_j) \rightarrow 0.
\end{equation*}
Without loss of generality, we assume that
$x_i=+1$ and $x_j=+1$.
From Eq.~(\ref{eq:gradFconverges}), we have
\begin{equation}
\nabla F(\bseta_n)_{ij}
=
\frac{1}{4}
\log
\frac{b_{ij}^{(n)}(+,+)b_{ij}^{(n)}(-,-)}
{b_{ij}^{(n)}(+,-)b_{ij}^{(n)}(-,+)}
\longrightarrow
 \lambda_{ij}. \label{eqgradFconvergesij}
\end{equation}
Therefore
$b_{ij}^{(n)}(+,-) \rightarrow 0$ or 
$b_{ij}^{(n)}(-,+) \rightarrow 0$ holds;
we assume $b_{ij}^{(n)}(+,-) \rightarrow 0$
without loss of generality.
Now we have 
\begin{equation*}
b_{i}^{(n)}(+)=b_{ij}^{(n)}(+,-)+b_{ij}^{(n)}(+,+) \rightarrow 0.
\end{equation*}

In this situation, the following claim holds.
\begin{claim}
Let $k \in N_i$.
In the limit of $n \rightarrow \infty$,
\begin{equation}
\sum_{x_i,x_k=\pm 1}x_i \log \frac{b^{(n)}_{ik}(x_i,x_k)}{b^{(n)}_i(x_i)}
=
\log
\left[
\frac{b_{ik}^{(n)}(+,+)b_{ik}^{(n)}(+,-)b_{i}^{(n)}(-)^2}
{b_{ik}^{(n)}(-,+)b_{ik}^{(n)}(-,-)b_{i}^{(n)}(+)^2}
\right] \label{eqclaim1}
\end{equation}
converges to a finite value.
\end{claim}
\begin{proof}[proof of claim]
From $b_{i}^{(n)}(+)  \rightarrow 0$, we have 
\begin{equation*}
b_{ik}^{(n)}(+,-), b_{ik}^{(n)}(+,+) \longrightarrow 0
\quad \text{ and } \quad
b_{i}^{(n)}(-)  \rightarrow 1.
\end{equation*}
\textbf{Case 1:}
$b_{ik}^{(n)}(-,+) \longrightarrow b_{ik}^{*}(-,+) \neq 0$ and
$b_{ik}^{(n)}(-,-) \longrightarrow b_{ik}^{*}(-,-) \neq 0$. \\
In the same way as Eq.~(\ref{eqgradFconvergesij}),
\begin{equation*}
\nabla F(\bseta_n)_{ik}
=
\frac{1}{4}
\log
\frac{b_{ik}^{(n)}(+,+)b_{ik}^{(n)}(-,-)}
{b_{ik}^{(n)}(+,-)b_{ik}^{(n)}(-,+)}
\longrightarrow
 \lambda_{ik}. %
\end{equation*}
Therefore
\begin{equation*}
\frac{b_{ik}^{(n)}(+,+)}
{b_{ik}^{(n)}(+,-)}
\longrightarrow
{}^{\exists}r \neq 0.
\end{equation*}
Then we see that Eq.~(\ref{eqclaim1}) converges to a finite value. \\
\textbf{Case 2:}
$b_{ik}^{(n)}(-,+) \longrightarrow 1$ and
$b_{ik}^{(n)}(-,-) \longrightarrow 0$. \\
Similar to the case 1, we have
\begin{equation*}
\frac{b_{ik}^{(n)}(+,+)b_{ik}^{(n)}(-,-)}
{b_{ik}^{(n)}(+,-)}
\longrightarrow
{}^{\exists}r \neq 0.
\end{equation*}
Therefore
$\frac{b_{ik}^{(n)}(+,-)}{b_{ik}^{(n)}(+,+)} \rightarrow 0$.
This implies that
$\frac{b_{i}^{(n)}(+)}{b_{ik}^{(n)}(+,+)} \rightarrow 1$.
Then we see that Eq.~(\ref{eqclaim1}) converges to a finite value.\\
\textbf{Case 3:}
$b_{ik}^{(n)}(-,+) \longrightarrow 0$ and
$b_{ik}^{(n)}(-,-) \longrightarrow 1$. \\
Same as the case 2.
\end{proof}
Now let us get back to the proof of Lemma \ref{lem:gradBethe}.
We rewrite Eq.~(\ref{Bethem}) as
\begin{equation}
\nabla F(\bseta_n)_{i}
=
\frac{1}{2}\log b^{(n)}_i(+)
-
\frac{1}{2}\log b^{(n)}_i(-)
+ \frac{1}{4}\sum_{k \in N_i}\sum_{x_i,x_k=\pm 1}x_i \log
\frac{b^{(n)}_{ik}(x_i,x_k)}{b^{(n)}_{i}(x_i)}  \label{eq:lemma1finaleq}
\end{equation}
From Eq.~(\ref{eq:gradFconverges}), this value converges to $\kappa_i$.
The second and the third terms in Eq.~(\ref{eq:lemma1finaleq}) 
converges to a finite value, while the
 first value converges to infinity.
This is a contradiction.
Therefore Lemma \ref{lem:gradBethe} is proved.
\end{proof}

The following lemma is a standard result in differential
topology, and utilized in the proof of Theorem \ref{thm:indexsum}.
We refer Theorem 13.1.2 and comments in p.104 of \cite{DFN} for the proof of this lemma.

\begin{lem}
\label{lem:degreeofmap}
Let $M_1$ and $M_2$ be compact, connected and orientable manifolds
with boundaries.
Assume that the dimensions of $M_1$ and $M_2$ are the same.
Let $f:M_1 \rightarrow M_2 $ be a smooth map satisfying
$f(\partial M_1) \subset \partial M_2 $.
A point $p \in M_2$ is called a {\it regular value} if
$\det(\nabla f(q))\neq 0$ for all $q \in f^{-1}(p)$.
For a regular value $p \in M_2$, we define the degree of the map $f$ by
\begin{equation}
\deg f(p) := \sum_{q \in f^{-1}(p)} \sgn( \det \nabla f(q)).
\end{equation}
Then $\deg f(p)$ does not depend on the choice of a regular value
$p \in M_2$. 
\end{lem}
The simplest example of this formula is $M_1=M_2=S^1$.
In this case, there is no boundary and $f$ is just a smooth map
from $S^1$ to itself. The degree of a map $f$ is the winding number.

\subsection{Proof of Theorem \ref{thm:indexsum}}
In this subsection, we prove Theorem \ref{thm:indexsum}.
For the choice of the sufficient statistics and notations, see the proof of Lemma \ref{lem:gradBethe}

Define a map $\Phi:L \rightarrow \mathbb{R}^{|V|+|E|}$ by
\begin{align}
\Phi(\bseta)_{i}& = 
(1-d_i) \frac{1}{2} \sum_{x_i = \pm 1} x_i \log b_i(x_i)
+ \frac{1}{4}\sum_{k \in N_i}\sum_{x_i,x_k=\pm 1}x_i \log
 b_{ik}(x_i,x_k), \label{Phii} \\
\Phi(\bseta)_{ij}& = 
\frac{1}{4} \sum_{x_i,x_j=\pm 1}x_i x_j \log  b_{ij}(x_i,x_j), \label{Phiij}
\end{align}
where $b_{ij}(x_i,x_j)$ and $b_i(x_i)$ 
are given by $\bseta=\{m_i,\chi_{ij}\} \in L$.
Therefore, 
we have $\nabla F = \Phi - \binom{\boldsymbol{h}}{\boldsymbol{J}}$
and
$ \nabla \Phi =  \nabla^{2} F$.
Then following claim holds. 
\begin{claim}
Under the assumption of Theorem \ref{thm:indexsum},
the sets 
$\Phi^{-1}(\binom{\boldsymbol{h}}{\boldsymbol{J}}),\Phi^{-1}(0) \subset L$ 
are finite and 
\begin{equation}
\sum_{\bseta \in \Phi^{-1}
({h \atop J} ) 
}
\sgn( \det \nabla \Phi(\bseta))
=
\sum_{\bseta \in \Phi^{-1}
(0) 
}
\sgn( \det \nabla \Phi(\bseta)), \label{eq1thmindexsum}
\end{equation}
holds.
\end{claim}

Before the proof of this claim, 
we prove Theorem \ref{thm:indexsum} under the claim.

\begin{proof}
[Proof of Theorem \ref{thm:indexsum}]
From Eq.~(\ref{Phii}) and (\ref{Phiij}),
it is easy to see that 
$\Phi(\bseta)=0 \Leftrightarrow \bseta=\{m_i=0,\chi_{ij}=0\}$.
Indeed, $\Phi (\bseta)_{ij}=0$ is equivalent to 
\begin{equation*}
 (1+m_i+m_j+\chi_{ij}) (1-m_i-m_j+\chi_{ij}) = (1-m_i+m_j-\chi_{ij}) (1+m_i-m_j-\chi_{ij})
\end{equation*}
and thus $\chi_{ij}=m_i m_j$.
Plugging into this relation into Eq.~(\ref{Phii}), one observes that $m_i=\chi_{ij}=0$.
Moreover, at this point $\bseta=\{m_i=0,\chi_{ij}=0\}$,
$\nabla \Phi =  \nabla^{2} F$ is a positive definite matrix because of Lemma \ref{lem:zerocorr}.
Therefore the RHS of Eq.~(\ref{eq1thmindexsum}) is
equal to one.
The LHS of Eq.~(\ref{eq1thmindexsum})
is equal to the LHS of Eq.~(\ref{eq:indexsum}),
because 
$\bseta \in \Phi^{-1}(\binom{\boldsymbol{h}}{\boldsymbol{J}} )  \Leftrightarrow \nabla F(\bseta)=0$.
Then the assertion of Theorem \ref{thm:indexsum} is proved.
\end{proof}

\begin{proof}[Proof of the claim]
First, we prove that 
$\Phi^{-1}( \binom{\boldsymbol{h}}{\boldsymbol{J}}  )=(\nabla F)^{-1}(0)$ is a finite set.
If not, we can choose a sequence $\{\bseta_n\}$ of distinct points 
from this set.
Let $\overline{L}$
be the closure of $L$.
Since $\overline{L}$ is compact, we can choose a subsequence that
 converges to some point $\bseta_{*} \in \overline{L}$. 
From Lemma \ref{lem:gradBethe}, $\bseta_{*} \in L$ and $\nabla F(q_{*})=0$ hold.
By the assumption in Theorem \ref{thm:indexsum},
we have $\det\nabla^{2} F(\bseta_{*})\neq 0$.
This implies that $\nabla F(\bseta)\neq 0$ in some neighborhood of $\bseta_{*}$.
This is a contradiction because $\bseta_n \rightarrow \bseta_{*}$.

Secondly, we prove the equality (\ref{eq1thmindexsum}) using Lemma \ref{lem:degreeofmap}.
Define a sequence of compact convex sets
$C_n:=\{\bseta \in L|\sum_{ij \in E}\sum_{x_i,x_j}-\log b_{ij} \leq n \}$,
which are smooth manifold with boundary and increasingly converge to $L$.
Since $\Phi^{-1}(0)$ and $\Phi^{-1} (\binom{\boldsymbol{h}}{\boldsymbol{J}})$
are finite,
they are included in $C_n$ for sufficiently large $n$.
Take $K>0$ and $\epsilon >0$ to satisfy 
$K -\epsilon > \lVert \binom{\boldsymbol{h}}{\boldsymbol{J}}  \rVert$.
From Lemma \ref{lem:gradBethe},
we see that $\Phi(\partial C_{n}) \cap B_0(K)=\phi$
for sufficiently large $n$.
Let $n_o$ be such a large number.
Let $\Pi_{\epsilon}:\mathbb{R}^{|V|+|E|} \rightarrow  B_0(K)$ be a smooth
 map that is identity on $B_0(K-\epsilon)$, 
monotonically increasing on $\lVert x \rVert$,
and
 $\Pi_{\epsilon}(x)=\frac{K}{\lVert x \rVert}x$
for $\lVert x \rVert \geq K$.
Then we obtain a composition map 
\begin{equation}
 \tilde{\Phi }:=\Pi_{\epsilon} \circ \Phi :C_{n_0} \rightarrow  B_0(K)
\end{equation}
that satisfy $\tilde{\Phi }(\partial C_{n_0}) \subset \partial B_0(K)$.
By definition, we have
$\Phi^{-1}(0)=\tilde{\Phi}^{-1}(0)$ and 
$\Phi^{-1} \binom{\boldsymbol{h}}{\boldsymbol{J}}
= \tilde{\Phi}^{-1} \binom{\boldsymbol{h}}{\boldsymbol{J}}$.
Therefore, both $0$ and $\binom{\boldsymbol{h}}{\boldsymbol{J}}$ are regular values
 of $\tilde{\Phi}$.
From Lemma \ref{lem:degreeofmap}, we have 
\begin{equation*}
\sum_{\bseta \in \tilde{\Phi}^{-1}
({h \atop J} ) 
}
\sgn( \det \nabla \tilde{\Phi}(q))
=
\sum_{\bseta \in \tilde{\Phi}^{-1}
(0) 
}
\sgn( \det \nabla \tilde{\Phi}(\bseta)).
\end{equation*}
Then, the assertion of the claim is proved.
\end{proof}

\section{Uniqueness of LBP fixed point}\label{sec:uniqueness}
This section gives a short derivation of a uniqueness condition of LBP on general graphs,
exploiting the index sum formula and the Bethe-zeta formula.
This condition is a reproduction of the Mooij's spectral condition,
though the stronger convergence property is proved under the same condition in \cite{MKsufficient}.
For binary pairwise case, numerical experiments in \cite{MKsufficient} suggests that
Mooij's condition often superior to conditions in \cite{Huniquness,TJgibbsmeasure,IFW}.

To assure the uniqueness in advance of running LBP,
we need a priori information on the LBP fixed points.
The following lemma gives such information on the correlation coefficients of the beliefs.
\begin{lem}
\label{betabound}
Let $\beta_{ij}$ be the correlation coefficient of a fixed point belief $b_{ij}$.
Then 
\begin{equation}
|\beta_{ij}| \leq \tanh(|J_{ij}|)  \quad \text{ and } \quad  \sgn(\beta_{ij})=\sgn(J_{ij}).
\end{equation}
\end{lem}
\begin{proof}
Since the belief is given by Eq.~(\ref{eq:defbelief1}),
we see that
\begin{equation}
 b_{ij}(x_i,x_j) \propto \exp (J_{ij}x_i x_j + \theta_i x_i+ \theta_j x_j)
\end{equation}
for some $\theta_i$ and $\theta_j$.
After a straightforward computation,
one observes that
\begin{equation}
\beta_{ij}= \frac{\sinh(2 J_{ij})}
{\sqrt{\cosh(2 \theta_i)+\cosh(2 J_{ij})}
\sqrt{\cosh(2 \theta_j)+\cosh(2 J_{ij})}}.
\end{equation}
The bound is attained when $\theta_{i}=0$ and $\theta_{j}=0$.
\end{proof}

\begin{cor}
[\cite{MKsufficient}]
\label{cor:uniqueradius}
Let $\bs{t}:=\{\tanh(|J_{[e]}|)\}_{e \in \vec{E}}$ be a set of directed edge weights of $G=(V,E)$.
If $\rho(\mathcal{M}(\bs{t})) < 1$,
then the fixed point of LBP is unique.
\end{cor}
\begin{proof}
Let $\bsbeta=\{\beta_{[e]}\}_{e \in \vec{E}}$.
Since $|\beta_{ij}| \leq \tanh(|J_{ij}|)$,
we have
$\rho(\mathcal{M}(\bsbeta)) \leq \rho(\mathcal{M}(\bs{t})) <1$
(Theorem \ref{app:thm:specrcomparison} in Appendix A). 
From Lemma \ref{lem:specutoc}, we have 
$\det(I - \mathcal{M}(\bsbeta))=\det(I -\matmu)$.
Therefore, from the \Bzf,
$\det (\Hesse F)$ is positive for any fixed point of LBP.
The index sum formula implies that the fixed point is unique.
\end{proof}

\section{Uniqueness result for graphs with nullity two}\label{sec:uniqueness2loop}
In this section we focus on graphs with nullity two and
show the uniqueness of the LBP fixed point for unattractive interactions.
In the proof of the above corollary,
we only used the bounds on the moduli of correlation coefficients.
In the following case of Corollary \ref{cor2loop},
we can utilize the information of signs.

To state the corollary, we need a terminology.
Two interactions $\{J_{ij},h_i\}$ and $\{J'_{ij},h'_i\}$ are
said to be {\it equivalent} if
there exists $(s_i) \in \{\pm 1\}^{V}$ such that $J'_{ij}=J_{ij}s_i s_j$
and $h'_i=h_i s_i$.
Since an equivalent model is obtained by a gauge transformation
$x_i \rightarrow x_i s_i$, the uniqueness property of LBP for equivalent models is unchanged.
Recall that a given model $\{J_{ij},h_i\}$ is {\it attractive} if $J_{ij} \geq 0$ for all $ij \in E$.
\begin{cor} 
\label{cor2loop}
Let $G$ be a connected graph with nullity two,
and assume that the interaction is not equivalent to attractive interactions,
then the LBP fixed point is unique.
\end{cor}
Interactions that are not equivalent to attractive interactions are
sometimes referred to as {\it frustrated}.
For attractive interactions, there are possibly multiple LBP fixed points \cite{MKproperty}.

For graphs with cycles, all the existing a priori conditions of the uniqueness 
upper bound the strength of interactions essentially 
and does not depend on the signs of $J_{ij}$ \cite{MKsufficient,Huniquness,TJgibbsmeasure,IFW}. 
In contrast,
Corollary \ref{cor2loop} applies to arbitrary strength of interactions if
the graph has nullity two and the signs of $J_{ij}$ are frustrated.

It is known that the LBP fixed point is unique if the graph has nullity less than two.
If the graph is a tree, the fixed point is obviously unique.
For graphs with nullity one, the uniqueness is deduced from the uniqueness of 
the Perron-Frobenius eigenvector, which correspond to the fixed point message \cite{W1loop}.
The result is easily understood form the variational point of view; 
in these cases, the \Bfe function are convex from Theorem \ref{thm:convexcondition}.
Compared to these cases, the uniqueness problem for graphs with nullity two is much more difficult.
If we write down the fixed point equation of the messages, 
it is a complicated non-linear equation. 
In such approaches, it is hard to directly figure out that the fixed point is unique.

\subsection{Example}
\label{sec:uniqunessexample}
In this subsection, we show an example to illustrate how to apply 
Theorem \ref{thm:BZ},
Theorem \ref{thm:indexsum} and Lemma \ref{betabound} to prove the uniqueness.
The complete proof of Corollary \ref{cor2loop} is given in the next subsection.

\begin{figure}
\begin{minipage}{.5\linewidth}
\begin{center}
\includegraphics[scale=0.4]{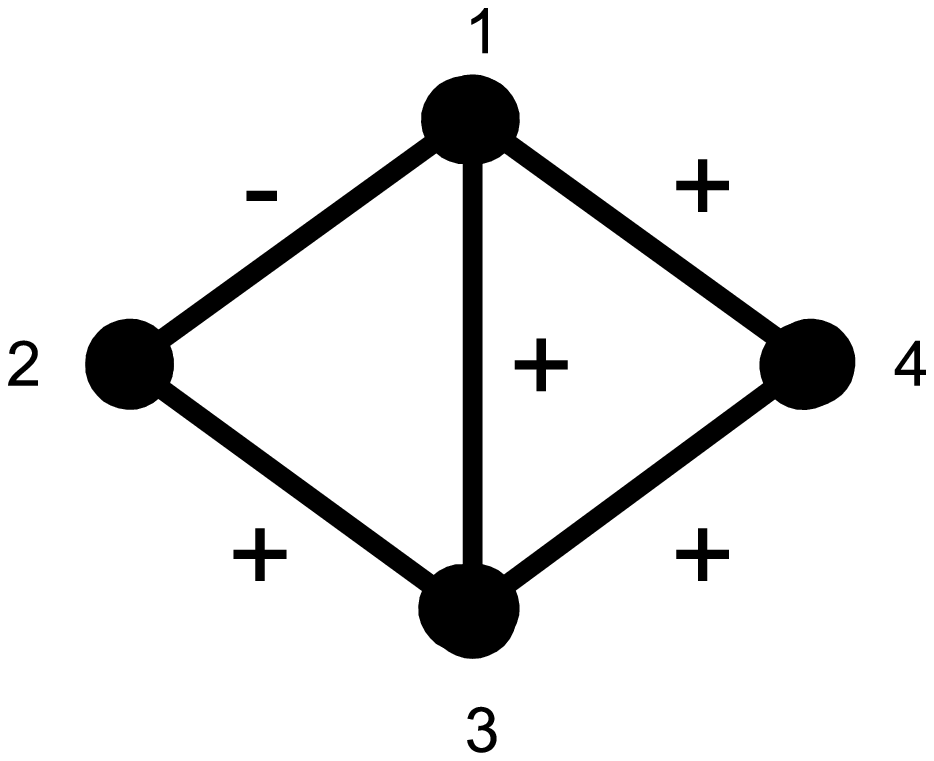} 
\caption{The graph $G$. \label{figureExampleGraph1}}
\end{center}
\end{minipage}
\begin{minipage}{.5\linewidth}
\begin{center}
\includegraphics[scale=0.3]{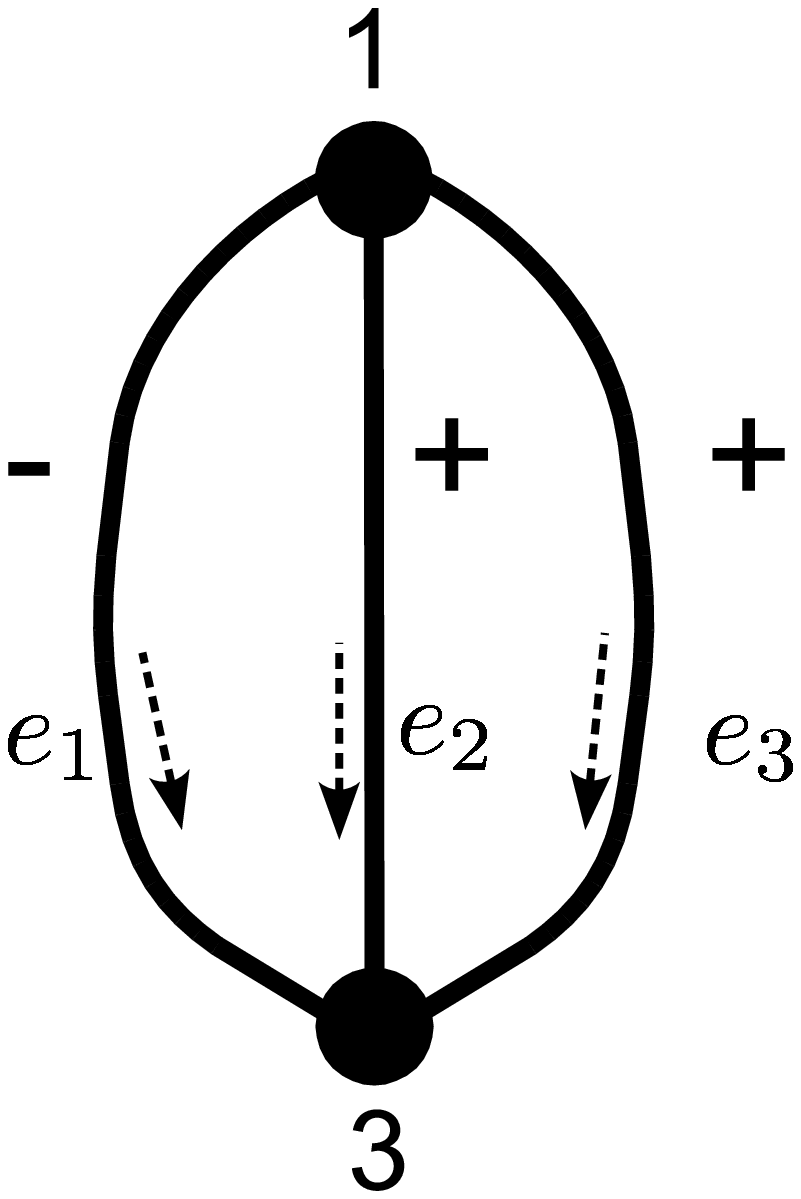}
\vspace{-3mm}
\caption{The graph $\cg{G}$. \label{figureExampleGraph2}}
\end{center}
\end{minipage}
\end{figure}

Let $V:=\{1,2,3,4\}$ and $E:=\{12,13,14,23,34\}$.
The interactions are given by
arbitrary $\{h_i\}$ and
$\{-J_{12},J_{13},J_{14},J_{23},J_{34}\}$
with $J_{ij} \geq 0$. See Figure \ref{figureExampleGraph1}.
The $+$ and $-$ signs represent that of two body interactions.
For the uniqueness of LBP fixed point, it is enough to check that
\begin{equation}
 \det(I -  \mathcal{M}(\bsbeta)) > 0
\end{equation}
for arbitrary
$0 \leq \beta_{13},\beta_{23},\beta_{14},\beta_{34} < 1$
and $-1 < \beta_{12} \leq 0$
because of Theorem \ref{thm:BZ} and Theorem \ref{thm:indexsum}.
The graph $\cg{G}$ in Figure \ref{figureExampleGraph2}
is obtained by erasing vertices $2$ and $4$ in $G$.
To compute $\det(I -  \mathcal{M}(\bsbeta))$,
it is enough to consider the smaller graph $\cg{G}$.
In fact,
\begin{align}
\det(I -  \mathcal{M}(\bsbeta))
&=
\zeta_{G}(\boldsymbol{\beta})^{-1} \nonumber \\
&=
\prod_{\mathfrak{p} \in \mathfrak{P}_G }
(1-g(\mathfrak{p}))              \label{GtoG'eq2}  \\
&=
\prod_{\mathfrak{\cg{p}} \in \mathfrak{P}_{\cg{G}} }
(1-g(\mathfrak{\cg{p}}))             \label{GtoG'eq3}   \\
&=
\zeta_{\cg{G}}(\boldsymbol{\cg{\beta}})^{-1} 
=
\det(I -  \cg{\mathcal{M}}(\cg{\bsbeta}), \nonumber
\end{align}
where $\cg{\beta}_{e_1}:=\beta_{12}\beta_{23}$,
$\cg{\beta}_{e_2}:=\beta_{13}$, 
$\cg{\beta}_{e_3}:=\beta_{14}\beta_{34}$ and
$\cg{\beta}_{e_i}=\cg{\beta}_{\bar{e_i}}$
($\bar{e_i}$ is the opposite directed edge to $e_i$).
The equality between Eq.~(\ref{GtoG'eq2}) and (\ref{GtoG'eq3})
is obtained by the one-to-one correspondence between
the prime cycles of $G$ and $\cg{G}$.
By definition,
we have 
\begin{equation*}
 \cg{\mathcal{M}}(\cg{\bsbeta})
=
\begin{small}
\left[
\begin{array}{cccccc}
0 & 0 & 0 & 0   & \cg{\beta}_{e_1} & \cg{\beta}_{e_1}   \\
0       & 0 &0 & \cg{\beta}_{e_2}  & 0 & \cg{\beta}_{e_2} \\
0 & 0  & 0 &  \cg{\beta}_{e_3}  & \cg{\beta}_{e_3} & 0         \\ 
0 & \cg{\beta}_{e_1}& \cg{\beta}_{e_1} &0 & 0 & 0  \\
\cg{\beta}_{e_2} & 0& \cg{\beta}_{e_2} &0 & 0 & 0  \\
\cg{\beta}_{e_3} & \cg{\beta}_{e_3}& 0 & 0 & 0 & 0 
\end{array}
\right],
\end{small}
\end{equation*}
where
the rows and columns are indexed by $e_1,e_2,e_3,\bar{e_1},\bar{e_2}$
and $\bar{e_3}$.
Then the determinant is
\begin{align*}
\det(I -   \cg{\mathcal{M}}(\cg{\bsbeta})  )
&=
\det \left[
I -
\left(
\begin{array}{ccc}
 0   & \cg{\beta}_{e_1} & \cg{\beta}_{e_1}    \\
 \cg{\beta}_{e_2}  & 0 & \cg{\beta}_{e_2} \\
  \cg{\beta}_{e_3}  & \cg{\beta}_{e_3} & 0         
\end{array}
\right)
\right]
\det \left[
I +
\left(
\begin{array}{ccc}
 0   & \cg{\beta}_{e_1} & \cg{\beta}_{e_1}    \\
 \cg{\beta}_{e_2}  & 0 & \cg{\beta}_{e_2} \\
  \cg{\beta}_{e_3}  & \cg{\beta}_{e_3} & 0         
\end{array}
\right)
\right] \nonumber \\
&=
(1 - \cg{\beta}_{e_1} \cg{\beta}_{e_2} - \cg{\beta}_{e_1} \cg{\beta}_{e_3} -
 \cg{\beta}_{e_2} \cg{\beta}_{e_3} 
- 2 \cg{\beta}_{e_1} \cg{\beta}_{e_2} \cg{\beta}_{e_3}) \nonumber \\
& \quad \qquad
(1 - \cg{\beta}_{e_1} \cg{\beta}_{e_2} - \cg{\beta}_{e_1} \cg{\beta}_{e_3} - 
\cg{\beta}_{e_2} \cg{\beta}_{e_3}
+ 2 \cg{\beta}_{e_1} \cg{\beta}_{e_2} \cg{\beta}_{e_3}).
\end{align*}
Since  $-1 < \cg{\beta}_{e_1} \leq 0$ and
$0 \leq \cg{\beta}_{e_2},\cg{\beta}_{e_3} < 1$,
we conclude that this is positive.

\subsection{Proof of Corollary \ref{cor2loop}}
In this subsection, we prove Corollary \ref{cor2loop} by classifying the graphs with nullity two.
There are two operations on graphs
that do not change the set of prime cycles. 
The first one is adding or erasing a vertex of degree two on any edge.
The second one is adding or removing an edge with a vertex of degree
one.
With these two operations,
all graphs with nullity two 
are reduced to three types of graphs. 
The first type is in Figure \ref{figureExampleGraph2}.
The other two types are in Figure \ref{figureExampleGraph3s}.

\begin{figure}
\begin{center}
\includegraphics[scale=0.4]{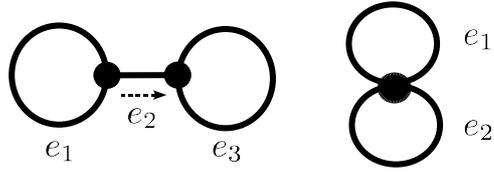} 
\caption{Two other types of graphs. \label{figureExampleGraph3s}}
\end{center}
\end{figure}
\begin{figure}
\includegraphics[scale=0.4]{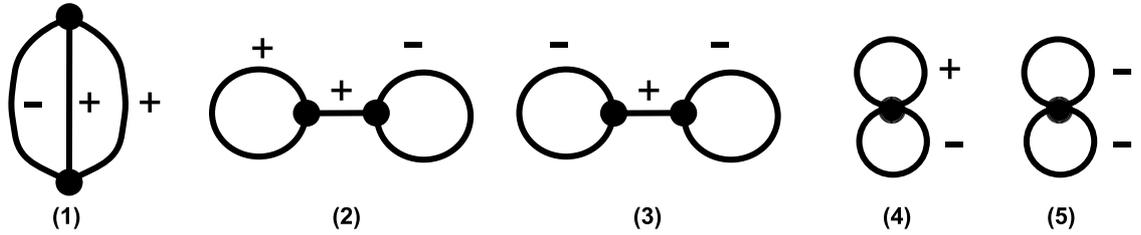} 
\caption{List of interaction types. \label{figureExampleGraphAllTypes}}
\end{figure}

Up to equivalence of interactions,
all types of signs of two body 
interactions are listed in Figure 
\ref{figureExampleGraphAllTypes} except for the attractive case.
We will check the uniqueness for each case in order.
As discussed in the previous example, all we have to do is to prove
\begin{equation}
 \det(I -  \mathcal{M}(\bsbeta)) > 0
\end{equation}
for correlation coefficients $\beta$ in a certain region.
\\
\textbf{Case (1):} Proved in Subsection \ref{sec:uniqunessexample}.\\
\textbf{Case (2):}
In this case,
\begin{equation*}
\mathcal{M}(\bsbeta)
=
\begin{small}
\left[
\begin{array}{cccccc}
 {\beta}_{e_1} &0   & 0 &0   &  {\beta}_{e_1} & 0   \\
{\beta}_{e_2}      & 0 &0 & {\beta}_{e_2}  & 0 & 0 \\
0 & {\beta}_{e_3}  &  {\beta}_{e_3} &  0  & 0 & 0         \\ 
0 & 0 & 0 &{\beta}_{e_1} & {\beta}_{e_1} & 0  \\
0 & 0& {\beta}_{e_2} &0 & 0 & {\beta}_{e_2}  \\
0 &  {\beta}_{e_3} & 0 & 0 & 0 & {\beta}_{e_3} 
\end{array}
\right]
\end{small},
\end{equation*}
where rows and columns are labeled by
$e_{1},e_{2},e_{3},\bar{e_1},\bar{e_2}$ and $\bar{e_3}$.
Then the determinant is
\begin{equation}
\det(I -  \mathcal{M}(\bsbeta))
=
(1- {\beta}_{e_1})
(1- {\beta}_{e_3})
(1- {\beta}_{e_1}-{\beta}_{e_3}+ {\beta}_{e_1}{\beta}_{e_3} 
-4 {\beta}_{e_1}{\beta}_{e_2}^2 {\beta}_{e_3} ).    \label{zetatype2}
\end{equation}
This is positive when 
$0 \leq \beta_{e_1},{\beta}_{e_2} < 1$ and
$-1 < \beta_{e_3} \leq 0$.\\
\textbf{Case (3):}
The determinant Eq.~(\ref{zetatype2})
is also positive when 
$0 \leq {\beta}_{e_2} < 1$ and
$-1 < \beta_{e_1} ,\beta_{e_3} \leq 0$.\\
\textbf{Case (4):}
In this case,
\begin{equation*}
\mathcal{M}(\bsbeta)
=
\begin{small}
\left[
\begin{array}{cccc}
 {\beta}_{e_1} &  {\beta}_{e_1}   &  0 & {\beta}_{e_1}      \\
{\beta}_{e_2}      & {\beta}_{e_2}    & {\beta}_{e_2}    & 0   \\
0 & {\beta}_{e_1}  &  {\beta}_{e_1} &  {\beta}_{e_1}          \\ 
{\beta}_{e_2}  & 0  & {\beta}_{e_2}  &{\beta}_{e_2}  
\end{array}
\right]
\end{small},
\end{equation*}
where rows and columns are labeled by
$e_{1},e_{2},\bar{e_1}$ and $\bar{e_2}$.
Then we have
\begin{equation}
\det(I -  \mathcal{M}(\bsbeta))
=
(1- {\beta}_{e_1})
(1- {\beta}_{e_2})
(1- {\beta}_{e_1}-{\beta}_{e_2}
- 3{\beta}_{e_1}{\beta}_{e_2})
.    \label{zetatype3}
\end{equation}
This is positive when 
$0 \leq {\beta}_{e_1} < 1$ and
$-1 < \beta_{e_2} \leq 0$.\\
\textbf{Case (5):}
The determinant Eq.~(\ref{zetatype3}) is positive
when
$-1 < \beta_{e_1},\beta_{e_2} \leq 0$.

\section{Discussion}
This chapter developed a new differential topological method for investigating the uniqueness of the LBP fixed point.
Our method is based on the index sum formula, which states that the sum of indices at the LBP fixed points is equal to one.
From this formula and the \Bzf,
the uniqueness is proved if $\det (I -\mathcal{M}(\bsbeta)) >0$
at all LBP fixed points.

Applying this method, we proved the uniqueness under Mooij's spectral condition \cite{MKsufficient,MKproperty}.
Our proof gives an interpretation why
the directed edge matrix $\mathcal{M}$ appears in a sufficient condition for the uniqueness.

We also showed the uniqueness for graphs with nullity two and frustrated interactions.
Though the computation of the exact partition function on such a graph is a feasible problem,
the proof of the uniqueness requires involved techniques.
Indeed our result implies that certain class of non-linear equations have the unique solutions. 
It is noteworthy that our approach is applicable to graphs that have nullity greater than two.
For example, let $B_3$ be the bouquet graph, which has the unique vertex and three self loop edges.
Applying our method for graphs homeomorphic to $B_3$, we can straightforwardly show the uniqueness
if the signs of the two body interactions are $(+,-,-)$.
It may be mathematically interesting to find a class of edge-signed graphs 
that are guaranteed to have the unique LBP fixed point with the prescribed signs of two body interactions.

For the first application, in Corollary \ref{cor:uniqueradius}, we only used a priori bounds for moduli of correlation coefficients
while, in Corollary \ref{cor2loop}, we also used information of signs but the scope was restricted to graphs with nullity two.
It would be interesting if we can utilize both information for general graphs and
show the uniqueness under stronger condition than the Mooij's spectral condition.

In this chapter, for simplicity, we focused on pairwise binary models.
However, our method can be extended to multinomial models;
the index sum formula is generalized to multinomial cases in a straightforward way.
Combining with the Bethe-zeta formula, we can show the uniqueness of the LBP fixed points under some conditions
in analogous manners.
It would be interesting to compare the uniqueness condition for multinomial models
obtained by this method and those of obtained in previous researches \cite{MKsufficient,Huniquness,TJgibbsmeasure,IFW}.

%% file: chapter6.tex

\section{Introduction}
Given a graphical model, 
the computation of the partition function and marginal distributions are important problems.
For multinomial models, evaluations of these quantities on general graphs are NP-hard \cite{Bcomputational,Ccomputational},
but efficient exact algorithm is known for tree-structured graphs.
Extending this efficient algorithm,
\lbp, or the Bethe approximation, provides an empirically good approximations for these problems on general graphs.
Theoretical understanding of these approximation errors is important issue
because it would provide practical guides for further improvement of the method.
Since the method is exact for trees, the error should be related to
cycles of graphs.

In this line of researches, Ikeda et al. \cite{ITAdecoder,ITAinfo} have developed perturbative analysis of marginals 
based on information-geometric methods.
In their analysis, cycles yield non-zero m-embedded curvature of a submanifold called $E^{*}$ and
this curvature produces the discrepancy between the true and the approximate marginals.
For pairwise models, another approaches for corrections of the Bethe approximation is considered by 
Montanari and Rizzo \cite{MRloop}.
In their method, correlations of neighbors, which are neglected in the Bethe approximation, are counted.
Parisi and Slanina \cite{PSloop} also derive corrections based on field-theoretic formulation of the Bethe approximation.
In this chapter, different from the aforementioned approaches,
we investigate an expansion called Loop Series (LS) introduced by Chertkov and Chernyak \cite{CCloopPRE,CCloop}. 
In the loop series approach, the partition function and marginal distributions 
of a {\it binary} model are expanded to sums, where the first terms are exactly the LBP solutions.
Therefore, the analysis of the remaining terms is equivalent to the quality analysis of the LBP approximation.
The most remarkable feature of the expansion, which is not achieved by the other approaches,
is that the number of terms in the sum is finite.
In the expansion of the partition function,
all terms are labeled by subgraphs of the factor graph called {\it generalized loops}, or {\it sub-coregraphs},
allowing us to observe the effects of cycles in the factor graph.
The contribution of each term is the product of local contributions along the subgraph
and is easily calculated by the LBP output.

The number of terms in the \ls is exponential with respect to the nullity of the graph,
so the direct summation is not feasible. 
(The number of sub-coregraphs is discussed in the next chapter.)
However, it provides a theoretical background to understand the approximation errors
and ways of correcting the bare LBP approximation.

The remainder of this chapter is organized as follows.
Section \ref{sec:derivationLS} derives the LS of the partition function in our notation.
In an analogous manner, we also expand the marginal distributions.
Section \ref{sec:applicationsLS} presents applications of the LS.
In section \ref{sec:perfectmatching}, we discusses the LS in a special case: the perfect matching problem.

\section{Derivation of loop series}
\label{sec:derivationLS}
In this section, we introduce the LS initiated by Chertkov and Chernyak \cite{CCloop,CCloopPRE}.
The scope of the method is the binary multinomial models.
Though our notations of the LS is different from \cite{CCloop,CCloopPRE}, 
it is essentially equivalent to theirs.
Our representation of the LS is motivated by the graph polynomial treatment of the LS in the next chapter.

Beforehand, we review the setting.
Assume that the given graphical model is
\begin{equation}
 p( x )= \frac{1}{Z} \prod_{\alpha \in F} \Psi_{\alpha}(x_{\alpha}),
\end{equation}
where $H=(V,F)$ is the factor graph, $\Psi_{\alpha}$ is a non-negative function,
$x=(x_i)_{i \in V}$, and $x_i \in \{ \pm 1\}$.
We perform inferences by the LBP algorithm using the binary multinomial \ifa on $H$.
After the LBP algorithm converged, 
we obtain the beliefs $\beliefs$ and the Bethe approximation of the partition function $Z_B$,
which is the starting point of the \ls.

\subsection{Expansion of partition functions}
The aim of this subsection is to show and prove the following theorem.
We define a set of polynomials $\{f_n(x)\}_{n=0}^{\infty}$ 
inductively by the relations $f_0(x)=1,f_1(x)=0$ and $f_{n+1}(x)=x f_n(x) + f_{n-1}(x)~ (n \geq 1)$.
Therefore, $f_2(x)=1,f_3(x)=x$ and so on.
Moreover, $f_{n}(-x)=(-1)^{n}f_n(x)$ and the coefficients of $f_n(x)$ are non-negative integers.

\begin{thm}
[Loop series expansion]
\label{thm:LS}
Let $H=(V,F)$ be the factor graph of the given graphical model
and let $B_H=(V \cup F,E_{B_H})$ be the bipartite representation of $H$. 
Assume that LBP converges to a fixed point and obtains $Z_B$ and $\beliefsw$. 
Then the following expansion of the partition function holds.
\begin{equation} 
\label{eq:LSfactor}
Z=Z_B \sum_{s \subset E_{B_H}} r(s),
\qquad \quad
r(s):=
(-1)^{|s|}
\prod_{\alpha \in F}
\beta^{\alpha}_{I_{\alpha}(s)}
\prod_{i \in V} f_{d_i(s)}(\gamma_{i}),
\end{equation}
where 
$ m_i =\E{b_i}{x_i}$,
\begin{align}
   &\beta^{\alpha}_{I}:= \E{b_{\alpha}}{ \prod_{i \in I} \frac{x_i-m_i}{\sqrt{1-m_i^2}} }  \quad \text{ for }I \subset \alpha, \label{def:betamulti} \\
   &\gamma_i := \frac{2 m_i}{\sqrt{1-m_i^2}}. \label{def:gammamulti}
\end{align}
Note that, for $s \subset E_{B_H}$,
$I_{\alpha}(s) \subset V$ is the set of vertices which are connected to $\alpha$ by edges of $s$,
and 
$d_i(s)$ is the number of factors connected to $i \in V$ by the edges of $s$.
\end{thm}
In Eq.~(\ref{eq:LSfactor}), there is a summation over all subsets of $E_{B_H}$.
An edge set $s \subset E_{B_H}$ is identified with the spanning subgraph $(V \cup F,s)$ of $B_H$.
Since $f_1(x)=0$ and $\beta^{\alpha}_{\{i\}}=0$,
a subgraph $s$ makes a contribution to the summation only if $s$ 
has neither vertices nor factors of degree one. 
Therefore, the summation is over all
coregraphs of the forms $(V \cup F,s)$;
we call them {\it sub-coregraphs}.
In relevant papers, such subgraphs are called generalized loops
\cite{CCloopPRE,CCloop} or closed subgraphs \cite{Nseries,Nsubgraph}.
\begin{figure}
\begin{minipage}{.4\linewidth}
\begin{center}
\includegraphics[scale=0.4]{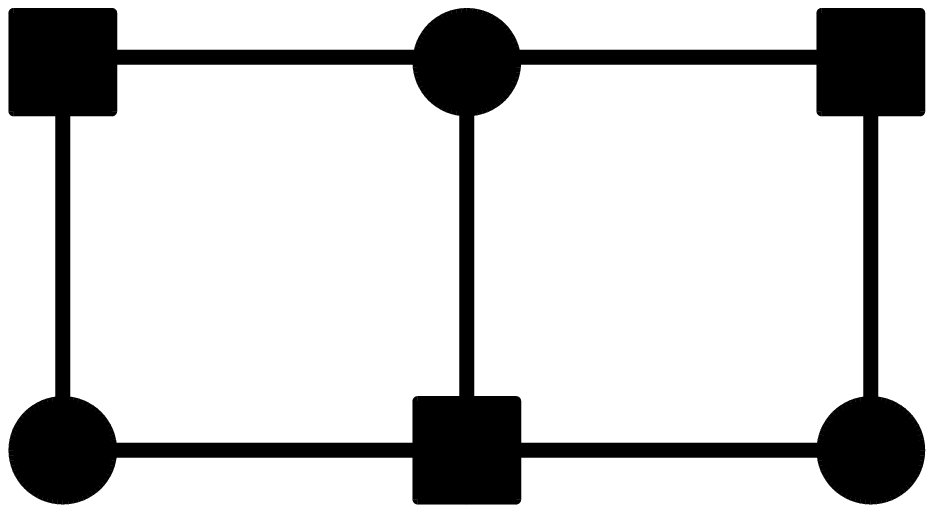} 
\caption{A hypergraph. \label{figureExampleGraph1}}
\end{center}
\end{minipage}
\begin{minipage}{.6\linewidth}
\begin{center}
\includegraphics[scale=0.23]{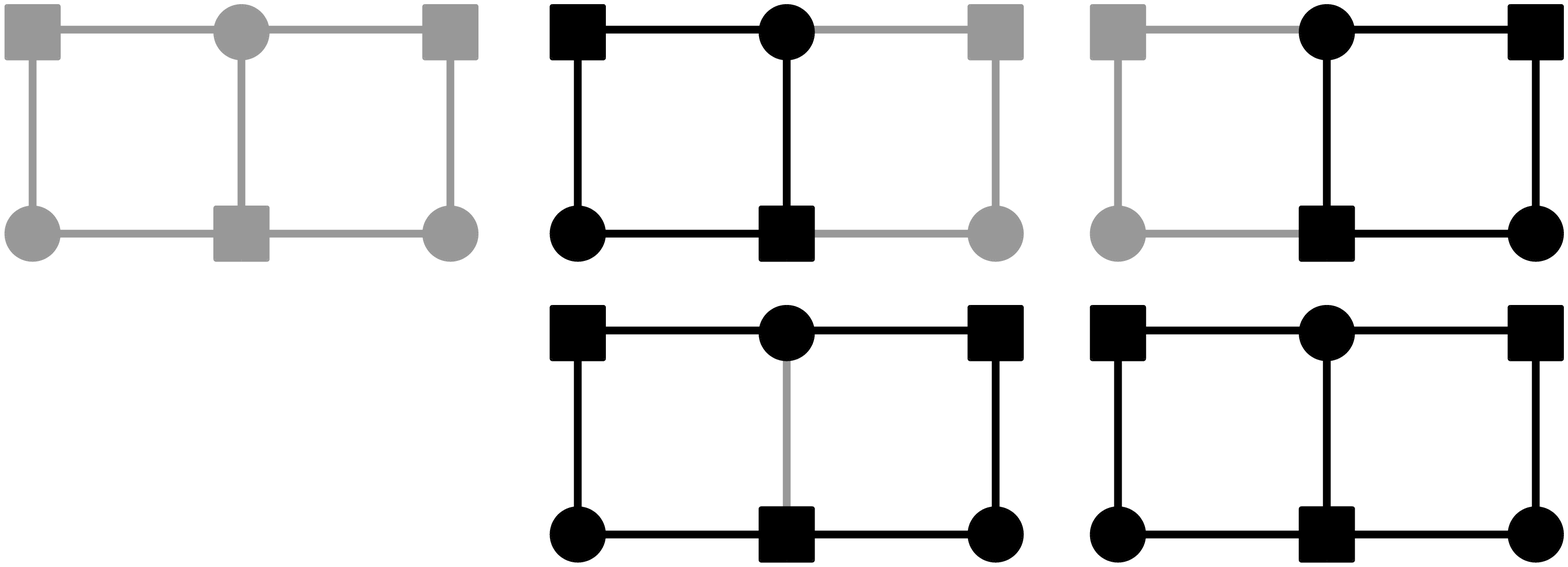}
\vspace{2.2mm}
\caption{The list of sub-coregraphs. \label{figureExampleGraph2}}
\end{center}
\end{minipage}
\end{figure}
Figure~\ref{figureExampleGraph1} and \ref{figureExampleGraph2} give an example
of a hypergraph and its sub-coregraphs.
In this example, there are five sub-coregraphs.
The number of sub-coregraphs will be discussed in Section \ref{sec:numsubcoregraph}.

Each summand $r(s)$ of the expansion is easily calculated by the resulting beliefs.
Actually, $\beta^{\alpha}_{I}$ is the multi-correlation coefficient of variables in $I$
and $\gamma_i$ is the scaled bias of the variable $x_i$.
Both of them are efficiently calculated and $r(s)$ is a product of them.

The contribution of the empty set is $r(\emptyset)=1$
because $f_{0}(x)=1$ and $\beta^{\alpha}_{\emptyset}=1$.
Therefore, the ``first'' term of the \ls expansion Eq.~(\ref{eq:LSfactor}) is
$Z_B$.
In this sense, the LS is an expansion from the Bethe approximation.

Note that, for pairwise case,
there is an understanding of the loop series
from a view point of message passing schemes \cite{WFloop}.
The correlation coefficients $\beta_{ij}$ are 
the second eigenvalues of the message transfer matrices.

The proof of Theorem \ref{thm:LS} is divided into two lemmas.
The first lemma provides a compact characterization of the ratio
of the true partition function and its Bethe approximation.
The second lemma is the key identity for the \ls expansion.

\begin{lem}
\label{lem:ZZB}
 \begin{equation}
\frac{Z}{Z_B} =
\sum_{ x \in \{ \pm 1\}^{V} }
\prod_{\alpha \in F}
\frac{b_{\alpha}(x_{\alpha})}{\prod_{i \in \alpha}b_{i}(x_i)}
\prod_{i \in V}
b_i(x_i).
\end{equation}
\end{lem}
\begin{proof}
From the definition of $Z_B$ in Subsection \ref{sec:LBPcharacterizations},
we have $\log Z_{B}= \sum_{\alpha \in F} \psi_{\alpha}(\theta^{\alpha})+\sum_{i \in V}(1-d_i) \psi_{i}(\theta_i)$.
From Eq.~(\ref{eq:beliefbyepara1},\ref{eq:beliefbyepara1}) and Assumption \ref{asm:modelindludes},
the condition Eq.~(\ref{eq:productcondition}) comes to
\begin{equation}
\prodf \Psi_{\alpha}= \prodf  b_{\alpha} 
\prodv b_i^{1-d_i} \times Z_B. \label{eq:ZBfactorout}
\end{equation}
Taking the sum over $(x_i)_{i \in V}$, we obtain the asserting equation.
\end{proof}

\begin{lem}
\label{lem:LSidentity}
For each factor $\alpha \in F$ and $I \subset \alpha$, we introduce
an indeterminate $\beta^{\alpha}_{I}$, 
where we use notation
$\beta^{\alpha}_{\phi}=1$ and $\beta^{\alpha}_{I}=0$ if $|I|=1$.
The following identity holds.
\begin{equation}
\begin{split}
\sum_{\{x_i\}}
\prod_{\alpha \in F}
\sum_{I \subset \alpha}
\beta^{\alpha}_{I}
(x_{i_1} \xi_{i_1}^{-x_{i_1}})
\cdots
(x_{i_k} \xi_{i_k}^{-x_{i_k}}) 
\prod_{i \in V}
\frac{\xi_{i}^{x_i}}{\xi_{i}+\xi_{i}^{-1}} \\ 
=
\sum_{s \subset E_{B_{H}}}
(-1)^{|s|}
\prod_{\alpha \in F}
\beta^{\alpha}_{I_{\alpha}(s)}
\prod_{i \in V} f_{d_i(s)}(\xi_{i}-\xi_{i}^{-1}), 
\end{split}\label{thm1mod}
\end{equation}
where $I=\{i_1,\ldots,i_k\}$. 
\end{lem}
\begin{proof}
 \begin{equation*}
\begin{split}
(L.H.S.)&=\sum_{\{x_i\}}\sum_{s \subset E_{B_H}}
\prodf  \big\{  \beta^{\alpha}_{I_{\alpha}(s)}   \prod_{j \in I_{\alpha}(s)}(x_{j} \xi_{j}^{-x_{j}})   \big\}
\prod_{i \in V} \frac{\xi_{i}^{x_i}}{\xi_{i}+\xi_{i}^{-1}} \\
&=
\sum_{s \subset E_{B_H}}
\prodf (-1)^{|I_{\alpha}(s)|} \beta^{\alpha}_{I_{\alpha}(s)}
\prod_{i \in V} \sum_{x_i=\pm 1}(-x_i  \xi_{i}^{-x_i})^{d_{i}(s)}
\frac{\xi_{i}^{x_i}}{\xi_{i}+\xi_{i}^{-1}} \\ 
&=
\sum_{s \subset E_{B_H}} (-1)^{|s|}
\prodf   \beta^{\alpha}_{I_{\alpha}(s)}
\prod_{i \in V} 
\frac{-(-\xi_{i})^{-d_i(s)+1}+\xi_{i}^{d_i(s)-1}}{\xi_{i}+\xi_{i}^{-1}}.
\end{split}
\end{equation*}
In the first equality, we took the sum $I \subset \alpha$ out of the product $\alpha \in F$.
On the other hand, by the definition of $f_n$, we have
\begin{equation}
f_n(\xi-\xi^{-1})=
\frac{\xi^{n-1}-(-\xi)^{-n+1}}{\xi_{}+\xi_{}^{-1}}. \label{fn}
\end{equation}
Then the identity is proved.
\end{proof}

\begin{proof}
[Proof of Theorem \ref{thm:LS}]
Let $\{ \beta^{\alpha}_{I}  \}_{I \subset \alpha, |I| \geq 2}$ be given by Eq.~(\ref{def:betamulti})
and let $\{\xi_{i}\}_{i \in \alpha}$ be given by Eq.~(\ref{def:gammamulti}) and $\gamma_i=\xi_i -\xi_i^{-1}$.
Then $b_{\alpha}(x_{\alpha})$ has the following form:
\begin{equation}
b_{\alpha}(x_{\alpha})=
\frac{1}{\prod_{i \in \alpha} (\xi_{i} + \xi_{i}^{-1})}
\sum_{I \subset \alpha}
\beta^{\alpha}_{I}
\prod_{i \in I} x_i
\prod_{j \in \alpha \smallsetminus I} \xi_{j}^{x_j}. \label{eq:beliefBG}
\end{equation}
Indeed, from the Eq.~(\ref{eq:beliefBG}), we can check the $2^{|\alpha|}-1$ conditions 
Eq.~(\ref{def:gammamulti},\ref{def:betamulti}), which determine $b_{\alpha}$ completely.
Using $b_i(x_i)=\frac{\xi_i^{x_i}}{\xi_i+\xi_i^{-1}}$, we see that the R.H.S. of Lemma \ref{lem:ZZB}
is equal to the L.H.S. of Lemma \ref{lem:LSidentity}.
Accordingly, the expansion formula is proved.
\end{proof}

\subsection{Expansion of marginals}
In this subsection, we expand the true marginal distribution $p_1(x_1):=\sum_{x \smallsetminus 1} p(x)$ rather than the
partition function.
For the sake of simplicity,
we write down the expansion of $p_1(+1)-p_1(-1)$
rather than $p_1(+1)$ or $p_1(-1)$.
Since the variable is binary, this is enough.
We define a set of polynomials $\{g_n(x)\}_{n=0}^{\infty}$ 
inductively by the relations $g_0(x)=x,g_1(x)=-2$ and $g_{n+1}(x)=x
g_n(x) + g_{n-1}(x)$. 
Therefore, $g_2(x)=-x$, $g_3(x)=-x^2-2$, and so on.
\begin{thm}
\label{thm:LSmarginal}
In the same situation as Theorem \ref{thm:LS}, the following expansion holds.
\begin{equation}
\frac{Z}{Z_B}
\frac{p_{1}(+1)-p_{1}(-1)}{\sqrt{b_{1}(1)b_{1}(-1)}}
\hspace{-1mm}
=
\hspace{-1mm}
\sum_{s \subset E_{B_H}}
\hspace{-1mm}
(-1)^{|s|}
\prod_{\alpha \in F}
\hspace{-0.7mm}
\beta^{\alpha}_{I_{\alpha}(s)}
\hspace{-3mm}
\prod_{i \in V \smallsetminus \{1\}}
\hspace{-2.5mm}
 f_{d_i(s)}(\gamma_{i})
\hspace{1mm}
g_{d_1(s)}
(\gamma_1).  \label{eq:LSmarginal}
\end{equation}
\end{thm}
\begin{proof}
The proof proceeds in a similar fashion to the proof of Theorem \ref{thm:LS}.
Analogous to Lemma \ref{lem:ZZB}, we have
\begin{equation}
\frac{Z}{Z_B}(p_1(+1)-p_1(-1)) =
\sum_{ x \in \{ \pm 1\}^{V} } x_1
\prod_{\alpha \in F}
\frac{b_{\alpha}(x_{\alpha})}{\prod_{i \in \alpha}b_{i}(x_i)}
\prod_{i \in V}
b_i(x_i),
\end{equation}
which is obtained by Eq.~(\ref{eq:ZBfactorout}).
On the other hand, 
using
\begin{equation}
\sum_{x_1=\pm 1}x_1 (-x_1 \xi_1^{-x_1})^{n}
\frac{\xi_{1}^{x_1}}{\xi_{1}+\xi_{1}^{-1}}
=
\frac{
g_{n}
(\xi_{1}-\xi_{1}^{-1})
}
{\xi_{1}+\xi_{1}^{-1}},
\end{equation}
we obtain a similar identity to Lemma \ref{lem:LSidentity}.
That is,
\begin{equation}
\begin{split}
(\xi_1+\xi_1^{-1})
\sum_{\{x_i\}} x_1
\prod_{\alpha \in F}
\sum_{I \subset \alpha}
\beta^{\alpha}_{I}
(x_{i_1} \xi_{i_1}^{-x_{i_1}})
\cdots
(x_{i_k} \xi_{i_k}^{-x_{i_k}}) 
\prod_{i \in V}
\frac{\xi_{i}^{x_i}}{\xi_{i}+\xi_{i}^{-1}} \\ 
=
\hspace{-1mm}
\sum_{s \subset E_{B_H}}
\hspace{-1mm}
(-1)^{|s|}
\prod_{\alpha \in F}
\hspace{-0.7mm}
\beta^{\alpha}_{I_{\alpha}(s)}
\hspace{-3mm}
\prod_{i \in V \smallsetminus \{1\}}
\hspace{-2.5mm}
 f_{d_i(s)}(\gamma_{i})
\hspace{1mm}
g_{d_1(s)}
(\gamma_1). 
\end{split}
\end{equation}
Combining these equations and $(\xi_1+\xi_1^{-1})=(b_1(+1)b_1(-1))^{-1/2}$, 
the theorem is proved. 
\end{proof}
By definition, the sum in the expansion is taken over all the subgraphs $s=(V \cup F, s)$ of $B_H$.
The subgraph $s$ contributes to the sum only if no vertices nor factors are degree one
except for the vertex $1 \in V$.

The ``first'' term $r(\emptyset)$ is equal to $g_0(\gamma_1)=\gamma_1=\frac{b_1(+1)-b_1(-1)}{\sqrt{b_1(+)b_1(-1)}}$.
Therefore, omitting the remaining terms, we obtain an approximation
\begin{equation}
 \argmax_{x_1 = \pm 1}p_1(x_1) \approx \argmax_{x_1 = \pm 1} b_1 (x_1).
\end{equation}
The assignment $x_1$ that maximize the marginal probability distribution $p_1(x_1)$ is called {\it Maximum Posterior Marginal} (MPM) assignment.
The above argument suggests that the Bethe approximation, which is obtained by taking the first term,
for the MPM problem of $p_1$ is given by the MPM of $b_1$.   

MPM problems of marginal distributions are especially important in the application of error
correcting codes.
In such applications, the receiver wants to infer the sent bits rather than the probabilities.

\section{Applications of LS}
\label{sec:applicationsLS}
In this section, we discuss applications of the \ls.
The first subsection provide our application of the LS and
The next subsection reviews results of other authors.

\subsection{One-cycle graphs}
In \cite{W1loop}, it is shown that if the graph has the unique cycle and the concerning node 
$1 \in V$ is on the cycle, 
then the assignment that maximize the belief $b_1$ gives the exact MPM assignment.
Using Theorem \ref{thm:LSmarginal}, we can easily show the result as follows \cite{WFloop}.
\begin{thm}
 Let $G=(V,E)$ be a graph with a single cycle with the node $1 \in V$ on it.
(See Figure \ref{fig:DBhypergraph} for example.)
Then,
$p_1(1)-p_1(-1)$ and
$b_1(1)-b_1(-1)$ have the same sign.
\end{thm}
\begin{proof}
In the right hand side of Eq.~(\ref{eq:LSmarginal}), only two subgraphs $s$ are contribute to the sum:
the empty set and the unique cycle.
From $g_0(\gamma_1)=\gamma_1$ , $g_2(\gamma_1)=-\gamma_1$ and
 $|\beta^{\alpha}_{\{ ij \}}|\leq 1$, we see that the sum is positively proportional to
 $\gamma_1$.
\end{proof}

If $1$ is not on the unique cycle, this property does not hold.
In this case, three types of subgraphs appear in Eq.~(\ref{eq:LSmarginal}).

\subsection{Review of other applications}
\label{sec:applicationLS}
{\bf Attractive models:}
A notable feature of the Bethe approximation of the partition function is that, for certain classes of models,
it lower bounds the true partition function, i.e., $Z_{B} \leq Z $.
As shown in \cite{SWWattractive}, this fact is deduced utilizing the LS for ``attractive models'':
a subclass of binary multinomial models.
Their definition of attractive models coincides with the condition $J_{ij} \geq 0$ for pairwise binary case.
In fact, if one of the following conditions
\begin{enumerate}
 \item $\gamma_i \geq 0$ for all $i \in V$, and $(-1)^{|I|}\beta^{\alpha}_{I} \geq 0$ for all $I \subset \alpha$
 \item $\gamma_i \leq 0$ for all $i \in V$, and $\beta^{\alpha}_{I} \geq 0$ for all $I \subset \alpha$
\end{enumerate}
holds, $Z/Z_B$ is obviously greater or equal to one.

~\\
{\bf Planar graphs:}
In \cite{CCTplanar}, Chertkov et al. have shown that
the partial sum of the loop series over the one-cycle sub-coregraphs 
reduces to evaluation of the partition function of the perfect matchings on an extended, $3$-regular graph. 
Weights of the perfect matchings are easily calculated by the LBP output. 
If the graph is planar, then the extend graph is also planar by construction
and the computation of the perfect matching partition function reduces to a Pfaffian.
Thus, the partial sum of the \ls is computed by a tractable determinant. 
Moreover, they find that the entire LS is reducible to a weighted Pfaffian series,
where each Pfaffian is a partial sum of the \ls.

In \cite{GKCplanar}, G\'{o}mez et al. have proposed an approximation algorithm for partition functions on planar graphs 
based on the above result.
Experimental results are presented for planar-intractable binary models, showing significant improvements over LBP.

~\\
{\bf Independent sets:}
The LS framework is utilized to analyze the performance of the Bethe approximation
for counting independent sets.
An independent set is a set of vertices such that there is no edge between any two of the vertices.
In \cite{CCGSSindependent}, Chandrasekaran et al. established that for any graph $G=(V,E)$
with max-degree $d$ and girth larger than
$8 d \log_2 |V| $, the multiplicative error decays as $1+O(|V|^{- \gamma})$ for a certain $\gamma > 0$.

\section{Special Case: Perfect matchings}
\label{sec:perfectmatching}
In this section, we apply the LS technique to a special class of problems:
the partition function of perfect matchings.
As we show in Theorem \ref{thm:PMLS}, the loop series expansion has an interesting form in this case.
First, we introduce the partition function of perfect matchings
using language of graphical models.
Then we describe the Bethe approximation in this specific case.

Let $G=(V,E)$ be a graph with non-negative edge weights $\bsw=\{w_{ij}\}_{ij \in E}$.
A {\it matching} $\match$ of $G$ is a set of edges such that any edges do not occupy a
same vertex. 
A matching is {\it perfect} if all the vertices are occupied.
The partition function of the perfect matchings is given by
\begin{equation}
 Z(\bsw) = \sum_{ \match : \text{perfect}}\hspace{1mm} \prod_{ij \in \match} w_{ij}, \label{def:Zpmatching}
\end{equation}
where the sum is taken over all the perfect matchings.
(An extension of this class of partition functions called {\it monomer-dimer partition functions}
will appear in Subsection \ref{subsectionomegamonomerdimer})
This partition function is formulated by a graphical model
over edge binary variables $ \sigma = \{  \sigma_{ij}=0,1 \}_{ij \in E} $.
A perfect matching $\match$ is identified with $\sigma$ that satisfy $\sum_{j \in N_i} \sigma_{ij}=1$ for all $i \in V$.
Let us define
\begin{equation*}
 \Psi_i(\sigma_i)=
\begin{cases}
w_{ij_0}  \quad &\text{: if } \sum_{j \in N_i}\sigma_{ij}=1 \text{ and } \sigma_{ij_0}=1, \\
0               &\text{: otherwise},
\end{cases}
\end{equation*}
where $\sigma_{i}= \{ \sigma_{ij} \}_{ j \in N_i}$.
In fact, Eq.~(\ref{def:Zpmatching}) is the normalization constant of a graphical model:
\begin{equation}
 p(\sigma)= \frac{1}{Z(\bsw)} \prod_{i} \Psi_i (\sigma_i).
\end{equation}
This is a probability distribution over all the perfect matchings.
Note that the corresponding factor graph $H$ has the vertex set and the factor set identified with $E$ and $V$, respectively.

We apply the Bethe approximation and loop series method to this partition function.
Since the functions $\Psi_i$ have zero values,
the domain of the \Bfe function $F$ is restricted.
We choose parameters $v_{ij}= b_{ij}(1)$,
then $b_i$ is determined by $v_{ij}$:
\begin{equation*}
 b_{i}(\sigma_i) = 
\begin{cases}
 v_{ij_0} \quad &\text{ if } \sum_{j \in N_i}\sigma_{ij}=1 \text{ and } \sigma_{ij_0}=1, \\
0             &\text{ otherwise}.
\end{cases}
\end{equation*}
Therefore, the \Bfe function is significantly simplified in our case:
\begin{equation}
 F(v)= -\sum_{ ij \in E}v_{ij} \log(w_{ij})
+
\sum_{ij \in E} \left\{ 
v_{ij}\log v_{ij} -(1-v_{ij})\log(1-v_{ij})
\right\}, \label{eq:PMBFE}
\end{equation}
where the domain of this function is given by
\begin{equation}
 L':= \{ v = \{v_{ij}\}_{ij \in E}| ~v_{ij} > 0, \sum_{j \in N_i} v_{ij}=1  \quad {}^{\forall}i \in V \}. \label{eq:PMlocalpolytope}
\end{equation}

To analyze the stationary points of the \Bfe function, which gives the Bethe approximation,
it is useful to introduce the following Lagrangian
\begin{equation}
 \mathcal{L}(v,\mu)= F(v)-
 \sum_{i \in V} (\mu_i+1)\left(\sum_{j \in N_i} v_{ij}-1 \right). \label{eq:Lagr}
\end{equation}
Looking for a stationary point of Eq.~(\ref{eq:Lagr}) over the $v$ variables, one arrives at the
following set of quadratic equations for each variables $v_{ij}$
\begin{equation}
 v_{ij}(1-v_{ij})=w_{ij}\exp\left(\mu_i+\mu_j\right) \quad \text{ for all } ij \in E. \label{eq:PMBP1}
\end{equation}
We call a point $v \in L'$ that satisfy this equation {\it LBP solution}.
At an LBP solution, the Bethe approximation of the partition function is given by Eq.~(\ref{eq:PMBFE}) 
and $F(v)= - \log Z_B$.

\subsection{Loop Series of perfect matching}
\label{sec:LSperfectmatching}
\begin{thm}
\label{thm:PMLS}
Let $v= \{ v_{ij} \}$ be an LBP solution and $Z_B$ be the Bethe approximation.
Then the following expansion hold.
 \begin{equation}
   Z=Z_{B} \sum_{s \subset E} r(s),\quad 
   r(s)= \prod_{i\in V} (1-d_i(s))
   \prod_{ij \in s} \frac{v_{ij}}{1-v_{ij}} .\label{eq:PMLS}
  \end{equation}
\end{thm}
\begin{proof}
We transform the $0,1$ variables $\sigma_{ij}$ to $\pm 1$ variables $x_{ij}$ by
$x_{ij}= 1- 2 \sigma_{ij}$.
Then $m_{ij}=\E{b_{ij}}{x_{ij}}= 1-2v_{ij}$.
For a factor $\alpha=i$ (i.e. the factor corresponding to the vertex $i$ of the original graph)
and $I \subset \alpha$, one derives
\begin{align*}
 \beta^{i}_{I}
&=  \E{b_{i}}{ \prod_{j \in I} \frac{ v_{ij} - \sigma_{ij} }{ \sqrt{v_{ij}(1-v_{ij})} }  } \\
&= \left(
\sum_{j \in I}v_{ij}(v_{ij}-1) \prod_{l \in I \smallsetminus j}v_{il} 
+ \sum_{j \notin I}v_{ij}  \prod_{l \in I}v_{il} 
\right) 
\prod_{j \in I}\frac{ 1 }{ \sqrt{v_{ij}(1-v_{ij})} }  \\
&=(1-|I|)\prod_{j \in I}\sqrt{ \frac{v_{ij}}{1-v_{ij}} }.
\end{align*}
Then the expansion Eq.~(\ref{eq:PMLS}) is deduced easily.
\end{proof}

In \cite{Nseries}, Nagle derived an expansion of monomer-dimer partition function,
which reduces to an expansion of the perfect matching partition function
if the monomer weights are zero.
Our expansion Eq.~(\ref{eq:PMLS}) extends his reduced expansion,
where only regular graphs and uniform edge weights cases are discussed.
Note that a similar $(1-d_i(s))$ type loop series also appear in the definition of the omega polynomial,
which we will discuss in Section \ref{sec:omega}.

For given non-negative square matrix $A$ of size $N$,
the permanent of $A$ is equal to the partition function of the perfect matchings
on the complete bipartite graph of size $N$ and with edge weights $A_{ij}$.
In \cite{CKVbeyond}, Chertkov et al. discuss this permanent problem
and use the loop series Eq.~(\ref{eq:PMLS}) to improve the Bethe approximation.
It is noteworthy that, empirically, the Bethe approximation lower bounds the true permanent.

\subsection{Loop Series by \IB type determinant formula}
\label{sec:Per_BP_Per}
The aim of this subsection is to demonstrate the importance and ubiquity of the graph zeta function
in the context of LBP and the Bethe approximation.
More precisely, we give another proof of Theorem \ref{thm:PMLS} based on the \IB type determinant formula.
Analogous to the proof of Theorem \ref{thm:LS}, the proof proceeds in two steps.
First, in Lemma \ref{lem:PMPM'}, we give a compact representation of the ratio of the partition function.
Secondly, we prove an identity involving an average of determinants.

\begin{lem}
\label{lem:PMPM'} 
For an LBP solution $v=\{v_{ij}\}$, the following is true 
\begin{equation}
 \frac{Z(\bsw)}{ Z_{B}} =  Z(v.*(1-v)) \prod_{ij\in E}(1-v_{ij})^{-1}, \label{eq:lem:PMPM'} 
\end{equation}
where $A.*B$ marks the element by element multiplication of the two matrices, $A$ and $B$ of
equal sizes.
\end{lem}
\begin{proof}
From the definition of the \Bfe Eq.~(\ref{eq:PMBFE}), and conditions for LBP solutions Eqs.~(\ref{eq:PMlocalpolytope},\ref{eq:PMBP1}), one derives
\begin{equation}
 Z_{B}=
\prod_{ij \in E}(1-v_{ij}) \prod_{ij \in E}\Big( \frac{w_{ij}}{v_{ij}(1-v_{ij})} \Big)^{v_{ij}} =
\prod_{ij \in E}(1-v_{ij}) \prod_{i \in V} \exp (- \mu_i).  \label{ZBPs}
\end{equation}
On the other hand Eq.~(\ref{eq:PMBP1}) results in $Z(\bsw)= Z(v.*(1-v))  \prodv \exp(- \mu_i)$. 
Combining the two equations, we arrive at Eq.~(\ref{eq:lem:PMPM'}).
\end{proof}
Note that this lemma, representing the ratio in terms of another partition function of perfect matchings,
reminisce Lemma \ref{lem:ZZB}.

\begin{lem}
[LS as average of determinants]\label{lem:detav}
Let $\vec{E}$ be the set of directed edges of the graph $G=(V,E)$
and $x=(x_{\edij})_{(\edij) \in \vec{E}}$ be a set of random variables that satisfies $\E{}{x_{\edij}}=0$, $\E{}{x_{\edij}x_{\edji}}=1$ and
$\E{}{x_{\edij}x_{\edkl}}=0 \quad (\{k,l\} \neq \{i,j\})$. (Here and below $\E{x}{\cdots}$ stands for
the expectation over the random variables $x$.) Then the following relation 
holds;
 \begin{equation*}
  \sum_{s \subset E}  \prod_{i\in V} (1-d_i(s)) \prod_{ij \in s} \frac{v_{ij}}{1-v_{ij}} 
   =\E{x}{ \det [ I - i \mathcal{V} \mathcal{M}] },
 \end{equation*}
where $\mathcal{V}=\diag (\sqrt{v_{ij}/(1-v_{ij})} x_{\edij})$.
\end{lem}
\begin{proof}
Expanding the determinant, one derives
\begin{equation*}
 \det [ I - i \mathcal{V} \mathcal{M}]
=\sum_{\{ e_1,\ldots,e_n \} \subset \vec{E} } \det \mathcal{M}|_{ \{ e_1,\ldots,e_n \} }
(-i)^{n} \prod_{l=1}^n (\mathcal{V})_{e_l,e_l}. 
\end{equation*}
Evaluating expectation of each summand, one observe that it is nonzero
only if $(\edij)  \in \{ e_1,\ldots,e_n \}$ implies $(\edji) \in \{ e_1,\ldots,e_n \}$.  
Thus we arrive at
\begin{equation*}
 \E{x}{ \det [ I - i \mathcal{V} \mathcal{M}] }
= \sum_{s \subset E}   (-1)^{|s|}  \det \mathcal{M}|_s  \prod_{ij \in s}
\frac{v_{ij}}{1-v_{ij}} = \sum_{s \subset E} r(s),
\end{equation*}
where $\mathcal{M}|_s$ is the restriction to the set of directed edges obtained by $s$.
In the final equality, we used the formula $\det{\mathcal{M}}=(-1)^{|E|} \prod_{i \in V}(1-d_i)$,
which is proved in Theorem \ref{thm:detofM},
for the subgraph $s$. 
\end{proof}

The determinant in the expectation is nothing but the reciprocal of the graph zeta function.
As we show below, we obtain the loop series of Theorem \ref{thm:PMLS}
by applying the Ihara-Bass type determinant formula.

\begin{proof}
[Proof of Theorem \ref{thm:PMLS} by Lemma \ref{lem:PMPM'} and \ref{lem:detav}]\quad\\
We use Lemma \ref{lem:detav} in a case that the random variables $x_{ij} =x_{\edij}=x_{\edji}$
taking $\pm 1$ values with probability $1/2$.
Using the \IB type determinant formula Eq.~(\ref{eq:IBfornonhyper}), one derives
\begin{equation*}
\det [ I - i \mathcal{V} \mathcal{M}]  = 
\det \hat{A} 
\prod_{ij \in E}(1-v_{ij})^{-1},
\end{equation*}
where 
\begin{equation*}
 \hat{A}_{ij}=\sqrt{-v_{ij}(1-v_{ij})} \hspace{1mm} x_{ij}.
\end{equation*}
Therefore,
\begin{align*}
 \sum_{s \subset E} r(s) 
= \E{x}{ \det [ I - i \mathcal{V} \mathcal{M}] }
&= \E{x}{\det \hat{A}} \prod_{ij \in E}(1-v_{ij})^{-1} \\ 
&= Z(v.*(1-v)) \prod_{ij \in E}(1-v_{ij})^{-1} \\
&= \frac{Z(\bsw)}{Z_B}.
\end{align*}
\end{proof}

\section{Discussion}
This chapter developed an alternative method for deriving and expressing the LS for the partition function
of binary models.
The loop series expansion of marginal distributions is also developed.
The form of the loop series expression reflects, in some sense, the geometry of the factor graph. 
In fact, utilizing the LS, we showed that the MPM assignment is exact for 1-cycle graphs.
In this proof, restriction on the appearing subgraphs was essential.
Such graph geometric viewpoints are further discussed in the next chapter, treating the LS as a graph polynomial.

The loop series is not independent from the graph zeta functions.
Indeed, for the perfect matchings problem, the loop series is also derived from the 
Ihara-Bass type determinant formula as discussed in Subsection \ref{sec:Per_BP_Per}.
This result suggests deep connections between LBP, the \Bfe and the graph zeta function.
However, we do not know how to derive the general \ls expansion based on graph zeta techniques.
It would be interesting to find such a derivation,
elucidating the relation between the \ls and the graph zeta function.

%% file: chapter7a.tex

%
%
%
%
%
%
%
%
%
%
%
%
%
%
%

%
%
%
\section{Introduction}
This chapter treats the Loop Series (LS) as a weighted graph characteristics called theta polynomial $\Theta_{G}(\bs{\beta},\bs{\gamma})$.
Since the LS is the ratio of the partition function and its Bethe approximation,
elucidating mathematical structures of the LS are worth interest.
In this chapter, we only discuss the {\it binary pairwise} models.

Our motivation for the graph polynomial treatment of the LS
is to ``divide the problem in two parts.''
The \ls is evaluated in two steps: 
1.~the computation of $\bs{\beta}=(\beta_{ij})_{ij \in E}$ and $\bs{\gamma}=(\gamma_{i})_{i \in V}$ by an LBP solution; 
2.~the summation of all the contributions from the sub-coregraphs. 
Since it seems difficult to derive strong results on the first step, we intend to focus on the second step.
If there is an interesting property in the form of the \ls sum, or the $\Theta$-polynomial,
the property should be related to the behavior of the error of the partition function approximation.

For example, if the graph is a tree,
the $\Theta$-polynomial is equal to one because there are no sub-coregraphs in trees.
This fact implies that the Bethe approximation gives the exact values of the partition functions
on trees.
Another notable success, in this line of approach, is the proof of $Z \geq Z_B$ for attractive models 
with means biased in one direction \cite{SWWattractive}.
The result can be understood by the property of $\Theta_G$: the coefficients of $\Theta_{G}(\bsbeta,\bsgamma)$ are non-negative.
(See Subsection \ref{sec:applicationLS}.)
Though we have not been successful in deriving properties of $\Theta_{G}(\bsbeta,\bsgamma)$
that can be used to derive unproved properties of the Bethe approximation,
we show that the $\Theta$-polynomial has an interesting property called \dcr
if the vertex weights $\gamma_i$ are set to be the same.
We further analyze the bivariate graph polynomial $\theta_{G}(\beta,\gamma)$,
which is obtained as the two-variable version of $\Theta_{G}(\beta,\gamma)$,
and the univariate graph polynomial $\omega_{G}(\beta)$,
which is obtained from $\theta_{G}(\beta,\gamma)$ by specializing $\gamma=2 \sqrt{-1}$ and 
eliminating a factor $(1-\beta)^{|E|-|V|}$.
We believe that our results give hints for future investigations on the $\Theta$-polynomial.

\subsection{Basic notations and definitions}\label{sec1.2}
In the first place, we review basic notations and definitions on graphs following Subsection \ref{sec:basicsgraph}.
For clarity, we summarize them for the case of graphs, not hypergraphs.
Let $G=(V,E)$ be a finite graph, where
$V$ is the set of vertices and 
$E$ is the set of undirected edges.
In this chapter, a graph means a multigraph, in which
loops and multiple edges are allowed.
Note that, in graph theory, a {\it loop}\footnote{
The term ``loop'' in `` loopy belief propagation'' and ``\ls'' has no relation to this definition of loop.  
}
is an edge that connects a vertex to itself.
A subset $s$ of $E$ is identified with the {\it spanning subgraph}
$(V,s)$ of $G$ unless otherwise stated.

In this chapter, we use a symbol $e$ to represent an undirected edge,
though it was mainly used to represent a directed edge in previous chapters.
By the notation of $e=ij$ we mean that 
vertices $i$ and $j$ are the endpoints of $e$. 
The number of ends of edges connecting to a vertex $i$ is called 
the {\it degree} of $i$ and denoted by $d_i$.

The number of connected components of $G$ is denoted by $k(G)$. 
The {\it nullity} and the {\it rank} of $G$ are defined by
$n(G):=|E|-|V|+k(G)$ and $r(G):=|V|-k(G)$ respectively.

For a graph $G$, the {\it core} of the graph $G$ 
is given by a process of clipping
vertices of degree one step by step \cite{Stopology}. 
This graph is denoted by $\core(G)$.
A graph $G$ is called a {\it coregraph} if $G=\core(G)$. 
In other words, a graph is a coregraph if and only if the degree of each
vertex is not equal to one.

For an edge $e \in E$,
the graph
$G \backslash e$ is obtained by deleting $e$
and $G/e$ is obtained by contracting $e$. 
If $e$ is a loop, $G/e$ is the same as $G \backslash e$.
The disjoint union of graphs $G_1$ and
$G_2$ is denoted by $G_1 \cup G_2$.
The graph with a single vertex and
$n$ loops is called the {\it bouquet graph} and denoted by $B_n$.

\subsection{Graph polynomials}
Partition functions studied in statistical physics have been a source of 
many graph polynomials. 
For example, the partition functions of
the q-state Potts model and 
the bivariated random-cluster model of Fortuin and Kasteleyn
derive graph polynomials.
They are known to be equivalent to the Tutte polynomial
\cite{Bollobas}.
Another example is
the monomer-dimer partition function with uniform 
monomer and dimer weights,
which is essentially the matching polynomial \cite{HLmonomerdimer}.

The most important feature of our graph polynomials is the
deletion-contraction relation:
\begin{align}
&\theta_{G}({\beta},\gamma)=
(1-\beta) \theta_{G\backslash e}({\beta},\gamma) +
\beta  \theta_{G/e}({\beta},\gamma),   \nonumber \\
&\omega_{G}({\beta})=
\omega_{G\backslash e}({\beta}) +
\beta  \omega_{G/e}({\beta}),   \nonumber
\end{align}
where $e \in E$ is not a loop.
Furthermore, these polynomials are multiplicative:
\begin{equation}
 \theta_{G_1 \cup G_2} =  \theta_{G_1}  \theta_{G_2}  \quad \text{ and } \quad 
 \omega_{G_1 \cup G_2} =  \omega_{G_1}  \omega_{G_2}. \nonumber
\end{equation}
Graph invariants that satisfy 
the deletion-contraction relation and the multiplicative law
are studied by Tutte \cite{Tring} in the name of V-function.
Our graph polynomials $\theta_{G}$ and $\omega_{G}$ are essentially
examples of V-function.

Graph polynomials that satisfy deletion-contraction relations
arise from wide range of problems
\cite{Bollobas,EMinter1}. 
To the best of our knowledge,
all of the graph polynomials that satisfy \dcrs and appear in some applications 
are known to be equivalent to the Tutte polynomial
or obtained by its specialization.
We can list the chromatic polynomial, the flow polynomial and the reliability polynomial
for such examples.
The Tutte polynomial have a reduction formula for loops,
but our new graph polynomials do not have such reduction formulas for loops and
are essentially different from the Tutte polynomial.

\subsection{Overview of this chapter}
This chapter discusses the following topics.
First, in Section \ref{sec:MLS}, we define the weighted graph characteristic
$\Theta_{G}(\bsbeta,\bsgamma)$.
An interesting property called \dcr is shown when the vertex weights connected by the contracted edge are equal.
In Section \ref{sec:numsubcoregraph}, we derive upper and lower bounds 
on the number of sub-coregraphs, which are attained by $3$-regular graphs and bouquet graphs respectively. 
Section \ref{sec:theta} is a discussion on the $\theta$-polynomial.
We see that the $\theta$-polynomial is essentially a new interesting example of a special class of graph polynomials called V-function.
Section \ref{sec:omega} is devoted to investigations of the $\omega$-polynomial
including a study on the special value $\beta=1$.
We show that the polynomial coincides with the monomer-dimer partition function with weights parameterized by $\beta$.
Especially, it is essentially the matching polynomial if the graph is regular.

\section{Loop series as a weighted graph polynomial}
\label{sec:MLS}
\subsection{Definition}
In the first place, we introduce the expansion of the LS as a weighted graph polynomial.
We associate complex numbers with vertices and edges $\bs{\gamma}=(\gamma_{i})_{i \in V}$ and
$\bs{\beta}=(\beta_{e})_{e \in E}$ respectively,
making a graph $G$ be a {\it weighted graph}.

Recall that the set of polynomials $\{f_n(x)\}_{n=0}^{\infty}$ is defined
inductively by the relations
\begin{equation}
 f_0(x)=1, \quad
f_1(x)=0,   \quad
\text{ and } \quad
f_{n+1}(x)=x f_n(x) + f_{n-1}(x). \label{defindf}
\end{equation}
Therefore, $f_2(x)=1,f_3(x)=x$ and so on.
Note that,
these polynomials are transformations of the Chebyshev polynomials of the second
kind:
$f_{n+2}(2 \sqrt{-1} z)=(\sqrt{-1})^{n} U_{n}(z)$,
where $U_{n}(\cos \theta)= \frac{\sin ((n+1)\theta)}{\sin \theta}$.

Since we are considering graphs,
the expression in Eq.~(\ref{eq:LSfactor}) reduces to the following form.
\begin{defn}
\label{defTheta}
Let $\bs{\beta}=(\beta_{e})_{e \in E}$ 
 and $\bs{\gamma}=(\gamma_{i})_{i \in V}$
be the weights of $G$. 
We define
\begin{equation}
\Theta_{G}(\bs{\beta},\bs{\gamma}):= 
\sum_{s \subset E}
\prod_{e \in s} 
\beta_{e}
\prod_{i \in V} f_{d_i(s)}(\gamma_i),  \label{eq:defnTheta} 
\end{equation}
where $d_i(s)$ is the degree of the vertex $i$ in $s$.
\end{defn}
If all the vertex weights are equal to $\gamma$, it is denoted by $\Theta_{G}(\bs{\beta},{\gamma})$.
In addition, if all the edge weights are set to be the same, it is denoted by $\theta_{G}(\beta,\gamma)$.

In Eq.~(\ref{eq:defnTheta}), there is a summation over all subsets of $E$.
Recall that an edge set $s$ is identified with the spanning subgraph $(V,s)$.
Since $f_1(x)=0$, 
the subgraph $s$ makes a contribution to the summation only if $s$ does
not have a vertex of degree one. 
Therefore, the summation is regarded as the summation over all
coregraphs of the forms $(V,s)$;
we call them {\it sub-coregraphs}.
In relevant papers, such subgraphs are called generalized loops
\cite{CCloopPRE,CCloop} or closed subgraphs \cite{Nseries,Nsubgraph}.

It is trivial by definition that
\begin{align}
&\Theta_{G_1 \cup G_2}(\bs{\beta},\bs{\gamma})
 =\Theta_{G_1}(\bs{\beta},\bs{\gamma})\Theta_{G_2}(\bs{\beta},\bs{\gamma}), \label{Thetaproperty0} \\
&\Theta_{B_0}(\bs{\beta},\bs{\gamma})=1, \label{Thetaproperty1} \\ 
&\Theta_{G}(\bs{\beta},\bs{\gamma})=\Theta_{\core (G)}(\bs{\beta},\bs{\gamma}). \label{Thetaproperty2}
\end{align}
These properties are reminiscence of the properties of the graph zeta function:
$\zeta_{H_1 \cup H_2} =\zeta_{H_1} \zeta_{H_2}$,~$\zeta_{ \emptyset } = 1$ and $\zeta_H =\zeta_{\core(H)}$.
\subsection{Deletion-contraction relation}
Assuming a certain relation on $\bsgamma=(\gamma_i)_{i \in V}$,
we prove the most important property of the graph polynomial $\Theta$ called a 
deletion-contraction relation. 
The following formula of $f_n(x)$ is essential in the proof of the relation.
\begin{lem}
\label{lem:f}
 ${}^\forall n,m \in \mathbb{N}$,
\begin{equation}
f_{n+m-2}(x)= f_n(x)f_m(x)+f_{n-1}(x)f_{m-1}(x). \nonumber
\end{equation}
\end{lem}
\begin{proof}
Easily proved by induction using Eq.(\ref{defindf}).
\end{proof}

\begin{thm}
[Deletion-contraction relation]
\label{thm:Mcd}
Let $e=ij \in E$ be not a loop.
Assume that the weights $(\bs{\beta},\bs{\gamma})$ on $G$ satisfies 
$\gamma_{i}=\gamma_{j}$.
The weights on $G \backslash e$ and $G/e$
are naturally induced and denoted by $(\bs{\beta'},\bs{\gamma'})$
and $(\bs{\beta''},\bs{\gamma''})$ respectively. 
(On $G/e$, the weight on the new vertex, which is the fusion of $i$ and $j$, is set to be $\gamma_i$.)
Under these conditions, we have
\begin{equation}
 \Theta_{G}(\bs{\beta},\bs{\gamma})=
(1-\beta_e) \Theta_{G\backslash e}(\bs{\beta'},\bs{\gamma'}) +
\beta_{e}  \Theta_{G/e}(\bs{\beta''},\bs{\gamma''}).  \nonumber
\end{equation}
\end{thm}
\begin{proof}
Classify subgraph $s$ in the sum of Eq.~(\ref{eq:defnTheta}) whether $s$ includes $e$ or not.
A subgraph $s  \ni e=ij$ in the former case yields 
$-\beta \Theta_{G \backslash e}+\beta\Theta_{G / e}$,
where Lemma \ref{lem:f} is used with $n=d_i$ and $m=d_j$.
A subgraph $s \not \ni e$ in the latter case yields $\Theta_{G \backslash e}$.
\end{proof}

Especially, $\Theta_{G}(\bs{\beta},\gamma)$ satisfies this relation.
By successive applications of the relations, $\Theta_{G}(\bs{\beta},\gamma)$ 
can be reduced to the values at disjoint unions of bouquet graphs.
The \dcr allows another expansion of $\Theta_{G}(\bs{\beta},\gamma)$ as a sum over the subgraphs.

\begin{thm} 
\label{cortheta}
\begin{equation}
 \Theta_{G}(\bs{\beta},\gamma)=
\sum_{s \subset E}
\prod_{n=0}
\theta_{B_n}(1,\gamma)^{i_n(s)}
\prod_{e \in s}
\beta_{e}
\prod_{e \in  E \smallsetminus s}
(1-\beta_e ),  \label{corthetaeq1}
\end{equation}
where $i_n(s)$ is the number of connected components of the subgraph
$s$ with nullity $n$.
\end{thm}
\begin{proof}
In this proof,
the right hand side of Eq.~(\ref{corthetaeq1}) is 
denoted by $ \tilde{\Theta}_{G}(\bs{\beta},\gamma)$.
First, we check that ${\Theta}_{G}$ and 
$\tilde{\Theta}_{G}$ are equal at the bouquet graphs.
\begin{align}
\tilde{\Theta}_{B_n}(\bs{\beta},\gamma)
&=
\sum_{s \subset E}
\theta_{B_{|s|}}(1,\gamma)
\prod_{e \in s}
\beta_{e}
\prod_{e \in  E \smallsetminus s}
(1-\beta_e ) \nonumber \\
&=
\sum_{s \subset E}
\sum_{k =0}^{|s|}
{|s| \choose k}f_{2k}(\gamma)
\prod_{e \in s}
\beta_{e}
\sum_{t \subset E \smallsetminus s}
\prod_{e \in t}
(-\beta_e ) \nonumber \\
&=
\sum_{u \subset E}
\sum_{s \subset u}
\sum_{k =0}^{|s|}
{|s| \choose k}f_{2k}(\gamma)
(-1)^{|u|-|s|}
\prod_{e \in u}
\beta_e  \nonumber \\
&=
\sum_{u \subset E}
\sum_{l=0}^{|u|}
\sum_{k =0}^{l}
{|u| \choose l}
{l \choose k}f_{2k}(\gamma)
(-1)^{|u|-l}
\prod_{e \in u}
\beta_e. \nonumber
\end{align}
Using the equality
$\sum_{j=k}^{n} {n \choose j}{j \choose k} (-1)^{n+j}=\delta_{n,k}$,
which is obtained by comparing the coefficients of $\left( 1 - (1-x) \right)^n=x^n$,
we have
\begin{equation}
\tilde{\Theta}_{B_n}(\bs{\beta},\gamma)
=
\sum_{u \subset E} f_{2|u|}(\gamma) 
\prod_{e \in u}\beta_{e} 
=
\Theta_{B_n}(\bs{\beta},\gamma). \nonumber
\end{equation}

Secondly, we see that $ \tilde{\Theta}_{G}(\bs{\beta},\gamma)$
satisfies the deletion-contraction relation 
\begin{equation}
 \tilde{\Theta}_{G}(\bs{\beta},\gamma)=
(1-\beta_e) \tilde{\Theta}_{G\backslash e}(\bs{\beta'},\gamma) +
\beta_{e}  \tilde{\Theta}_{G/e}(\bs{\beta''},\gamma) \nonumber
\end{equation}
for all non-loop edges $e$,
because
the subsets including $e$ amount to
$\beta_e \tilde{\Theta}_{G / e}(\bs{\beta},\gamma)$ 
and
the other subsets
amount to
$(1-\beta_e) \tilde{\Theta}_{G \smallsetminus e}(\bs{\beta},\gamma)$.

Applying this form of deletion-contraction relations to both $\Theta_{G}$
and $\tilde{\Theta}_{G}$, we can reduce the values at $G$ to those of  
disjoint unions of the same bouquet graphs.
Therefore
we conclude that $\tilde{\Theta}_{G}=\Theta_{G }$.
\end{proof}

Eq.~(\ref{corthetaeq1}) resembles the famous
{\it random-cluster model} of Fortuin and Kasteleyn \cite{FKrandom}
\begin{equation}
R_{G}(\bs{\beta},\kappa)
=
\sum_{s \subset E}
\kappa^{k(s)}
\prod_{e \in s}
\beta_{e}
\prod_{e \in  E \smallsetminus s}
(1-\beta_e ),  \nonumber
\end{equation}
which is a special case of the colored Tutte polynomial in \cite{BRTutte}.
This function satisfies a deletion-contraction relation of the form
\begin{equation}
 R_{G}(\bs{\beta},\kappa) = (1-\beta_e) R_{G \backslash e}(\bs{\beta'},\kappa) 
+\beta_e R_{G/e}(\bs{\beta''},\kappa) 
\quad \text{ for all } e \in E. \nonumber
\end{equation}
Note that this relation holds also for loops in contrast to
$\Theta_{G}(\bs{\beta},\gamma)$.
This difference comes form that of the coefficients of subgraphs $s$:
$\kappa^{k(s)}$ and $\prod \theta_{B_n}(1,\gamma)^{i_n(s)}$.

\section{Number of sub-coregraphs}
\label{sec:numsubcoregraph}
An important property of the LS is that only the sub-coregraphs contribute to the sum,
not all the subgraphs.
Thus it is worth investigating how many sub-coregraphs in a graph among the subgraphs.
This section discusses the number of sub-coregraphs,
i.e., the number of terms in the LS.
A simple analysis of the theta polynomial gives bounds on the numbers.
First, we compute $\theta_{G}(1,\gamma)$.

\begin{lem} 
\label{lembeta1}
For a connected graph $G$,
\begin{align}
\theta_{G}(1,\xi-\xi^{-1})
&=
\xi^{1-n(G)}(\xi+\xi^{-1})^{n(G)-1}
+
\xi^{n(G)-1}(\xi+\xi^{-1})^{n(G)-1}.
\label{lembeta1eq1}
\end{align}
Note that the value $\theta_{G}(1,\gamma)$ is determined by the nullity $n(G)$.
\end{lem}
\begin{proof}
For the proof, we use Lemma \ref{lem:LSidentity}, which gives an alternative representation
of $\Theta_{G}$.
In this graph case, Eq.~(\ref{thm1mod}) reduces to
\begin{equation}
\Theta_{G}(\bs{\beta},(\xi_{i}-\xi_{i}^{-1})_{i \in V})= 
\sum_{x_1,\ldots,x_N=\pm 1}
\prod_{e \in E \atop e=ij}
(1+x_{i}x_{j}\beta_{e}\xi_{i}^{-x_i}\xi_{j}^{-x_j})
\prod_{i \in V}
\frac{\xi_{i}^{x_i}}{\xi_{i}+\xi_{i}^{-1}}. \label{thetaidentityeq}
\end{equation}
We set $\beta_e=1$ and $\xi_i=\xi$.
If $x_i \neq x_j$, then $1+x_{i}x_{j}\xi^{-x_i}\xi^{-x_j}=0$.
Since $G$ is connected, only the two terms of $x_1=\cdots=x_N=1$ and $x_1=\cdots=x_N=-1$ contribute to the sum.
Then the equality is proved.
\end{proof}

If $\xi=\frac{1+\sqrt{5}}{2}$, then $\xi-\xi^{-1}=1$. From Eq.~(\ref{lembeta1eq1}),
we see that
\begin{equation}
\theta_{G}(1,1)=
 \Biggl(\frac{5-\sqrt{5}}{2}\Biggr)^{n(G)-1}
\hspace{-2mm}
+
\Biggl(\frac{5+\sqrt{5}}{2}\Biggr)^{n(G)-1}. \label{numup}
\end{equation}
Setting $\xi=1$,
We also deduce from Eq.~(\ref{lembeta1eq1}) that
\begin{equation}
\theta_{G}(1,0)=2^{n(G)}. \label{numlow}
\end{equation}

\subsection{Bounds on the number of sub-coregraphs}
For a given graph $G$,
let
$\mathcal{C}(G):=\{ s; s \subset  E, (V,s) \text{ is a coregraph.} \}$
be the set of sub-coregraphs of $G$. 
In the following theorem,
the values in Eqs.~(\ref{numup}) and (\ref{numlow})
are used to bound the number of sub-coregraphs.
The upper bound is first proved in \cite{WFloop}.

\begin{thm}
\label{thmbound}
For a connected graph $G$,
\begin{equation}
2^{n(G)}
\leq
\left| \mathcal{C}(G) \right|
\leq
 \Biggl(\frac{5-\sqrt{5}}{2}\Biggr)^{n(G)-1}
\hspace{-2mm}
+
\Biggl(\frac{5+\sqrt{5}}{2}\Biggr)^{n(G)-1}. \label{thmboundeq1}
\end{equation}
The lower bound is attained if and only if
$\core(G)$ is a subdivision of a bouquet graph,
and the upper bound is attained if and only if 
$\core (G)$ is a subdivision of a 3-regular graph
or $G$ is a tree.
\end{thm}
Note that a {\it subdivision} of a graph $G$ is 
a graph that is obtained by adding vertices
of degree 2 on edges.
\begin{proof}
It is enough to consider the case that $G$ is a coregraph
and does not have vertices of degree 2,
because 
the operations of taking core and subdivision
do not change the nullity and the set of sub-coregraphs essentially.

From the definition Eq.~(\ref{defthetaeq}), we can write
\begin{equation}
\theta_{G}(1,\gamma)
=
\sum_{s \in \mathcal{C}} w(s;\gamma),  \nonumber
\end{equation}
where $w(s;\gamma)=\prod_{i \in V}f_{d_{i}(s)}(\gamma)$.
For all $s \in \mathcal{C}$,
we claim that
\begin{equation}
w(s;0) \leq 1 \leq w(s;1). \label{thmboundeq2}
\end{equation}
The left inequality of Eq.~(\ref{thmboundeq2})
is immediate from the fact that $f_n(0)=1$ if $n$ is even
and $f_n(0)=0$ if $n$ is odd.
The equality holds if and only if
all vertices have even degree in $s$. 
Since $f_n(1)>1$ for all $n>4$ and $f_2(1)=f_3(1)=1$,
we have $w(s;1) \geq 1$.
The equality holds if and only if 
$d_i(s) \leq 3$
for all $i\in V$.
Therefore  
the inequalities in Eq.~(\ref{thmboundeq1}) 
are proved.
The upper bound is attained
if and only if $G$ is a 3-regular graph or
the $B_0$.
For the equality condition of the lower bound,
it is enough to prove the following claim.
\begin{claim}
Let $G$ be a connected graph, and assume that the degree of every vertex is 
at least 3
and $d_i(s)$ is even for every $i \in V$
and $s \in \mathcal{C}$.
Then $G$ is a bouquet graph.
\end{claim}
If $G$ is not a bouquet graph, there is a non-loop edge $e=i_0j_0$.
Then $E$ and $E \smallsetminus e$ are sub-coregraphs of $G$.
Thus $d_{i_0}(E)$ or $d_{i_0}(E \smallsetminus e)=d_{i_0}(E)-1$ is odd.
This is a contradiction.
\end{proof}

\subsection{Number of sub-coregraphs in 3-regular graphs}
If the core of a graph is a subdivision of a 3-regular graph,
we obtain more information on the number of specific types of sub-coregraphs. 

We can rewrite Lemma \ref{lembeta1} as follows.
\begin{lem} 
\label{lemthetacoeff}
Let $G$ be connected and not a tree. Then we have
\begin{equation}
\theta_{G}(1,\gamma)
=
\sum_{l=0}^{n(G)-1} C_{n(G),l} \gamma^{2l},  \nonumber
\end{equation}
where $C_{n,l}:=\sum_{k=l+1}^{n}\binom{n}{k}\binom{k+l-1}{2l}$ for 
$1 \leq l \leq n-1$ and $C_{n,0}:= 2^{n}$.
\end{lem}
\begin{proof}
First we note that
for $k \geq 1$,
\begin{equation}
f_{2k}(\gamma)=\sum_{l=0}^{k-1}\binom{k+l-1}{2l}\gamma^{2l}  
\quad  \text{  and  } \quad
f_{2k+1}(\gamma)=\sum_{l=0}^{k-1}\binom{k+l}{2l+1}\gamma^{2l+1}.  \nonumber
\end{equation}
This is easily proved inductively using
Eq.~(\ref{defindf}).
Then Lemma \ref{lembeta1} derives
\begin{align*}
\theta_{G}(1,\gamma)
=
\theta_{B_{n(G)}}(1,\gamma)
&=\sum_{k=1}^{n(G)} \binom{n(G)}{k} f_{2k}(\gamma)
+f_{0}(\gamma) \\
&=
\sum_{l=0}^{n(G)-1}\sum_{k=l+1}^{n(G)} \binom{n(G)}{k} \binom{k+l-1}{2l}\gamma^{2l}
+1 \\
&=
\sum_{l=0}^{n(G)-1} C_{n(G),l} \gamma^{2l}.
\end{align*}
\end{proof}

\begin{thm}
\label{thm:numberSC3regular}
Let $G$ be a connected graph and not a tree.
If every vertex of $\core(G)$ has the degree at most $3$,
then 
\begin{equation}
C_{n(G),l}=  
|\{ s \in \mathcal{C}(G) ; s \text{ has exactly }2l \text{ vertices of degree
3.} \}|  \nonumber
\end{equation}
for $0 \leq l \leq n(G)-1$.
\end{thm}
\begin{proof}
From the assumption of this theorem, all degrees of a sub-coregraph $s$ are at most three.
Accordingly, for a sub-coregraph $s$,
$\prod_{i \in V}f_{d_{i}(s)}(\gamma)=\gamma^{2l}$ holds,
where $2l$ is the number of vertices of degree three in the subgraph $s$.
Therefore, Lemma \ref{lemthetacoeff} implies the assertion.
\end{proof}

For example, consider the hypergraph in Figure \ref{figureExampleGraph1} as a graph.
Since this graph is a subdivision of a $3$-regular graph, we can apply the Theorem \ref{thm:numberSC3regular}.
This graph has nullity $2$ and
\begin{equation*}
 C_{2,1}=1, \quad C_{2,0}=4.
\end{equation*}
From Figure \ref{figureExampleGraph2}, one observes that Theorem \ref{thm:numberSC3regular} is correct for this case.

\section{Bivariate graph polynomial $\theta_G$}
\label{sec:theta}
In this section, we discuss $\theta_G$.
For a given graph $G$, 
\begin{equation}
\theta_{G}({\beta},{\gamma}):= 
\sum_{s \subset E}
\beta^{|s|}
\prod_{i \in V} f_{d_i(s)}(\gamma)
\quad
\in \mathbb{Z} [ \beta,\gamma ],  \label{defthetaeq}
\end{equation}
where $d_{i}(s)$ is the degree of the vertex $i$ in $s$.

\subsection{Basic properties}
The following facts are immediate from the previous results.
\begin{prop}$\quad$ 
\label{propbasictheta}
\begin{itemize}
\item[{\rm (a)}]
$\theta_{G_1 \cup G_2}({\beta},{\gamma})
 =\theta_{G_1}({\beta},{\gamma})\theta_{G_2}({\beta},{\gamma}). $
\item[{\rm (b)}]
$\theta_{B_n}({\beta},\gamma)=
\sum_{k=0}^{n} {n \choose k} f_{2k}(\gamma) \beta^{k}.$ 
\item[{\rm (c)}]
$ \theta_{G}({\beta},{\gamma})=\theta_{\core(G)}({\beta},{\gamma}).$
\item[{\rm (d)}]
$ \theta_{G}(\beta,\gamma)= (1-\beta) \theta_{G\backslash e}(\beta,\gamma) + \beta  \theta_{G/e}({\beta},\gamma)$ for a non loop edge $e$.
\item[{\rm (e)}]
$\theta_{G}(\beta,\gamma)= \sum_{s \subset E} \prod_{n=0} \theta_{B_n}(1,\gamma)^{i_n(s)} \beta^{|s|} (1-\beta)^{|E|-|s|}.$
\end{itemize}
\end{prop}

\begin{example}
For a tree $T$, $\theta_{T}(\beta,\gamma)=1$.
For the cycle graph $C_n$, which has $n$ vertices and $n$ edges,
$\theta_{C_n}(\beta,\gamma)=1+\beta^{n}$. 
For the complete graph $K_{4}$,
$\theta_{K_4}(\beta,\gamma)=1+4\beta^{3}+3 \beta^{4}+6
 \beta^{5}\gamma^{2}+\beta^{6}\gamma^{4}$.
For the graph $X_1$, which is in Figure \ref{X1X2},
$\theta_{X_1}(\beta,\gamma)=1+3 \beta^{2}+\beta^{3}\gamma^{2}$.
For the graph $X_2$, which is also in Figure \ref{X1X2},
$\theta_{X_2}(\beta,\gamma)=1+2\beta+\beta^{2}+\beta^{3}\gamma^{2}$.
\end{example}

\begin{figure}
\begin{center}
\includegraphics[scale=0.3]{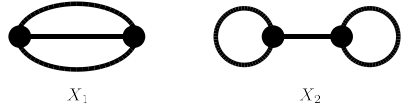} 
\vspace{-1mm}
\caption{Graph $X_1$ and $X_2$}
\label{X1X2}
\end{center}
\end{figure}

\subsection{$\theta_G$ as a Tutte's V-function}
In 1947 \cite{Tring}, Tutte defined
a class of graph invariants called V-function.
The definition is as follows.
\begin{defn}
Let $\mathcal{G}$ be the set of isomorphism classes of finite undirected
 graphs, with loops and multiple edges allowed.
Let $R$ be a commutative ring.
A map $\mathcal{V} : \mathcal{G} \rightarrow R$
is called a {\it V-function}
if it satisfies the following two conditions:
\begin{align*}
&{\rm(i)} \quad 
\mathcal{V}(G)=\mathcal{V}(G \backslash e)+\mathcal{V}(G/e)
\qquad \text{if $e \in E$ is not a loop,} \qquad \qquad \\
&{\rm(ii)} \quad
\mathcal{V}(G_1 \cup G_2)= \mathcal{V}(G_1) \mathcal{V}(G_2). \qquad \qquad
\end{align*}
\end{defn}

Our graph invariant $\theta$ is essentially an example of 
a V-function.
In the definition of V-functions, the coefficients of the
deletion-contraction relation are $1$, while those of $\theta$
are $(1-\beta)$ and $\beta$.
If we modify $\theta$ to
\begin{equation}
\hat{\theta}_{G}(\beta,\gamma):= (1-\beta)^{-|E|+|V|}\beta^{-|V|}
 \theta_{G}(\beta,\gamma) ,  \nonumber
\end{equation}
this is a V-function
$\hat{\theta} : \mathcal{G} \rightarrow  \mathbb{Z}[\beta,\gamma,\beta^{-1},(1-\beta)^{-1}]$.

By successive applications of the conditions of a V-function,
we can reduce the value at any graph to the values at bouquet graphs.
Therefore we can say that a V-function is completely determined by its
boundary condition, i.e., the values at the bouquet graphs.
Conversely,
Tutte shows in \cite{Tring} that for an arbitrary boundary condition,
there is a V-function that satisfies it.
More explicitly, the V-function satisfying a boundary condition
$\{ \mathcal{V}(B_n) \}_{n=0}$
is given by
\begin{equation}
 \mathcal{V}(G)= \sum_{s \subset E}
\prod_{n=0} z_n^{i_n(s)}, \label{V-funcTutterep}
\end{equation}
where $z_n:=\sum_{j=0}^{n}{n \choose j}(-1)^{n+j}\mathcal{V}(B_j)$ and
$i_n(s)$ is the number of connected components of the subgraph
$s$ with nullity $n$.
In our case of $\hat{\theta}$, the expansion Eq.~(\ref{V-funcTutterep}) is equivalent 
to the (e) of Proposition \ref{propbasictheta}.

The Formulas (\ref{defthetaeq}) and (e) of Proposition \ref{propbasictheta} are both 
represented in the sum of the subsets of edges,
but the terms of a subset are different.
Generally, a V-function does not have a representation
corresponding to Eq.~(\ref{defthetaeq});
this representation makes $\theta$ worthy to be investigated among V-functions.

\subsection{Comparison with Tutte polynomial}
The most famous example of a V-function is the Tutte polynomial
multiplied with a trivial factor.
Many graph polynomials, which appear in computer science, engineering, statistical physics, etc.,
are equivalent to the Tutte polynomial or its specialization \cite{EMinter1}.
The {\it Tutte polynomial} is defined by
\begin{equation}
 T_{G}(x,y):=\sum_{s \subset E}
(x-1)^{r(G)-r(s)}
(y-1)^{n(s)}.\label{defTutte}
\end{equation}
It satisfies a deletion-contraction relation
\begin{equation*}
T_{G}(x,y)=
\begin{cases}
 x T_{G \backslash e}(x,y)  \qquad \qquad \qquad \qquad \text{ if $e$ is a bridge,} \\
 y T_{G \backslash e}(x,y)  \qquad \qquad \qquad \qquad \text{ if $e$ is a loop,}   \\
 T_{G \backslash e}(x,y)+ T_{G / e}(x,y) \qquad \quad \text{ otherwise. }
\end{cases} 
\end{equation*}
It is easy to see that
$\hat{T}_{G}(x,y):=(x-1)^{k(G)}T_{G}(x,y)$ is a V-function
to $\mathbb{Z}[x,y]$.
For bouquet graphs, $\hat{T}_{B_n}(x,y)=(x-1)y^{n}$.
In the case of Tutte polynomial, 
Eq.~(\ref{V-funcTutterep}) derives 
Eq.~(\ref{defTutte}).

The V-functions 
$\hat{\theta}$ and $\hat{T}$ are 
essentially different.
The assertion in the following remark implies the
difference irrespective of transforms between $(\beta,\gamma)$ and $(x,y)$.
\begin{remark}
\label{propthetaTutte}
For any field $K$, inclusions
$\phi_1: \mathbb{Z}[\beta,\gamma,\beta^{-1},(1-\beta)^{-1}] \hookrightarrow K$,
and $\phi_2: \mathbb{Z}[x,y] \hookrightarrow K$,
we have
\begin{equation}
\phi_1 \circ \hat{\theta} \neq \phi_2 \circ \hat{T}. \nonumber
\end{equation}
\end{remark}
\begin{proof}
It is easy to see that
$\phi_2(\hat{T}_{B_n}) /\phi_2(\hat{T}_{B_0})= \phi_2(y)^{n}$ and
$\phi_1(\hat{\theta}_{B_n}) / \phi_1(\hat{\theta}_{B_0}) 
= \phi_1(1-\beta )^{-n} \phi_1(\sum_{k=0}^{n} {n \choose k} f_{2k}(\gamma) \beta^{k}).$
If $\phi_1 \circ \hat{\theta} = \phi_2 \circ \hat{T}$,
then
$a_n:=$
$\sum_{k=0}^{n} {n \choose k} f_{2k}(\gamma') \beta'^{k}$
$= z^{n}$
for some $z \in K$, where $\gamma'=\phi_1(\gamma)$ and $
 \beta'=\phi_1(\beta)$.
The equation $a_1^{2} = a_2$ gives $\gamma'^{2}\beta'^{2}=0$.
This is a contradiction because $\beta \neq 0$ and $\gamma \neq 0$.
\end{proof}

As suggested in the proof of the above remark,
if we set $\gamma=0$, the polynomial $\theta_{G}(\beta,0)$ is included in 
the Tutte polynomial.
\begin{prop}
\label{thmgamma0}
\begin{equation*}
 \theta_{G}(\beta,0)=
(1-\beta)^{n(G)}\beta^{r(G)}
T_{G}
\Big(
\frac{1}{\beta},\frac{1+\beta}{1-\beta}
\Big).
\end{equation*}
\end{prop}
\begin{proof}
From Proposition \ref{propbasictheta}.(b)
and $f_{2k}(0)=1$, we have
\begin{equation*}
 \hat{\theta}_{B_n}(\beta,0)=(1-\beta)^{1-n}\beta^{-1}
\sum_{k=0}^{n} {n \choose k} \beta^{k} =(1-\beta)^{1-n}\beta^{-1} (1+\beta)^{n}.
\end{equation*}
We also have 
$\hat{T}_{B_n}(\frac{1}{\beta},\frac{1+\beta}{1-\beta})=(\beta^{-1}-1)(\frac{1+\beta}{1-\beta})^{n}$.
Therefore 
$\hat{\theta}_{B_n}(\beta,0)=\hat{T}_{B_n}(\frac{1}{\beta},\frac{1+\beta}{1-\beta})$.
Since V-functions are determined by the values at the bouquet graphs,
 $\hat{\theta}_{G}(\beta,0)=\hat{T}_{G}(\frac{1}{\beta},\frac{1+\beta}{1-\beta})$
 holds for any graph $G$.
\end{proof}

This result is natural in the view of the Ising partition function Eq.~(\ref{defIsing})
with uniform coupling constant $J$ and no external fields ($h_i=0$).
Here, for simplicity, we call it the simple Ising partition function.
The Tutte polynomial is equivalent to the partition
function of the q-Potts model \cite{Bollobas}, where $q$ is the number of allowed states at each vertex.
If we set $q=2$, it becomes the simple Ising partition function.
In terms of the Tutte polynomial, it correspond to the parameters
$(x,y)=(\frac{1}{\beta},\frac{1+\beta}{1-\beta})$.
Therefore, $T_{G}(\frac{1}{\beta},\frac{1+\beta}{1-\beta})$
is the simple Ising partition function in essence.
On the other hand, at a point of $\gamma=0$, $\theta_{G}(\beta,0)$
is also the simple Ising partition function essentially,
because
the representation of $\theta_{G}(\beta,0)$ in the sum of sub-coregraphs
coincides with the expansion of van der Waerden \cite{Welsh}.

We can say that the Tutte polynomial is an extension of 
the simple Ising partition function 
to the $q$-state model 
while
the polynomial $\theta$ is an extension of it
to a model with specific form of local external fields.

%
%
%
%

%% file: chapter7b.tex

\section{Univariate graph polynomial $\omega_G$}
\label{sec:omega}
\subsection{Definition and basic properties}
In this section we define the second graph polynomial
$\omega$ by setting $\gamma=2\sqrt{-1}$.
It is easy to check that
$f_n(2\sqrt{-1})=(\sqrt{-1})^{n}(1-n)$,
using Eq.~(\ref{defindf}).
Therefore
\begin{equation}
\theta_{G}(\beta,2\sqrt{-1})
=
\sum_{s \subset E} 
(- \beta)^{|s|}
\prod_{i \in V}(1-d_i(s)). \label{theta2i} \\ 
\end{equation}
An interesting point of this specialization is 
the relation to the monomer-dimer partition function
with specific form of monomer-dimer weights,
as described in Section \ref{subsectionomegamonomerdimer}.
Furthermore the product of $(1-d_i(s))$ resembles the \ls
of the perfect matching problem given in Theorem \ref{thm:PMLS}.

From Eq.~(\ref{lembeta1eq1}), $\theta_G(1,2\sqrt{-1})=0$
unless all the nullities of connected components of $G$ are less than $2$.
The following theorem asserts that $\theta_G(\beta,2\sqrt{-1})$ can be divided by
$(1-\beta)^{|E|-|V|}$.
We define $\omega_G$ by dividing that factor.
\begin{thm}
\label{thmomega}
\begin{equation}
 \omega_{G}(\beta):= \frac{\theta_{G}(\beta,2 \sqrt{-1})}{(1-\beta)^{|E|-|V|}} \quad
 \in \mathbb{Z}[\beta].  \nonumber
\end{equation}
\end{thm}
In Eq.~(\ref{theta2i}), $\theta_{G}(\beta,2 \sqrt{-1})$ is given by the summation over all
sub-coregraphs and each term is not necessarily divisible by $(1-\beta)^{|E|-|V|}$.
If we use the representation in (e) of Proposition \ref{propbasictheta}, however, each summand
is divisible by the factor as
we show in the following theorem.
Theorem \ref{thmomega} is a trivial consequence of Theorem \ref{thmomegaaltrep}.
\begin{thm}
\label{thmomegaaltrep}
\begin{equation}
\omega_{G}(\beta)
=
\sum_{s \subset E } 
\beta^{|s|}
\prod_{n=0}
h_{n}(\beta)^{i_n(s)}, \nonumber
\end{equation}
where $h_0(\beta):=(1-\beta)$, $h_1(\beta):=2$ and 
$h_n(\beta):=0$ for $n \geq 2$. 
\end{thm}
\begin{proof}
Proposition \ref{propbasictheta} (b) and
$f_m(2\sqrt{-1})=(\sqrt{-1})^{m}(1-m)$, we have
\begin{align}
\theta_{B_n}(1,2\sqrt{-1})
=
\sum_{k=0}^{n}{n \choose k}
(-1)^{k}(1-2k)
=
\begin{cases}
1 
\quad \text{  if }n=0 \\
2
\quad \text{  if }n=1 \\
0
\quad \text{  if }n \geq 2.
\end{cases}  \nonumber
\end{align}
Proposition \ref{propbasictheta} (e) gives
\begin{align*}
\omega_{G}(\beta)
&=
\sum_{s \subset E}
\prod_{n=0}
\theta_{B_n}(1,2\sqrt{-1})^{i_n(s)}
\beta^{|s|}
(1-\beta)^{|V|-|s|} \\
&=
\sum_{s \subset E}
\prod_{n=0}
[(1-\beta)^{1-n}
\theta_{B_n}(1,2\sqrt{-1})
]^{i_n(s)}
\beta^{|s|}.
\end{align*}
Then the assertion is proved.
\end{proof}

\begin{example}
\label{exampleomega}
$\\$
For a tree $T$, $\omega_{T}(\beta)=1-\beta$.
For the cycle graph $C_n$,
$\omega_{C_n}(\beta)=1+\beta^{n}$.
For the complete graph $K_4$,
$\omega_{K_4}(\beta)=1+2\beta+3\beta^{2}+8\beta^{3}+16\beta^{4}$.
For graphs in Figure \ref{X1X2},
$\omega_{X_1}(\beta)=1+\beta+4\beta^{2}$ and
$\omega_{X_2}(\beta)=1+3\beta+4\beta^{2}$.
\end{example}

We list basic properties of $\omega$ below.
\begin{prop}$\quad$ 
\label{propbasicomega}
\begin{itemize}
\item[{\rm (a)}]
$\omega_{G_1 \cup G_2}(\beta)=\omega_{G_1}(\beta)\omega_{G_2}(\beta). $
\item[{\rm (b)}]
$\omega_{G}(\beta)=\omega_{G \backslash e}(\beta)+\beta
\omega_{G / e}(\beta )
\quad
\text{if } e \in E \text{ is not a loop.}$
\item[{\rm (c)}]
$\omega_{B_n}(\beta)=1+(2n-1)\beta$.
\item[{\rm (d)}]
$ \omega_{G}(\beta)=\omega_{\core(G)}(\beta). $
\item[{\rm (e)}]
$\omega_{G}(\beta)$ is a polynomial of degree $|V_{\core(G)}|$.
The leading coefficient is $\prod_{i \in V_{\core(G)}} (d_i -1)$
and the constant term is $1$.
\item[{\rm (f)}]
Let $G^{(m)}$ be the graph obtained by subdividing each edge to m edges. Then,
 \begin{equation*}
 \omega_{G^{(m)}}(\beta)=(1+\beta+ \cdots + \beta^{m-1})^{|E|-|V|}\omega_{G}(\beta^m).
 \end{equation*}
\end{itemize}
\end{prop}
\begin{proof}
The assertions (a-e) are easy.
(f) is proved by $|E_G|-|V_G|=|E_{G^{(m)}}|-|V_{G^{(m)}}|$ and
$\theta_{G^{(m)}}(\beta,2\sqrt{-1})
=\theta_{G}(\beta^{m},2\sqrt{-1})$.
\end{proof}

\begin{prop} 
\label{propomeganonneg}
If $G$ does not have connected components of nullity $0$, then
the coefficients of $\omega_{G}(\beta)$ are non-negative.
\end{prop}
\begin{proof}
We prove the assertion by induction on the number of edges.
Assume that every connected component is not a tree.
If $G$ has only one edge,
then $G=B_1$ and the coefficients are non-negative.
Let $G$ have $M (\geq 2)$ edges and assume that
the assertion holds for the graphs with at most $M-1$ edges.
It is enough to consider the case that $G$ is a connected coregraph
because of Proposition \ref{propbasicomega}.(a) and (d).
If all the edges of $G$ are loops, $G=B_n$ for some $n \geq 2$ 
and the coefficients are non-negative.
If $G=C_M$, the coefficients are also non negative as in Example
\ref{exampleomega}. 
Otherwise, we reduce $\omega_{G}$ 
to graphs with nullity not less than $1$
by an application of the
deletion-contraction relation 
and
see that the coefficients of 
$\omega_{G \backslash e}$ and $\omega_{G/e}$ are both non-negative by the induction hypothesis.
\end{proof}

\subsection{Relation to monomer-dimer partition
  function} \label{subsectionomegamonomerdimer}
In the next theorem, we prove that
the polynomial $\omega_{G}(\beta)$
is the monomer-dimer partition function with specific form of weights.

As defined in Section \ref{sec:perfectmatching},
a {\it matching} of $G$ is a set of edges such that any edges do not occupy a
same vertex. 
It is also called a {\it dimer arrangement} in statistical
physics \cite{HLmonomerdimer}.
We use both terminologies.
The number of edges in a matching ${\bf D}$ is denoted by $|{\bf D}|$.
If a matching ${\bf D}$ consists of $k$ edges, then it is called a {\it k-matching}. 
The vertices covered by the edges in ${\bf D}$ are denoted by
$[{\bf D}]$. 
The set of all matchings of $G$ are denoted by $\mathcal{D}$.

The monomer-dimer partition function with edge weights 
$\bs{\mu}=(\mu_{e})_{e \in E}$ and vertex weights 
$\bs{\lambda}=(\lambda_{i})_{i \in V}$
is defined by
\begin{equation}
\Xi_G(\bs{\mu},\bs{\lambda})
:=
\sum_{{\bf D} \in \mathcal{D}}
\prod_{e \in {\bf D}} \mu_{e}
\prod_{i \in V \backslash [{\bf D}]}  \lambda_{i}.  \nonumber
\end{equation} 
We write $\Xi_G(\mu,\bs{\lambda})$ if all weights $\mu_{e}$ are set
to be the same $\mu$.

\begin{thm} 
\label{thmmonomer}
Let $\lambda_{i}:=1+(d_i - 1) \beta$, then
\begin{equation}
\omega_{G}(\beta)
=
\Xi_{G}(- \beta, \bs{\lambda}).  \nonumber
\end{equation}
\end{thm}
\begin{proof}
We show that $\Xi_{G}(- \beta, \bs{\lambda})$ satisfies 
the deletion-contraction relation and the boundary condition of the form in 
Proposition \ref{propbasicomega}.(c).
For the bouquet graph $B_n$, ${\bf D}= \phi$ is the only possible
dimer arrangement, and thus
\begin{equation}
\Xi_{B_n}(- \beta, \bs{\lambda})=1+(2 n-1)\beta =\omega_{B_n}(\beta). \nonumber
\end{equation}
For a non-loop edge $e=i_0j_0$, we show that the deletion-contraction
 relation is satisfied.
A dimer arrangement ${\bf D} \in \mathcal{D}$ is classified into the following
 five types: 
(a) ${\bf D}$ includes $e$, 
(b) ${\bf D}$ does not include $e$ and ${\bf D}$ covers both $i_0$ and $j_0$, 
(c) ${\bf D}$ covers $i_0$ while does not cover $j_0$, 
(d) ${\bf D}$ covers $j_0$ while does not cover $i_0$, 
(e) ${\bf D}$ covers neither $i_0$ nor $j_0$.  
According to this classification, $\Xi_{G}(- \beta, \bs{\lambda})$ is a
 sum of the five terms $A,B,C,D$ and $E$.
We see that
\begin{align*}
C
=
\sum_{{\bf D} \in \mathcal{D} \atop [{\bf D}] \ni i_0, [{\bf D}] \not\ni j_0}
&(- \beta)^{|{\bf D}|}
\prod_{i \in V \backslash [{\bf D}]}
\lambda_{i}
\\
=
\sum_{{\bf D} \in \mathcal{D} \atop [{\bf D}] \ni i_0, [{\bf D}] \not\ni j_0}
&(- \beta)^{|{\bf D}|}
(1+(d_{j_0}-2)\beta)
\prod_{i \in V \backslash [{\bf D}] \atop i \neq j_0}
\lambda_{i}  \nonumber \\
&+
\beta
\sum_{{\bf D} \in \mathcal{D} \atop [{\bf D}] \ni i_0, [{\bf D}] \not\ni j_0}
(- \beta)^{|{\bf D}|}
\prod_{i \in V \backslash [{\bf D}] \atop i \neq j_0}
\lambda_{i}   \\
=:C_1 + \beta C_2.
\end{align*}
In the same way, $D=D_1 + \beta D_2$.
Similarly,
\begin{align*}
E
&=
\sum_{{\bf D} \in \mathcal{D} \atop [{\bf D}] \not\ni i_0, [{\bf D}] \not \ni j_0}
(- \beta)^{|{\bf D}|}
\lambda_{i_0}\lambda_{j_0}
\prod_{i \in V \backslash [{\bf D}] \atop i \neq i_0,j_0}
\lambda_{i} \\
&=
\sum_{{\bf D} \in \mathcal{D} \atop [{\bf D}] \not\ni i_0, [{\bf D}] \not \ni j_0}
(- \beta)^{|{\bf D}|}
(1+(d_{i_0}-2)\beta)(1+(d_{j_0}-2)\beta)
\prod_{i \in V \backslash [{\bf D}] \atop i \neq i_0,j_0}
\lambda_{i} \nonumber \\
&+
\beta
\sum_{{\bf D} \in \mathcal{D} \atop [{\bf D}] \not\ni i_0, [{\bf D}] \not \ni j_0}
(- \beta)^{|{\bf D}|}
(2+(d_{i_0}+d_{j_0}-3)\beta) 
\prod_{i \in V \backslash [{\bf D}] \atop i \neq i_0,j_0}
\lambda_{i} \\
&=:E_1 + \beta E_2.
\end{align*}
We can straightforwardly check that
\begin{equation} 
\Xi_{G \backslash  e}(- \beta, \bs{\lambda '})
=
B+C_1+D_1+E_1  \nonumber
\end{equation}
and
\begin{equation} 
\beta \Xi_{G / e }(- \beta, \bs{\lambda ''}) 
=
A+\beta C_2 +\beta D_2 + \beta E_2, \label{thmmonomereq10}
\end{equation}
where $\bs{\lambda'}$ and $\bs{\lambda''}$ are defined by the degrees of 
$G \backslash e$ and $G / e$ respectively.
Note that $C_2+D_2$ in Eq.~(\ref{thmmonomereq10}) corresponds to the dimer
 arrangements in $G / e$ that cover the new vertex formed by the contraction.
This shows the deletion-contraction relation.
\end{proof}
Let $p_{G}(k)$ be the number of k-matchings of $G$. 
The {\it matching polynomial} $\alpha_{G}$ is defined by
\begin{equation}
\alpha_{G}(x)= \sum_{k=0}^{\lfloor \frac{|V|}{2} \rfloor}(-1)^{k}
p_{G}(k) x^{|V|-2k}.  \nonumber
\end{equation}
The matching polynomial is essentially the monomer-dimer
partition function with uniform weights;
if we set all vertex weights $\lambda$ and all edge weights $\mu$ respectively, 
we have
\begin{equation}
\Xi_{G}(\mu, \lambda)=
\alpha_{G}\Big( \frac{\lambda}{\sqrt{- \mu}} \Big)
{\sqrt{- \mu}}^{|V|}. \nonumber
\end{equation}
Therefore, for a $(q+1)$-regular graph $G$, 
Theorem \ref{thmmonomer} implies
\begin{equation}
\omega_{G}(u^{2})
=
\alpha_{G}\Big(
\frac{1}{u} +q u \Big)
u^{|V|}.\label{cormatchingeq1}
\end{equation}
In \cite{Nseries}, Nagle derives a sub-coregraph expansion of the
monomer-dimer partition function with uniform weights,
or matching polynomials, 
on regular graphs.
With a transform of variables, 
his expansion theorem
is essentially equivalent to
Eq.~(\ref{cormatchingeq1}).
We can say that
Theorem \ref{thmmonomer} gives an extension of the
expansion to non-regular graphs.

As an immediate consequence of Eq.~(\ref{cormatchingeq1}),
we remark on the
symmetry of the coefficients of 
$\omega_{G}$ for regular graphs.
\begin{cor}
Let $G$ be a $(q+1)-$regular graph $(q \geq 1)$ with $N$ vertices and
$w_k$ be the $k$-th coefficient of $\omega_{G}(\beta)$.
Then we have 
\begin{equation}
\qquad \qquad
w_{N-k}=w_{k} q^{N-2k}
\quad \quad
\text{ for }  0 \leq k \leq N.  \nonumber
\end{equation}
\end{cor}

\subsection{Zeros of $\omega_{G}(\beta)$}
Physicists are interested in
the complex zeros of partition functions, 
because it restricts the occurrence of phase transitions,
i.e., discontinuity of physical quantities with respect to parameters
such as temperature.
In the limit of infinite size of graphs,
analyticity of the scaled log partition function 
on a complex domain
is guaranteed 
if there are no zeros in the domain and
some additional conditions hold.
(See \cite{YLstat1,Sbounds}.)
For the monomer-dimer partition
function, 
Heilman and Lieb \cite{HLmonomerdimer}
show the following result.
\begin{thm}
[\cite{HLmonomerdimer} Theorem 4.6.] \label{thmmonomerdimerzero}
If $\mu_{e} \geq 0$ for all $e \in E$ and
${\rm Re}(\lambda_{j}) > 0 $  for all $j \in V$
then
$\Xi_G(\bs{\mu},\bs{\lambda}) \neq 0$.
The same statement is true if 
${\rm Re}(\lambda_{j}) < 0 $  for all $j \in V$.
\end{thm}

Since our polynomial $\omega_{G}(\beta)$  
is a monomer-dimer partition function,
we obtain a bound of the region of complex zeros.

\begin{cor}
\label{cor:omegazeros}
Let $G$ be a graph and
let $d_m$ and $d_M$ be the minimum and maximum degree in
$\core(G)$ respectively
and assume that $d_m \geq 2$.
If $\beta \in \mathbb{C}$ satisfies $\omega_{G}(\beta)=0$,
then
\begin{equation}
\frac{1}{d_M -1} \leq |\beta| \leq \frac{1}{d_m -1}.  \nonumber
\end{equation}
\end{cor}
\begin{proof}
Without loss of generality, we assume that $G$ is a coregraph.
Let  $\beta = |\beta| {\rm e}^{\iunit \theta}$ satisfy
$\omega_{G}(\beta)=0$,
where
$0 \leq \theta < 2 \pi$ and $\iunit$ is the imaginary unit.
Since $\omega_{G}(0)=1$ and the coefficients of $\omega_{G}(\beta)$ is
 not negative from Proposition \ref{propomeganonneg},
we have $\beta \neq 0$ and $\theta \neq 0$.
We see that
\begin{equation}
\omega_{G}(\beta)
=
\Xi_G(-\beta,\bs{\lambda}) 
=
\Xi_G(|\beta|,\iunit {\rm e}^{-\iunit \theta /2} \bs{\lambda})
(\iunit{\rm e}^{-\iunit \theta /2})^{-|V|},  \nonumber
\end{equation}
where $\lambda_j=1+(d_j-1)\beta$,
and ${\rm Re}(\iunit {\rm e}^{-\iunit \theta /2}\lambda_{j})$
$=(1-(d_j-1)|\beta|) \sin\frac{\theta}{2}$.
From Theorem \ref{thmmonomerdimerzero},
the assertion follows.
\end{proof}
Especially, if the graph is a $(q+1)$-regular graph,
the roots are on the circle of radius $1/q$,
which is also directly seen by Eq.~(\ref{cormatchingeq1})
combining the famous result on the roots of matching polynomials
\cite{HLmonomerdimer}:
the zeros of matching polynomials are on the real interval
$(-2\sqrt{q},2 \sqrt{q})$.

\subsection{Determinant sum formula}
Let 
$\mathcal{T}:=\{C \subset E ; d_{i}(C)=0 \text{ or } 2 \text{ for all }i \in V \}$
be the set of unions of vertex-disjoint cycles.
In this subsection,
an element $C \in \mathcal{T}$ 
is identified with the subgraph $(V_{C},C)$,
where $V_{C}:=\{i \in V ; d_i(C) \neq 0\}$.
A graph $G \smallsetminus C$ is given by deleting all the vertices in $V_{C}$
and the edges of $G$ that are incident with them.

The aim of this subsection is Theorem \ref{coromega},
in which we represent $\omega_{G}$ as a sum of determinants.
This theorem is similar to the expansion of the matching polynomial
by characteristic polynomials \cite{GGmatching};
\begin{equation}
\alpha_{G}
(x)
=
\sum_{C \in \mathcal{T}}
2^{k(C)}
\det[x I - A_{G \smallsetminus C }], \label{expansionmatching}
\end{equation}
where $A_{G \smallsetminus C }$ is the adjacency matrix of $G \smallsetminus C$
and
$k(C)$ is the number of connected components of $C$.

\begin{thm}
\label{coromega}
\begin{equation}
\omega_{G}
(u^{2})
=
\sum_{C \in \mathcal{T}}
2^{k(C)}
\det
\Big(
[
I
-u {A_{G}}
+
u^{2}(D_{G}-I)
]
\Big|_{G\smallsetminus C}
\Big)
u^{|C|}, \label{expansionomega}
\end{equation}
where
$D_G$ is the degree matrix defined by 
$(D_G)_{i,j}:= d_i \delta_{i,j} $ and
$\cdot \big|_{G \smallsetminus C}$ denotes the restriction
to the principal minor indexed by the vertices of $G \smallsetminus C$.
\end{thm}
\begin{proof}
For the proof, we use the result of Chernyak and Chertkov
 \cite{CCfermion2}.
For given weights $\bs{\mu}=(\mu_e)_{e \in E}$ and 
$\bs{\lambda}=(\lambda_i)_{i \in V}$,
a $|V| \times |V|$  matrix $H$ is defined by
\begin{equation}
H:={\rm diag}(\bs{\lambda})- \sum_{e \in E} \sqrt{- \mu_e} A_e, \nonumber
\end{equation}
where $A_{e}=E_{i,j}+E_{j,i}$ for $e=ij$ and $E_{i,j}$ is the matrix base.
In our notation, their result implies
\begin{equation}
\Xi_G(\bs{\mu},\bs{\lambda}) 
=
\sum_{C \in \mathcal{T}}
2^{k(C)}
\det
H |_{G\smallsetminus C}
\prod_{e \in C}\sqrt{- \mu_e}. \nonumber
\end{equation}
If we
set $\lambda_i=1+(d_i - 1)u^2$ and $\sqrt{- \mu_e}=u$,
then the assertion follows.
\end{proof}

For regular graphs, Eqs.~(\ref{expansionmatching}) and (\ref{expansionomega}) are equivalent
because of Eq.~(\ref{cormatchingeq1}).
It is noteworthy that the matrix $\left(I -u {A_{G}}+u^{2}(D_{G}-I) \right)$ is nothing but 
the matrix that appear in the \IB formula of the Ihara zeta function.
It is also noteworthy that the region of zeros in Corollary \ref{cor:omegazeros}
resembles the region of poles of Ihara zeta function derived from Eq.~(\ref{eq:PFboundMgraph}).

\subsection{Values at $\beta=1$}
The value of $\omega_{G}(1)$ is interpreted as the number of a set
constructed from $G$.
For the following theorem, recall that $G^{(2)}$ is
obtained by adding a vertex on each edge in $G=(V,E)$.
The vertices of $G^{(2)}:=(V^{(2)},E^{(2)})$ are classified into 
$V_{O}$ and $V_{A}$,
where $V_{O}$ is the original vertices and 
$V_{A}$ is the ones newly added.
The set of matchings on $G^{(2)}$ is denoted by
$\mathcal{D}_{G^{(2)}}$.
\begin{thm} 
\label{thmomega1}
\begin{equation}
\omega_{G}(1)
=
|\{
{\bf D} \in \mathcal{D}_{G^{(2)}}
;
[{\bf D}]  \supset  V^{}_{O}
\}|.                   \nonumber
\end{equation}
\end{thm}
\begin{proof} 
From Theorem \ref{thmomegaaltrep}, we have
\begin{equation}
\omega_{G}(1)
=
\sum_{s \subset E, s=G_1 \cup \cdots \cup G_{k(s)} 
\atop n(G_j)=1 \text{ for } j=1\ldots k(s)}
 2^{k(s)}, \label{thmomega1eq2}
\end{equation}
where $G_j$ is a connected component of $(V,s)$.
We construct a map $F$ from 
$\{ {\bf D} \in \mathcal{D}_{G^{(2)}}; [{\bf D}] 
\supset  V^{}_{O}  \}$
to $s \subset E$
as 
\begin{equation}
F({\bf D}):= 
\{ e \in E ;
\text{ the half of }e \text{ is covered by an edge in } {\bf D}
\}. \nonumber
\end{equation}
Then the nullity of each connected component of $F({\bf D})$ is $1$ and
 $|F^{-1}(s)|=2^{k(s)}$.
\end{proof}

\begin{example}
For the graph $X_3$ in Figure \ref{figomega1}, 
$\omega_{X_3}(1)=\omega_{C_3}(1)=2$.
The corresponding arrangements are also shown in Figure \ref{figomega1}.
\end{example}

\begin{figure}
\begin{center}
\includegraphics[scale=0.28]{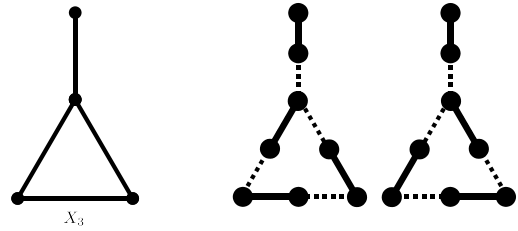}
\vspace{-1mm}
\caption{Graph $X_3$ and possible arrangements on $X^{(2)}_3$.}
\label{figomega1}
\end{center}
\end{figure}

In the end, we remark on the relations between the results on
$\omega_{G}(1)$ obtained in this paper. 
From Proposition \ref{propbasicomega},
$\omega_{G}(1)$ satisfies
\begin{equation}
\omega_{G}(1)=\omega_{G \backslash e}(1)+\omega_{G / e}(1)
\quad 
\text{ if }
e \in E \text{ is not a loop.}  \nonumber
\end{equation}
This relation can be directly observed from the interpretation
of Theorem \ref{thmomega1}.
Theorem \ref{thmmonomer} 
gives
\begin{equation}
 \omega_{G}(1)=
\sum_{{\bf D} \in \mathcal{D}} \hspace{-0.5mm} (-1)^{|{\bf D}|} \hspace{-2mm}
\prod_{i \in V \backslash [{\bf D}]} \hspace{-2mm} d_{i},  \nonumber
\end{equation}
which can be proved from Theorem \ref{thmomega1}
with the inclusion-exclusion principle.
Theorem \ref{coromega} gives
\begin{equation}
\omega_{G}(1)
 =
\sum_{C \in \mathcal{T}} 2^{k(C)}
\det{[D_G-A_G]
\Big|_{G\smallsetminus C}
}. \nonumber
\end{equation}
We can directly prove this formula from Eq.~(\ref{thmomega1eq2})
using a kind of matrix-tree theorem.

\section{Discussion}
In this chapter, we analyzed the LS ignoring the relations between the weights $\bsbeta$ and $\bsgamma$.
In other words, we treated the LS as a weighted graph polynomial $\Theta_G(\bsbeta,\bsgamma)$.
Under the treatment, we derived strict bounds on the number of sub-coregraphs, 
which scales with the nullity of graphs.
We also showed that $\Theta_G$ satisfies \dcr, assuming the vertex weights on the ends of the contracted edge are the same.
Though the result does not have direct implication for the properties of the Bethe approximation,
it demonstrates rich mathematical structures of the LS.

Specializing $\Theta_G(\bsbeta,\bsgamma)$, we introduced two graph polynomials and elucidated their properties. 
These are new instances of Tutte's V-function, allowing alternative sum expression with respect to all the subgraphs.
For the univariate graph polynomial $\omega_G$, we found interesting property such as the relation
to the monomer-dimer partition function and a little connection to the Ihara zeta function.

%% file: chapter8.tex

\section{Conclusion}
In this thesis, we analyzed mathematical properties of \lbp algorithm in emphasis of the graph geometry.
The exact inference on a graph requires ``global computation,'' which is computationally intractable,
whereas the approximate inference by the LBP algorithm only requires ``local computation.''
The global/local discrepancy is the origin of the approximation errors of the LBP algorithm.
The gap between the global and the local disappears if the graph geometry is trivial, i.e., tree.

This concept is not restricted to the LBP algorithm.
In fact, we often encounter ``global'' computational problem which is approximated by ``local'' computations 
such as message passing algorithms on graphs. 
Obviously, the max-product algorithm, which gives exact result of maximization problems associated with trees,
has the same difficulty.

In Part I, we introduced the graph zeta function and showed the \Bzf.
Since the LBP fixed points are characterized in terms of the \Bfe function,
the graph geometry should be reflected in the function.
The \Bzf claims that the graph zeta function is the key quantity that reflects the graph geometry
in the context of LBP algorithm.

The novel relationship between LBP, or the \Bfe function, and the graph zeta function
provides new techniques for the analysis of the properties of LBP and the \Bfe function.
We demonstrated applications of the techniques in this thesis.
For example, we showed that the region where the Hessian of the \Bfe function is related to
the nearest pole of the Ihara zeta function.
We also showed that locally stable fixed points of LBP are local minima of the \Bfe function.
For a certain class of models on graphs with nullity two,
the uniqueness of the LBP fixed point is proved by checking positives of the graph zeta function.

Since the relationship between LBP and the \Bfe is clarified by Yedidia et al \cite{YFWGBP},
many variants of the LBP algorithm have been proposed based on the understanding.
We believe that our new relation to the graph zeta function also opens the door
to the future developments or improvements of LBP algorithm.

In Part II, we investigated into the \ls.
Since the \ls is the sum with respect to sub-coregraphs,
the form of the \ls expansion reflects the graph geometry.
In fact, it is equal to $1$ if the underlying graph is a tree.
Our analysis was basically focused on the expression itself,
leaving the relations between the weights $\bsbeta$ and $\bsgamma$.

We analyzed mathematical properties of $\Theta_G$ 
and showed interesting properties such as \dcr.
We also showed partial connection between the \ls and graph zeta function.
However, many problems are left regarding the connection.

\section{Suggestions for future researches}
This section suggests possible extensions and developments of our analysis.

\subsection{Variants and extensions of the \Bzf}
As mentioned in Section \ref{sec:PreAdd}, there are many variants and extensions of the LBP algorithm.
Accordingly, it is natural to think of variants and extensions of \Bzf.

~\\
{\bf Fractional belief propagation:}
This extension is possible, using the Bartholdi type graph zeta function.
This extension will be discussed in a future paper.

~\\
{\bf Generalized belief propagation:}
Another possible direction of the extension of the formula is Generalized Belief Propagation (GBP).
We have not considered this extension.
The zeta function appear in this extension may be interesting from combinatorics view point.
And may prove or disprove the statement: ''locally stable fixed points of GBP are local minima of the Kikuchi free energy.''

~\\
{\bf Expectation propagation:}
We can also think of the extension to expectation propagation.
In the method, local exponential families are glued together
by local consistency condition of expectations of sufficient statistics.
In the proof of Bethe-zeta formula, key property was
$\var{b_{\alpha}}{\phi_i}=\var{b_i}{\phi}$, which is not guaranteed by the consistency of expectations.
Therefore, the extension of the \Bfe is nothing apparent.

\subsection{Dynamics and convergence of LBP algorithm}
In chapter \ref{chap:unique}, we developed a new approach to show the uniqueness of the LBP fixed point.
An interesting question is how we can extend our approach to show the convergence of the LBP algorithm.
By definition, the convergence property is stronger than the uniqueness property.
However, in binary pairwise case, the uniqueness condition in Corollary \ref{cor:uniqueradius} 
also guarantees the convergence \cite{MKsufficient}.

One vague suggestion for approaching to the dynamics and the convergence of LBP is considering graph covers.
A {\it cover} $\tilde{G}$ of a graph $G$ is a graph having a map $\pi$ to $G$ that is a surjection and a local isomorphism.
The Ihara zeta function has rich connection with graph covers \cite{STzeta1,STzeta2}.
For example, It is well known that $\zeta_G(u)^{-1}$ divides $\zeta_{\tilde{G}}(u)^{-1}$ \cite{MScoverings}. 
The uniqueness and convergence of LBP is also related to graph covers.
Obviously, the uniqueness on $\tilde{G}$ that has induced compatibility function from $G$ guarantees the uniqueness on $G$.  
The uniqueness of the Gibbs measure on the universal covering tree, i.e. the infinite depth computation tree, guarantees the
convergence of the LBP algorithm on $G$ \cite{TJgibbsmeasure}.
These fragmented facts suggests further developments of theories on
graph zeta functions, graph covers and the dynamics of the LBP algorithm.

Finally, it is noteworthy that the Ihara zeta function has an interpretation as a dynamical zeta function.
In general, dynamical zeta functions encodes information of periodic points of the given dynamical systems \cite{AMperiodic}.
It is known that the Ihara zeta function is the dynamical zeta function of a certain symbolic dynamical system derived by the graph \cite{KSzeta}.
It would be interesting to pursue the relation between this dynamical system and the LBP algorithm.

\subsection{Other researches related to LBP and graph zeta function}
Some recent researches have suggested the importance of zeta function. 
In the context of the LDPC codes, which is an important application of LBP, 
Koetter et al have shown the connection between pseudo-codewords and the edge zeta function \cite{KLVWpseudo,KLVWcharacterizations}.
Though there appears graph zeta function, our result and their result are basically different.
In the field of codes, parity check constraints are considered.
Thus the compatibility functions have values of zero.
In contrast, we considered arbitrary positive compatibility functions in applications of the \Bzf.
Compatibility functions with zero values are related to faces of the closure of $L$,
and limits of the \Bfe function to faces are nothing obvious.

For the Gaussian \bp, 
Johnson et al \cite{JCC} give zeta-like product formula of the partition function.

An implicit reason for the appearance of zeta function is the local nature of message passing algorithms.
Local operation does not distinguish covering graphs and the original graph in some sense.
Graph zeta functions are intimate relation to graph covers.
Though their works are not directly related to our work, from such viewpoints,
pursuing connections is an interesting future research topic.

%% file: appendix1.tex

\section{Linear algebraic formulas}
Basic notation is as follows.
The set of $n_1 \times n_2$ matrices is denoted by $\mat{n_1}{n_2}$.
For a square matrix $X$, the set of eigenvalue is denoted by $\spec{X}$ and
the spectral radius, i.e. the maximum modulus of the eigenvalues, is denoted by $\specr{X}$.

\label{sec:linearformula}
\begin{thm}
[\cite{HJmatrix}]%
\label{app:thm:specrcomparison}
Let $X=(x_{ij})$ and $Y=(y_{ij})$ be non-negative matrices satisfying $x_{ij} \leq y_{ij}~(i,j=1,\ldots,n)$,
then $\specr{X} \leq \specr{Y}$.  
\end{thm}

\begin{thm}
[\cite{HJmatrix}] %
\label{app:luboudofspecr}
Let $X=(x_{ij}) \in \mat{n}{n}$ be a non-negative matrix, then
\begin{equation}
  \min_{1 \leq j \leq n} \sum_{i=1}^n x_{ij} \leq \specr{X} \leq \max_{1 \leq j \leq n} \sum_{i=1}^n x_{ij}
\end{equation}  
and
\begin{equation}
  \min_{1 \leq i \leq n} \sum_{j=1}^n x_{ij} \leq \specr{X} \leq \max_{1 \leq i \leq n} \sum_{j=1}^n x_{ij}.
\end{equation}  
\end{thm}

\begin{defn}
Let $X=(x_{ij}) \in \mat{n}{n}$ be a non-negative matrix and
let $G_{X}$ be a directed graph consists of vertices $V=\{1,\ldots,n\}$ 
and directed edges $\edji$ for $x_{ij} \neq 0$.
The matrix $X$ is {\it irreducible} if $G_{X}$ is strongly connected, i.e.,
for each directed pair $(i,j) \in V \times V$, there is a directed walk from $i$ to $j$.
\end{defn}

\begin{thm}
[Perron Frobenius theorem \cite{HJmatrix}]\label{thm:PF}
Let $X$ be a non-negative matrix of size $n$.
Then $\specr{X}$ is an eigenvalue of $X$ having non-negative eigenvector.
The eigenvalue is called the {\it Perron-Frobenius eigenvalue} of $X$.
Furthermore, if $X$ is irreducible, $\specr{X}$ is positive, simple and having the positive eigenvector.
\end{thm}

\begin{prop}
[Schur complement]
\label{app:schur}
Let $X \in \mat{n}{n}$.
Let $Y$ be its inverse. 
The blocks of sizes $n_1$ and $n_2$ ($n=n_1+n_2$) are denoted by 
\begin{small}
\begin{equation*}
 X=
 \begin{bmatrix}
  X_{11} & X_{12} \\
  X_{21} & X_{22} \\
 \end{bmatrix}, \qquad
 Y=
 \begin{bmatrix}
  Y_{11} & Y_{12} \\
  Y_{21} & Y_{22} \\
 \end{bmatrix}.
\end{equation*} 
\end{small}
Then 
\begin{equation}
 X_{11}^{-1}=Y_{11}-Y_{12}Y_{22}^{-1}Y_{21} \label{app:eq:schur1}
\end{equation}
\begin{equation}
 \det X = \det X_{11} \det Y_{22}^{-1} \label{app:eq:schur2}
\end{equation}
\end{prop}
\begin{proof}
It is trivial that
 \begin{equation*}
 \begin{bmatrix}
  I_{n_1}     & -Y_{12}Y_{22}^{-1} \\
  0           & I_{n_2} \\
 \end{bmatrix}
 Y
=
 \begin{bmatrix}
  Y_{11}-Y_{12}Y_{22}^{-1}Y_{21}    & 0 \\
  Y_{21}                       & Y_{22} \\
 \end{bmatrix}.
\end{equation*} 
Multiplying $X$ form right, we obtain Eq.~(\ref{app:eq:schur1}).
Eq.~(\ref{app:eq:schur2}) is derived by taking the determinant of the above identity.
\end{proof}

\begin{prop}
\label{app:detdet}
For $A \in \mat{n}{m}$ and $B \in \mat{m}{n}$,
\begin{equation}
 \det(I_n-AB)=\det(I_m-BA)
\end{equation}
\end{prop}
\begin{proof}
Take the determinant of the following identity:
\begin{equation}
 \begin{bmatrix}
 I_n-AB & 0 \\
  B     & I_n \\
 \end{bmatrix}
 \begin{bmatrix}
  I_n   & -A \\
  0     & I_m \\
 \end{bmatrix}
=
 \begin{bmatrix}
 I_n    & -A \\
  0     & I_m \\
 \end{bmatrix}
 \begin{bmatrix}
 I_n   & 0 \\
 B     & I_m-BA \\
 \end{bmatrix}
\end{equation} 
\end{proof}

\begin{prop}
\label{app:detuniform}
Let $X$ be a $d \times d$ matrix of the form
\begin{equation}
X=
\begin{small}
\renewcommand{\arraystretch}{0.8}
\begin{bmatrix}
a      & b      & \cdots  & b \\
b      & a      & \cdots  & b \\
\vdots & \vdots & \ddots  & \vdots \\
b      & b      & \cdots  & a  \\
\end{bmatrix}.
\end{small}
\end{equation}
Then $\det X= (a-b)^{d-1}(a+(d-1)b)$ and
\begin{equation}
X^{-1}=
\begin{small}
\renewcommand{\arraystretch}{0.9}
\frac{1}{(a-b)(a+(d-1)b)}
\begin{bmatrix}
a+(d-2)b& -b      & \cdots  & -b \\
-b      & a+(d-2)b& \cdots  & -b \\
\vdots  & \vdots  & \ddots  & \vdots \\
-b      & b       & \cdots  & a+(d-2)b  \\
\end{bmatrix}
\end{small}
\end{equation}
\end{prop}

\section{On probability distributions}
\label{app:sec:probability}
\begin{prop}
\label{app:prop:normccm}
Let $x, y$ be vector valued random variables following a probability distribution $p$.
If $\var{p}{(x,y)}$ is regular, then
\begin{equation}
  \norm{\corr{p}{y}{x}}_2 < 1,
\end{equation}
where $\corr{p}{y}{x}$ is the \ccm and $\norm{\cdot}_2$ is the norm induced by the inner product.
(Definitions are found in Subsection \ref{sec:regionPD}.)
\end{prop}
\begin{proof}
From the assumption of this proposition,
$\var{p}{y}$ and $\var{p}{x}$ are both regular and
the \ccm given by Eq.~(\ref{eq:defccm}) is well defined.
For arbitrary vector $a$ and $b$,
\begin{align*}
 \inp{b}{\cov{p}{y}{x}a }
&=
\E{p}{b^{T}(y - \E{}{y})(x-\E{}{x})^{T} a } \\
&<
\sqrt{ \inp{b}{\var{}{y}b} } \sqrt{ \inp{a}{\var{}{x}a} },
\end{align*}
where we used Schwartz's inequality.
This inequality must be strict for all $a$ and $b$ because of the regularity of $\var{}{(x,y)}$.
Using the above inequality, we obtain
\begin{equation*}
 \norm{\corr{p}{y}{x}}_2  = \max_{ b\neq 0, a \neq 0 } 
\frac{   \inp{b}{\cov{p}{y}{x}a } }{  \sqrt{ \inp{b}{\var{}{y}b} } \sqrt{ \inp{a}{\var{}{x}a} }  } < 1.
\end{equation*}
\end{proof}

\chapter{Miscellaneous facts on LBP}
\section{Inference on tree}
\label{app:sec:infontree}
The following formula gives the covariance of the sufficient statistics on separated vertices on a tree.
The expression involves covariances of neighboring vertices and inverted variances. 
It is interesting that this type of expressions also appears in the linearization of the LBP update in Theorem \ref{thm:diffofLBP}
and the \Bzf.
\begin{prop}
\label{app:prop:covformula}
Let $\mathcal{I}=\{\mathcal{E}_{\alpha},\mathcal{E}_i \}$ be an \ifa on a tree structured factor graph $H=(V,F)$.
Let $p$ be the probability distribution obtained by an graphical model satisfying Assumption \ref{asm:modelindludes}.
(See Subsection \ref{sec:basicLBP}.)
For vertices $i,j$ of $H$, 
let $(i=i_1,\alpha_1,i_2,\alpha_2,\ldots,\alpha_n,j=i_{n+1})$ be the unique walk from $i$ to $j$
that satisfies $\alpha_l \neq \alpha_{l+1}$.
We have
\begin{align*}
 \cov{p}{\phi_i}{\phi_j} =
\cov{p_{\alpha_n}}{\phi_{i_{n+1}}}{\phi_{i_n}} 
&\var{p_{i_n}}{\phi_{i_n}}^{-1}  \cdots \\
&\cov{p_{\alpha_2}}{\phi_{i_{3}}}{\phi_{i_2}} 
\var{p_{i_2}}{\phi_{i_2}}^{-1} 
\cov{p_{\alpha_1}}{\phi_{i_{2}}}{\phi_{i_1}},
\end{align*}
where $p_{i_l}$ and $p_{\alpha_l}$ are marginal distributions of $p$, and
$\phi_{i_l}$ are sufficient statistics of exponential families $\mathcal{E}_{i_l}$ 
\end{prop}
\begin{proof}
Consider a tree $H=(V,F)$ given by $V=\{1,2,3\}$ and $F=\{\alpha=\{1,2\},\beta=\{2,3\}\}$.
We compute the covariance of $\phi_1$ and $\phi_3$ on this graph.
Other cases are reduced to this case.
Thus what we have to show is
\begin{equation}
  \cov{}{\phi_1}{\phi_3}=  \cov{}{\phi_1}{\phi_2} \var{}{\phi_2} ^{-1} \cov{}{\phi_2}{\phi_3}. \label{app:eq:cov13formula}
\end{equation}

Let us define a subfamily of the global exponential family that include the given graphical model:
\begin{align*}
 p(x_1,x_2,x_3;\theta_1,\theta_2,\theta_3) 
=
\exp & \Big( \theta_1 \phi_1(x_1) +\pxw{\bar{\theta}}{12} \phi_{12}(x_1,x_2) + \theta_2 \phi_2(x_2)  \\
& +\pxw{\bar{\theta}}{23} \phi_{23}(x_2,x_3) + \theta_3 \phi_3(x_3) 
- \psi(\theta_1,\theta_2,\theta_3)  \Big).
\end{align*}
From the assumption,
the given distribution $p$ is equal to $p(\bar{\theta})$ for some $\bar{\theta}=(\bar{\theta_1},\bar{\theta_2},\bar{\theta_3})$.
The \eparas are denoted by $(\eta_1,\eta_2,\eta_3)$. 
The variances and covariances are computed by the derivatives of the log partition function:
\begin{equation*}
 \cov{p}{\phi_i}{\phi_j}= \pds{\psi}{\theta_i}{\theta_j}.
\end{equation*}
\begin{claim}
Let $\varphi$ be the Legendre transform of $\psi$, then
 \begin{equation}
 \pds{ \varphi }{\eta_1}{\eta_3} =0 .
\end{equation}
\end{claim}
\begin{proof}[Proof of claim]
Local exponential families $\mathcal{E}_{12}, \mathcal{E}_2$ and $\mathcal{E}_{23}$ are denoted by
\begin{align*}
& b_{12}(x_1,x_2;\vx{\theta}{12}{1},\vx{\theta}{12}{2})= 
\exp \left( \vx{\theta}{12}{1} \phi_1(x_1) + \vx{\theta}{12}{2} \phi_2(x_2) +\bar{\theta}_{12} \phi_{12}(x_1,x_2) - \psi_{12}(\vx{\theta}{12}{1},\vx{\theta}{12}{2})  \right), \\
& b_{2}(x_2;\theta'_2)= \exp \left( \theta'_2 \phi_2(x_2)  - \psi_2(\theta'_2) \right), \\
& b_{23}(x_2,x_3;\vx{\theta}{23}{2},\vx{\theta}{23}{3})= 
\exp \left( \vx{\theta}{23}{2} \phi_2(x_2) + \vx{\theta}{23}{3} \phi_3(x_3) +\bar{\theta}_{23} \phi_{23}(x_2,x_3) - \psi_{23}(\vx{\theta}{23}{2},\vx{\theta}{23}{3})  \right).
\end{align*}
The dual parameter sets are denoted by $(\vx{\eta}{12}{1},\vx{\eta}{12}{2})$, $\eta'_2$ and $(\vx{\eta}{23}{2},\vx{\eta}{23}{3})$.
As usual, we are assuming that the \ifa satisfies Assumptions \ref{asm:expregular} and \ref{asm:marginallyclosed}.
(See Subsections \ref{sec:expfamily} and \ref{sec:infmodel}.)

If we set
\begin{equation*}
 \vx{\eta}{12}{1}= \eta_1, \quad \eta'_2= \vx{\eta}{12}{2}= \vx{\eta}{23}{2} = \eta_2  \quad \text{and}\quad \vx{\eta}{23}{3}=\eta_3,
\end{equation*}
we see that $b_{12}=p_{12}$, $b_2=p_2$ and $b_{23}=p_{23}$.
On the other hand, we have
\begin{equation}
 p(x_1,x_2,x_3) = 
\frac{p_{12}(x_1,x_2) p_{23}(x_2,x_3)}{p_2 (x_2)} \label{app:eq:123treeprob}
\end{equation}
because $H$ is a tree.
From Eqs.~(\ref{app:eq:123treeprob}) and (\ref{eq:negativeentropy}),
we drive
\begin{equation*}
 \varphi(\eta_1,\eta_2,\eta_3)= \varphi_{12}(\eta_1,\eta_2)+ \varphi_{23}(\eta_2,\eta_3) - \varphi_{2}(\eta_2)
\end{equation*}
From this equation, the assertion of this claim is immediately proved.
\end{proof}
Let us go back to the proof of Proposition \ref{app:prop:covformula}.
Using standard results presented in Subsection \ref{sec:expfamily},
it is easy to see that
\begin{equation*}
  \sum_{j=1}^3   \pds{\psi}{\theta_i}{\theta_j}  \pds{\varphi}{\eta_j}{\eta_k} = \delta_{i,k}.
\end{equation*}
Setting $(i,k)=(1,3),(2,3)$, we obtain
\begin{align*}
 \cov{}{\phi_1}{\phi_2}  \pds{\varphi}{\eta_2}{\eta_3} +  \cov{}{\phi_1}{\phi_3}  \pds{\varphi}{\eta_3}{\eta_3} = 0,\\
\var{}{\phi_2}   \pds{\varphi}{\eta_2}{\eta_3}  +\cov{}{\phi_2}{\phi_3}  \pds{\varphi}{\eta_3}{\eta_3}=0.
\end{align*}
We obtain Eq.~(\ref{app:eq:cov13formula}) from these equations.
\end{proof}

\section{The Hessian of $\mathcal{F}$}
\label{sec:HesseF}
This section derives Eq.~(\ref{eq:nablaFtype2}) calculating the second derivatives of $\mathcal{F}$.
Recall that the domain of the type 2 Bethe free energy function $\mathcal{F}$ is
\begin{equation*}
 A (\mathcal{I},\Psi):= 
\{\bstheta=\{ \fa{\theta},\theta_i \} | \pa{\theta}= \pa{\bar{\theta}} , 
\sum_{\alpha \ni i} \va{\bar{\theta}}{i} =(1-d_i)\theta_i+\sum_{\alpha \ni i} \va{\theta}{i} 
~{}^{\forall} \alpha \in F,~{}^{\forall} i \in \alpha \}.
\end{equation*}
We take $\{ \va{\theta}{i} \}_{ \alpha \in F, i \in \alpha } $ as a set of free parameters and thus 
the type 2 \Bfe function Eq.~(\ref{eq:defn:type2BFE}) is
\begin{equation*}
 \mathcal{F}( \{ \va{\theta}{i} \}  )
= -\sum_{\alpha \in F} 
\psi_{\alpha} \left( \pa{\bar{\theta}}, 
\{  \va{\theta}{i} \}_{i \in \alpha}
 \right)
-\sum_{i \in V}(1-d_i) 
\psi_{i}(
\frac{1}{1-d_i} \sum_{ \alpha \ni i} \left(  \va{\bar{\theta}}{i} -  \va{\theta}{i} \right)
). 
\end{equation*}

We introduce another coordinate $\{ \va{\xi}{i} \}_{ \alpha \in F, i \in \alpha } $ by
\begin{equation}
 \va{\theta}{i} =\sum_{ \beta \in N_i \smallsetminus \alpha} \vb{\xi}{i}.
\end{equation}
The first derivatives of $\mathcal{F}$ are
\begin{align*}
 \pd{\mathcal{F}}{\vb{\xi}{i}} 
&= \sum_{ \alpha \in N_i \smallsetminus \beta} \pd{\mathcal{F}}{ \va{\theta}{i}} \\
&= \sum_{ \alpha \in N_i \smallsetminus \beta}
\left( - \va{\eta}{i} + \eta_i \right).
\end{align*}

The second derivatives are 
\begin{align*}
  \pds{\mathcal{F}}{\vx{\theta}{\gamma}{j}}{\vb{\xi}{i}} 
&= \sum_{ \alpha \in N_i \smallsetminus \beta}
\left( - \pd{ \va{\eta}{i} }{ \vx{\theta}{\gamma}{j} } + \pd{\eta_i}{\vx{\theta}{\gamma}{j}} \right) \\
&= \sum_{ \alpha \in N_i \smallsetminus \beta}
\left(
 -\delta_{\alpha, \gamma} \cov{\alpha}{\phi_j}{\phi_i} + \delta_{i,j}\var{}{\phi_i}\left(\frac{-1}{1-d_i} \right) 
\right) \\
&=
\begin{cases}
 -\cov{\gamma}{\phi_j}{\phi_i}           &\text{if } \gamma \in N_i \smallsetminus \beta , j \neq i, \\
 -\var{\gamma}{\phi_i} + \var{i}{\phi_i} &\text{if } j=i, \gamma \neq \beta, \\
  \var{i}{\phi_i}                      &\text{if } j=i, \gamma = \beta, \\
   0                                    &\text{otherwise. }
\end{cases}
\end{align*}
Note that this matrix of the second derivatives is indexed by the directed edge set $\vec{E}$
because $\vx{\theta}{\gamma}{j}$ and $\vb{\xi}{i}$ are indexed by directed edges $(\gamma \rightarrow j)$
and $(\beta \rightarrow i)$ respectively.

At an LBP fixed point, $\var{\gamma}{\phi_i} =\var{i}{\phi_i}$ holds.
Therefore, 
\begin{equation}
 \left[  \pds{\mathcal{F}}{\vx{\theta}{\gamma}{j}}{\vb{\xi}{i}}  \right]
= \diag \left( \var{}{\phi_{t(e)}}| e \in \vec{E} \right) 
\left[ I - \matmu \right],
\end{equation}
where $\bsu=\{  u^{\alpha}_{\edij} \}$ is given by Eq.~(\ref{def:u}).
The above equation is nothing but Eq.~(\ref{eq:nablaFtype2}).

\section{Convexity of the \Bfe function}
\label{app:sec:convexity}
\begin{thm}
For any inference model $\mathcal{I}$ on a factor graph $H$ with nullity $n(H)=0,1$, 
the \Bfe function $F$ is convex on $L(\mathcal{I})$.
\end{thm}
\begin{proof}
This proof is an modification of Corollary 1 in \cite{Huniquness}, where only multinomial cases are considered.
First, we prove that
$\varphi_{\alpha}(\fa{\eta})- \varphi_i(\eta_i)$ is convex for all $ i \in \alpha$.
More precisely, we prove the positive semi-definiteness of the Hessian.
From $\var{b_i}{\phi_i}=\var{b_{\alpha}}{\phi}$ and Theorem \ref{app:schur},
we have
\begin{equation}
 \pdseta{\varphi_{i}}{i}{i}= \pdseta{\varphi_{\alpha}}{i}{i} - \pds{\varphi_{\alpha}}{\eta_i}{\pa{\eta}}
\left(    \pds{\varphi_{\alpha}}{\pa{\eta}}{\pa{\eta}}       \right)^{-1} \pds{\varphi_{\alpha}}{\pa{\eta}}{i}.
\end{equation}
Accordingly, $\exists X$
\begin{equation}
 X^{T} \nabla^2 (\varphi_{\alpha}(\fa{\eta})- \varphi_i(\eta_i))X
=
\begin{bmatrix}
\pds{\varphi_{\alpha}}{\pa{\eta}}{\pa{\eta}}   & 0 \\
0 & 0
\end{bmatrix}.
\end{equation}
Therefore, the Hessian of $\varphi_{\alpha}(\fa{\eta})- \varphi_i(\eta_i)$ is positive semidefinite.

The rest of the proof is the same as Theorem 1 and Corollary 1 of \cite{Huniquness}.
\end{proof}